\documentclass[letterpaper,11pt,reqno]{amsart}

\makeatletter

\usepackage{amssymb}
\usepackage{latexsym}
\usepackage{amsbsy}
\usepackage{subcaption}
\usepackage{amsfonts}
\usepackage{float}
\usepackage{hyperref}
\usepackage{graphicx}
\usepackage{enumerate}
\usepackage{enumitem}
\usepackage{enumitem}
\usepackage{mathrsfs}
\usepackage{url}
\usepackage[dvipsnames]{xcolor}
\usepackage{subcaption}
\usepackage{tikz}

\def\marginpar#1{\ignorespaces}

\newtheorem{theorem}{Theorem}[section]
\newtheorem{lemma}[theorem]{Lemma}
\newtheorem{proposition}[theorem]{Proposition}

\newtheorem{definition}[theorem]{Definition}

\numberwithin{equation}{section}
\definecolor{ao}{rgb}{0.0, 0.5, 0.0}
\makeatother
\begin{document}
\title[Finite fuel game]{A class of stochastic games and moving free boundary problems}
\author[Xin Guo]{{Xin} Guo}
\address{Department of Industrial Engineer and Operations Research, UC Berkeley. 
} \email{xinguo@berkeley.edu}

\author[Wenpin Tang]{{Wenpin} Tang}
\address{Department of Industrial Engineer and Operations Research, Columbia University.  
} \email{wt2319@columbia.edu}

\author[Renyuan Xu]{{Renyuan} Xu}
\address{Department of Industrial \& Systems Engineering, University of Southern California. }
\email{renyuanx@usc.edu}
\begin{abstract}
  In this paper we propose and analyze a class of 
  $N$-player stochastic games that include finite fuel stochastic games as a
  special case.  We first derive sufficient conditions for the Nash
  equilibrium (NE) in the form of a verification theorem. The associated Quasi-Variational-Inequalities
  include an essential game component regarding the interactions among
  players, which may be interpreted as the analytical representation of the conditional
  optimality for NEs.  The derivation of NEs involves
  solving first a multi-dimensional free boundary problem and then a
  Skorokhod problem.
  Finally, we present an intriguing connection between these NE strategies and  {\it controlled}
  rank-dependent stochastic differential equations.
\end{abstract}

\maketitle

\section{Introduction}

\quad Recently there are renewed interests in
$N$-player non-zero-sum stochastic games, inspired by the rapid growth in the
 theory of Mean Field Games
(MFGs) led by the pioneering work of  \cite{HMC06, LL06a, LL06b, LL07}.  
In this paper, we formulate and analyze a class of
stochastic $N$-player games that originated from the classic finite
fuel problem.  There are many reasons to consider this type of games.
Firstly, the finite fuel problem \cite{BC67, 
BSW, Karatzas83} is one of the landmarks in stochastic
control theory, therefore mathematically  a game formulation is natural.  Secondly, 
 in addition to the interest of stochastic control theory \cite{AB2006,BR2006, CMR1985,SS1991,SS1989}, its simple
yet insightful solution structures have had a wide range of
applications including economics and finance \cite{AS1998,  CFR2013, DR1990,LZ2008}, operations research and management \cite{GKTY2011,KRUK2000,KWON2020}, and queuing theory \cite{KT1992}.    Thirdly,  prior success in analyzing its stochastic game counterpart has been restricted to the special case of two-player games \cite{DF2016, HM2014,HSZ2015, KL2011,  KZ2015, Mannucci2004} or without the fuel constraint \cite{dianetti2020nonzero,GX2018}.  

\quad In this paper, we 
will analyze a class of  $N$-player stochastic games that include the finite fuel stochastic game as a special case.
The class of stochastic games presented in this paper goes as
follows. There are $N$ players whose dynamics ${\pmb{X}}_t=(X^1_t,
\cdots, X^N_t)$ are governed by the following $N$-dimensional diffusion
process:
\begin{equation}
dX^i_t ={b_i}(\pmb{X}_{t-})dt+ \pmb{\sigma_i}(\pmb{X}_{t-}) d\pmb{B}_t + d\xi^{i+}_{t} - d\xi^{i-}_{t}, \quad X^i_{0-} = x^i,  \quad  (i=1, \cdots, N), 
\end{equation}
where $\pmb{B} :=(B^1, \cdots, B^N)$ is a standard $N$-dimensional
Brownian motion in a filtered probability space $(\Omega, \mathcal{F},
\{\mathcal{F}_t\}_{t \ge 0}, \mathbb{P})$, with drift $\pmb{b}:= (b_1,
\cdots, b_N)$ and covariance matrix $\pmb{\sigma}:=(\pmb{\sigma_1},
\cdots, \pmb{\sigma_N})$ satisfying appropriate regularity conditions.
Player $i$'s control  $(\xi^{i,+},\xi^{i,-})$ is of finite variation.
Each player has access to some or all of $M$ types of resources.
Players interact through their objective functions $h^i(X^1_t, \cdots,
X^N_t)$, as well as their shared resources which are the ``fuels'' of
their controls.  The accessibility of these resources to the players and
how these resources are consumed by their respective players are
governed by a matrix $\pmb{A}:=(a_{ij})_{i,j}\in
\mathbb{R}^{N \times M}$. For instance,  when $M=1$ and $\pmb{A}=[1,1,\cdots,1]^T\in
\mathbb{R}^{N\times1 }$, this game ($\pmb{C_p}$)
corresponds to the $N$-player finite fuel game where the $N$ players
share a fixed amount of the same resource.  When $M=N$ and
$\pmb{A}=\pmb{I_N}$, this is an $N$-player game ($\pmb{C_d}$) where each
player has her individual fixed amount of resource.  In general, this
matrix $\pmb{A}$ describes the network structure of the $N$-player
game.

\quad The goal for player $i$ in the game is  to
minimize
\begin{equation*}
\mathbb{E} \int_0^{\infty} e^{-\alpha t} h^i(X^1_t, \cdots, X^N_t) dt, 
\end{equation*}
over appropriate admissible game strategies, which are specified in Section
\ref{section:setup}.   Note that this $N$-player game cannot be simply analyzed with
an MFG approach as the network structure would collapse if an
aggregation approach was applied.

\quad We will analyze the NEs of this stochastic game.  We first
derive sufficient conditions for the NE policy in the form of a
verification theorem (Theorem \ref{thm:verification}), which reveals
an essential game element regarding the interactions among players.
This is the Hamilton--Jacobi--Bellman (HJB) representation of the
conditional optimality for NE in a stochastic game.  To
understand the structural properties of the NEs, we proceed further to
analyze this stochastic game in terms of the game values, the NE
strategies, and the controlled dynamics.  Mathematically, the analysis
involves solving first a multi-dimensional free boundary problem and
then a Skorokhod problem with a {\em moving} boundary.  The boundary
is ``moving'' in that it moves in response to both  changes of the
system and  controls of other players.  The analytical
solution is derived by first exploring the two special games
$\pmb{C_p}$ and $\pmb{C_d}$.  Analyzing these two types of games
provides key insights into the solution structure of the general
game. Finally, we reformulate the NE strategies in the form of
controlled rank-dependent stochastic differential equations (SDEs),
and compare game values between games $\pmb{C_p}$ and $\pmb{C_d}$.

\smallskip
\paragraph{\bf Main contributions.}

(i) In the verification theorem for $N$-player games, we obtain the form of the HJB equations for general stochastic games with
singular controls.  Unlike all previous analysis that focused on
two-player games, we show that in addition to the standard HJBs that
correspond to stochastic control problems, there is an essential term
that is unique to stochastic games.  This term represents the
interactions among players, especially the ones who are active and
those who are waiting. This critical term was hidden in two-player
stochastic games and was previously (mis)understood as a regularity
condition.

(ii) The structural difference between games and control problems is
further revealed in the explicit solution to the NEs for $N$-player
games.  
In a control problem,
a free boundary depends on the state of the system; in stochastic
games, however, the ``face'' of the boundary moves based on the action
of herself and interaction among players in the game (Figure \ref{strategy_comparison}).
Note that this  free boundary for stochastic games with an infinite time horizon {\it moves} in a different sense from the one in \cite{CMR1985} for finite time control problems where the boundary is time dependent. Rather it  moves due to changes of the system and the competition in the game.

(iii) This difference is further highlighted in the framework of
controlled rank-dependent SDEs.  To the best of our knowledges, this
is the first time a stochastic game is explicitly connected with
rank-dependent SDEs in a more general form. This new form of
rank-dependent SDEs presents a fresh class of yet-to-be studied SDEs
(Section \ref{section:discussion}).

(iv) We recast the  controlled dynamics of the game solution
 in the framework of {\it controlled} rank-dependent
SDEs.  Compared with  the well-known rank-dependent SDEs, rank-dependent SDEs
with an additional control component are new.  We
establish the existence of the solution by directly constructing a
reflected diffusion process.  (See Section \ref{section:discussion}
for further discussions.)

(v) Finally, stochastic games considered in this paper are resource
allocation games. Resource allocation problems have a wide range of
applications including { inventory management, resource allocation,} cloud computing, smart power grid control, and
multimedia wireless networks \cite{GLSSB2018, GNT2006,LNPST2003, SMSW2012}.  However, the existing literature has been
unsuccessful in analyzing the resource allocation problem in the
setting of stochastic games.  Besides the technical contributions, our
analysis provides a useful economic insight: in a stochastic game of
resource allocations, sharing has lower cost than dividing and pooling
yields the lowest cost for each player.

\smallskip
{\paragraph{\bf Related work.}

There are a number of papers on non-zero-sum two-player games with singular controls.  By treating one player as a controller and the other as a stopper, Karatzas and Li~\cite{KL2011} 
analyze the existence of an NE for the game using a BSDE approach.
Hernandez-Hernandez, Simon, and Zervos~\cite{HSZ2015} study
the smoothness of the value function and show
that the optimal strategy may not be unique when the controller enjoys a first-move advantage.
Kwon and Zhang~\cite{KZ2015} 
  investigate a game of irreversible investment with singular controls and strategic exit.  
  They characterize  a class of market perfect equilibria and  identify a set of conditions under which the outcome of the game may be unique despite the multiplicity of the equilibria. 
  De Angelis and Ferrari~\cite{DF2016}
  establish the connection between singular controls and optimal stopping times for a non-zero-sum two-player game.   Mannucci~\cite{Mannucci2004} and Hamadene and Mu~\cite{HM2014} consider the fuel follower problem in a finite-time horizon with a bounded velocity, and establish via different techniques the existence of an NE of the two-player game. Very recently,  \cite{GX2018} compare the $N$-player game versus the MFG for the fuel follower problem.
  All  these works are  without the fuel constraint
and  are essentially built on one-dimensional stochastic control problems.
Furthermore, except for \cite{GX2018}, all of these papers are
restricted to the case of $N=2$.  
To the best of our knowledge,  our work is the first to complete the
mathematical analysis on an $N$-player stochastic game based on an
original two-dimensional control problem.}


\quad In our work the controlled dynamics 
are recast in the framework of {\it controlled} rank-dependent
SDEs. Rank-dependent SDEs without controls arise in the ``Up the
River'' problem \cite{Aldous} and in stochastic portfolio
theory \cite{F}, including the well-studied {\em Atlas model} \cite{BFK, IPBKF}.

\smallskip
\paragraph{\bf Notations and organization.}

Throughout the paper, we denote vectors/matrices by bold case letters,
e.g., $\pmb{x}$ and $\pmb{X}$. The transpose of a real vector $\pmb{x}$
is denoted as $\pmb{x}^T$. For a vector $\pmb{x}$, $\|\pmb{x}\|$
denotes its $l_2$ norm. For a matrix $\pmb{X}$, $\|\pmb{X}\|$ denotes
its spectral norm.

\quad The paper is organized as follows.
Section \ref{section:setup} presents the mathematical formulation of
the $N$-player game. Section \ref{section:verification} provides
a verification theorem for sufficient conditions of the NE of the game and the existence of Skorokhod problem for NE strategies.
Section \ref{section:gamepooling} studies game $\pmb{C_p}$ and
Section \ref{section:gamedivide} studies game $\pmb{C_d}$. With the
insight from these two games, Section \ref{section:general} analyzes
the general $N$-player game $\pmb{C}$.  Section \ref{section:relation}
compares games $\pmb{C_p}$, $\pmb{C_d}$ and $\pmb{C}$, discusses the
game values and their economic implications, and unifies their
corresponding controlled dynamics in the framework of the controlled
rank-dependent SDEs.
\section{Problem Setup}
\label{section:setup}


{\bf Controlled dynamics.}
Let  $(X^i_t)_{t\ge 0} $ be the position of player $i, 1 \le i \le N$. 
In the absence of
controls,   ${\pmb{X}}_t=(X^1_t, \cdots, X^N_t)$ is governed by the stochastic differential
equation (SDE):
\begin{equation} 
\label{Eq. diffusion}
d{\pmb{X}}_t=\pmb{b}({\pmb{X}}_t)
dt+\pmb{\sigma}({\pmb{X}}_{t})   d\pmb{B}_t, \quad  {\pmb{X}}_{0-}=(x^1, \cdots, x^N),
\end{equation} 
where $\pmb{B}: =(B^1, \cdots, B^N)$ is a standard $N$-dimensional Brownian motion in a filtered probability space $(\Omega,
\mathcal{F}, \{\mathcal{F}_t\}_{t \ge 0}, \mathbb{P})$, with the drift  $\pmb{b}(\cdot):=(b_1(\cdot),\cdots,b_N(\cdot))$ and the covariance matrix $\pmb{\sigma}(\cdot):=(\sigma_{ij}(\cdot))_{1 \le i , j \le N}$. 
{
As will be explained later in Section \ref{sc33}, we consider a weak formulation of the stochastic game.
To ensure the existence and the uniqueness of the SDE, $\pmb{b}(\cdot)$ and $\pmb{\sigma}(\cdot)$ are assumed to satisfy the condition:
 

\begin{enumerate}[font=\bfseries,leftmargin=2cm]
\item[H1.]
$\pmb{b}(\cdot)$ and $\pmb{\sigma}(\cdot)$ are bounded and continuous, and 
$\pmb{\sigma}(\cdot)$ is uniformly elliptic, i.e., there exists 
$\alpha > 0$ such that $
\xi^T \pmb{\sigma}(\pmb{x}) {\pmb{\sigma}^{\top}(\pmb{x})} \xi \ge \alpha |\xi|^2$,  for all $\pmb{x} \in \mathbb{R}^N, \, \xi \in \mathbb{R}^N.
$
\end{enumerate}

Assumption {\bf H1} ensures the existence of a {\it weak solution} to \eqref{Eq. diffusion} \cite{SV2007}.
}
Here and throughout the rest of the paper,  the infinitesimal generator $\mathcal {L}$ is 
\begin{equation}
\label{eqn:generator}
\mathcal{L} : = \sum_i b_i(\pmb{x}) \frac{\partial}{\partial x^i} + \frac{1}{2} \sum_{i,j} (\pmb{\sigma}(\pmb{x}) \pmb{\sigma}(\pmb{x})^T)_{i,j} \frac{\partial^2}{\partial x^i \partial x^j},
\end{equation}
where $\pmb{\sigma}(\pmb{x}) \pmb{\sigma}(\pmb{x})^T$ is assumed to be positive-definite for every $\pmb{x}\in \mathbb{R}^N$.

\quad If a control is applied to $X^i_t$,  then $X^i_t$ evolves as
\begin{equation}
\label{Eq:Xi}
dX^i_t = b_i(\pmb{X}_{t-})dt+ \pmb{\sigma}_i(\pmb{X}_{t-})  d\pmb{B}_t + d\xi^{i+}_{t} - d\xi^{i-}_{t}, \quad X^i_{0-} = x^i,
\end{equation}
where $\pmb{\sigma}_i$ is the $i^{th}$ row of the  covariance matrix $\pmb{\sigma}$.
Here the control  $(\xi^{i+}, \xi^{i-})$ is  a pair of  non-decreasing and c\`adl\`ag processes. In other words, $(\xi^{i+}, \xi^{i-})$ is the minimum decomposition of the finite variation process $\xi^i$ such that $\xi^i:=\xi^{i+} - \xi^{i-}$.

\vskip 3 pt

{\bf Game objective.}
The game is for player $i$  to minimize, for all $(\xi^{i+}, \xi^{i-})$ in  an appropriate admissible control set, over an infinite time horizon, the following objective function,
\begin{equation}
\label{eqn:objective}
\mathbb{E} \int_0^{\infty} e^{-\alpha t} h^i(X^1_t, \cdots, X^N_t) dt.
\end{equation}
Here $\alpha > 0$ is a constant discount factor.
In this game, players interact through their respective objective functions $h^i(\pmb{x}): \mathbb{R}^N \rightarrow \mathbb{R}^+$.

\begin{enumerate}[font=\bfseries,leftmargin=2cm]
\item[H2.] Each $h^i (\pmb{x})$  is  
twice differentiable, with  $k \le ||\nabla^2 h^{i }(\pmb{x})|| \le K$ for some $K > k >0$. 
\end{enumerate}
For example, $h^i(\pmb{x})=h(x^i-\frac{\sum_{j=1}^N x^j}{N})$ with $h(\cdot) \ge 0$ is a distance function between the position of player $i$ and the center of all players.  

\quad Note that in the objective function \eqref{eqn:objective}, there is no cost of control. 
With this formulation, the explicit solution structure of the NE for game \eqref{eqn:objective} is neat and insightful.
It is entirely possible to consider an N-player game with additional cost of control. For instance, one might study  the game formulation of  \cite{Karatzas83} with a proportional cost of control. We conjecture that the solution structure would be similar although the analysis will be more involved. This will be an interesting problem for future analysis. 

\vskip 3 pt
{\bf Admissible control policies.}
{Denote $\check{\xi}^i_t$
  as the cumulative amount of controls/resources consumed by player $i$ up to  time $t$. When ${\xi}^i_t$ is of finite variation, then there is a unique decomposition such that  ${\xi}^i_t:= \xi^{i+}_t- \xi^{i-}_t$, hence $\check{\xi}^i_t:= \xi^{i+}_t+ \xi^{i-}_t$. Here $\xi^{i+}$ and $\xi^{i-}$ are  non-decreasing and c\`adl\`ag processes which can be further decomposed in a differential form,
\begin{eqnarray}\label{eqn: differential_decomposition}
d \xi^{i\pm}_t = d (\xi^{i\pm}_t)^c + \Delta \xi^{i\pm}_t,
\end{eqnarray}
where $d (\xi^{i\pm}_t)^c $ is the continuous component and $\Delta \xi^{i\pm}_t:= \xi^{i\pm}_t -  \xi^{i\pm}_{t-}$ is the jump component of $d \xi^{i\pm}_t$. {Equivalently, we can write $\xi^{i\pm}_t = (\xi^{i\pm}_t)^c+ \sum_{s \leq t}\Delta \xi^{i\pm}_s.$
}

\quad Meanwhile, we consider a {\em weak formulation} of the stochastic game. (See  \cite[Chapter 2, Section 4.2]{YZ99} and  \cite[Section 5]{GR06}
for more discussions on weak formulations of  stochastic control problems).
That is, 
$(\pmb{B}_t, \, t \ge 0)$ is an $N$-dimensional Brownian motion with some filtration $(\mathcal{F}_t, \, t \ge 0)$, 
and the admissible control set $\mathcal{S}_N(\pmb{x},\pmb{y})$ for the $N$-player game  is 
\begin{eqnarray}
\label{eq:admset}
\begin{aligned}
\mathcal{S}_N(\pmb{x},\pmb{y}) : = \Bigg\{\pmb{\xi}\,:&\, \,\xi^i \in \mathcal{U}_N^{i} \mbox{ for } 1 \le i \le N, \,\,   \sum_{i = 1}^N \int_{0}^{\infty} \frac{a_{ij}Y_{t-}^j}{\sum_{k=1}^M a_{ik}Y_{t-}^k} d\check{\xi}^i_t \le y^j, 1 \le j \le M, \,\,  \\
&\,\,\mathbb{P}\left(\Delta \xi^i_t  \Delta \xi^k_t  \ne  0 \right) = 0 \mbox{ for all } t \ge 0 \mbox{ and } i \ne k \Bigg\},
\end{aligned}
\end{eqnarray}
where 
\begin{eqnarray*}
\mathcal{U}^{i}_N:= \Big\{ (\xi^{+}, \xi^{-}): \xi^{+} \mbox{ and } \xi^{-} \mbox{ are } { \mathcal{F}_t} \mbox{-progressively measurable, c\`adl\`ag, non-decreasing}, \\
\mathbb{E}\left[\int_{0}^{\infty}e^{-\alpha t}d\xi^{\pm}_t\right]<\infty \mbox{ and } \xi^{+}_{0-} = \xi^{-}_{0-} = 0\Big\},
\end{eqnarray*}
and 
\begin{equation}
\label{eq:XM}
Y_t^j = y^j - \sum_{i = 1}^N \int_{0}^t \frac{a_{ij}Y_{s-}^j}{\sum_{k=1}^M a_{ik}Y_{s-}^k}d \check{\xi}^i_s  \in \mathbb{R}_{+} \quad \mbox{and} \quad Y^j_{0-} = y^j,
\end{equation}
with $a_{ij}=0$ or $1$ for $1\le i\le N$ and $1\le j\le M$, $\sum_{j=1}^M a_{ij} >0$ for all $i=1, \cdots, N$, and $\sum_{i=1}^N a_{ij} >0$ for all $j = 1, \cdots, M$. 

}

\quad Here is the intuition for the admissible control set $\mathcal{S}_N(\pmb{x},\pmb{y})$. 
In this game, each player $i$ will make  decisions based on 
 the current positions of all players and the available resources. 
 In addition to  this adaptedness constraint,   the admissible control set $\mathcal{S}_N(\pmb{x},\pmb{y})$ specifies the resource allocation policy for each player.  
For $M$ different types of resources,  define $\pmb{A}:=(a_{ij})_{i,j}\in \mathbb{R}^{N \times M}$  to be the \textit{adjacent matrix}
with $a_{ij}= 0 $ or $1$. Then  $\pmb{A}$ describes the relationship between the players and the types of available resources, with $a_{ij}=1$ meaning that resource of type $j$  is available to  player $i$, and $a_{ij}=0$ meaning that resource of type $j$ is inaccessible to player $i$.  
The condition $\sum_{j=1}^M a_{ij} >0$ for all $i=1, \cdots , N$ implies that each player $i$ has access to at least one resource,
and the condition $\sum_{i=1}^N a_{ij} >0$ for all $j = 1, \cdots, M$ indicates that each resource $j$ is available to at least one player.
When player $i$ would like to exercise control, she will consume resources proportionally to all the resources available to her. She will stop consuming once all the available resources hit level zero.
This results in the form of the integrand in the expression of \eqref{eq:XM}. Note that the denominator is always no smaller than the numerator hence the integrand is well-defined with the convention $\frac{0}{0}=0$. 

\quad Take an example of $N=4$, $M=6$, with the matrix $\pmb{A}$ defined as in Figure \ref{allocation_constraint}.
The resource allocation policy is illustrated in Figure \ref{fig:allocation_policy}, with the amount of available resource  $y^1$ and $y^2$ of type one and two respectively. When player one wishes to apply controls of amount $\Delta$, say $\Delta\leq y^1+y^2$,  she will consume resources randomly from type one and two. 
So player one will take $\Delta \frac{y^1}{y^1+y^2}$ from resource one and $\Delta \frac{y^2}{y^1+y^2}$ from resource two. Finally,  the condition $\mathbb{P}(\Delta \xi^i_t \Delta \xi^k_t \ne  0) = 0$ for all $t \ge 0$ and $i \ne k$ excludes the possibility of simultaneous jumps of any two out of $N$ players, which facilitates designing feasible control policies when controls involve jumps. 
This condition is not a  restriction,  and instead should be interpreted as a {\it regularization}. See also \cite{BCG2019,GX2018,KZ2015}. 
When
there are multiple players who would like to jump at the same time, one can simply design a proper {\it order}, for instance by indexing the players and their jump orders,  so that they will move {\it sequentially}. 
\begin{figure}[H] 
\centering
\begin{subfigure}{.40\textwidth}
\[\pmb{A} =
\begin{bmatrix}
    1,1,0,0,0,0 \\
   0,0,1,0,1,0 \\
   0,0,0,0,0,1\\
  0,0,0,1,0,0
\end{bmatrix}\]
\end{subfigure}
\begin{subfigure}{.30\textwidth}
  \includegraphics[width=1.0\linewidth]{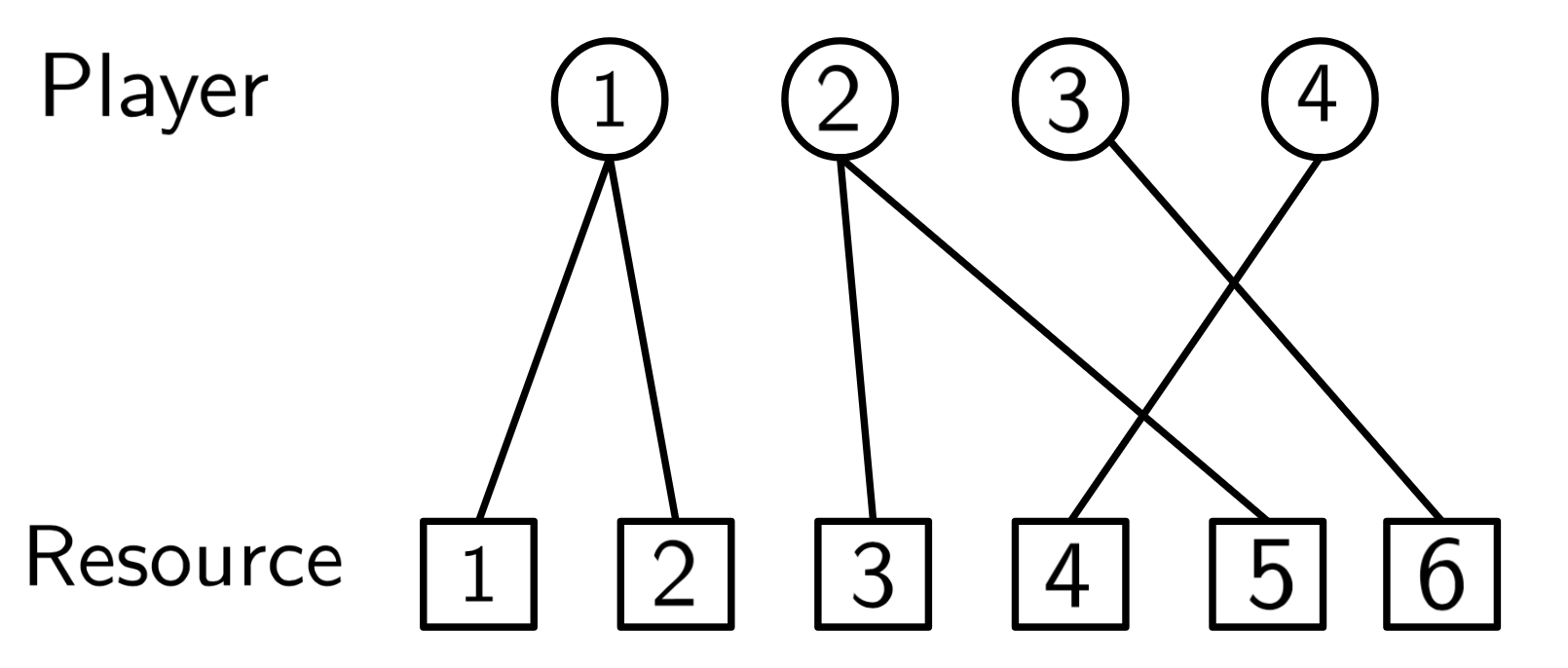}
  \subcaption{Relationship.}\label{fig:allocation_constraint}
\end{subfigure}
\begin{subfigure}{.25\textwidth}
  \includegraphics[width=1.0\linewidth]{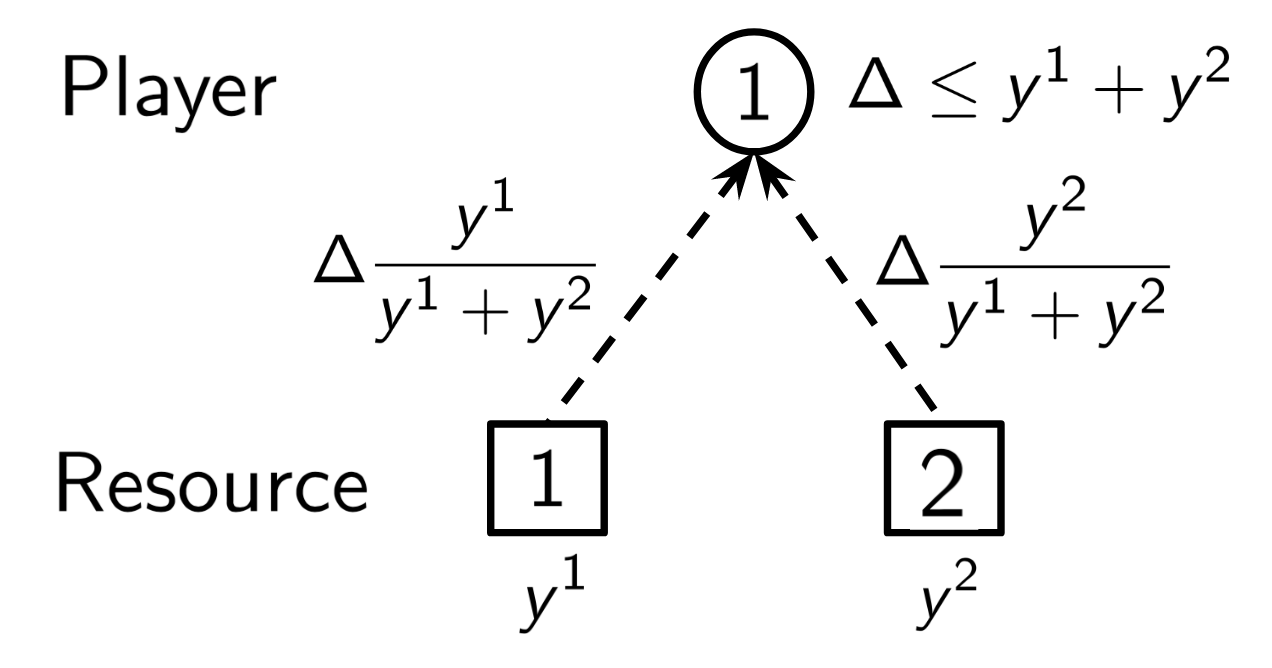}
  \subcaption{Resource allocation policy.}\label{fig:allocation_policy}
\end{subfigure}
\caption{Example of adjacent matrix $\pmb{A}$, relationship between the players and resources when $N=4$ and $M=6$.}
\label{allocation_constraint}
\end{figure}



{\bf Game formulation {and game criterion}.}
Let $\pmb{\xi}:=(\xi^1, \cdots, \xi^N)$ be the controls from the players. Let $\pmb{x}:= (x^1, \cdots, x^N)$ and  $\pmb{y}:=(y^1,\cdots,y^M)$.  
Then the  stochastic game is for each player $i$ to  minimize 
\begin{equation}
\label{eq:JA_game}
J^i(\pmb{x}, \pmb{y}; \pmb{\xi}):= \mathbb{E} \int_0^{\infty} e^{-\alpha t} h^i(\pmb{X}_t) dt,
\end{equation}
subject to the dynamics in \eqref{Eq:Xi} and \eqref{eq:XM} with the constraint in  \eqref{eq:admset}.
There are two special games of particular interest.
One is a game where all players pool their resources such that 
\begin{equation}
\label{pool}
{\sum_{i = 1}^N  \check{\xi}^i_{\infty}\leq y<\infty.}
 \end{equation}  
 
When $N=1$, this is a single player game corresponding to the finite fuel control problem which is well studied in \cite{BSW, Karatzas83}. 
We call this game a pooling game $\pmb{C_p}$. Clearly in terms of the adjacent matrix $\pmb{A}$, this corresponds to
$M=1$, and $\pmb{A}=[1,1,\cdots,1]^T\in \mathbb{R}^{N\times1 }$.
Another is a game where players divide the resource up front such that  
\begin{equation}
\label{dividing}
{\check{\xi}^i_{\infty}\leq y^i,}
\end{equation}
where $y^i$ is  the total amount of controls that player $i$ can exercise.
This game is called $\pmb{C_d}$, 
with $M=N$, and $\pmb{A}=\pmb{I_N}$.
Finally, we refer the game with a general matrix $\pmb{A}$ as game  $\pmb{C}$.

\quad We will analyze the $N$-player game under the criterion of  NE. 
Recall the definition of NE of $N$-player games. 
\begin{definition}\label{def:nash}
A tuple of admissible controls $\pmb{\xi}^{*}:=(\xi^{1*}, \cdots \xi^{N*})$ is a  NE of the $N$-player game (\ref{eq:JA_game}), if for each $\xi^i\in \mathcal{U}_N^i$ such that $(\pmb{\xi}^{-i*}, \xi^i)\in \mathcal{S}_N(\pmb{x},\pmb{y})$,
\begin{equation*}
J^i \left(\pmb{x},\pmb{y}; \pmb{\xi}^{*} \right) \le J^i \left(\pmb{x},\pmb{y}; \left(\pmb{\xi}^{-i*}, \xi^i \right) \right),
\end{equation*}
where $\pmb{\xi}^{-i*}=(\xi^{1*},\cdots,\xi^{i-1*},\xi^{i+1*},\cdots,\xi^{N*})$ and $(\pmb{\xi}^{-i*},\xi^{i})=(\xi^{1*},\cdots,\xi^{i-1*},\xi^{i},\xi^{i+1*},\cdots,\xi^{N*})$.
Controls that give  NEs are called the Nash Equilibrium Points (NEPs). The associated value function $J^i \left(\pmb{x},\pmb{y}; \pmb{\xi}^{*} \right)$ $(i=1,2,\cdots,N)$ is called the game value {for player $i$}.
\end{definition}

\section{ NE Game Solution: Verification Theorem and Skorokhod Problem}
\label{section:verification}

\quad In this section, we present general strategies to get the NE solution.
First we derive heuristically the quasi-variational inequalities (QVIs) for the value function (Section \ref{sc31}), which is then used for deriving sufficient conditions of an NEP via a verification theorem (Section \ref{sc32}).
We emphasize that both the QVIs  in Section \ref{sc31} and the verification theorem in Section \ref{sc32} hold for general diffusion processes given in \eqref{Eq:Xi}.
For explicitness, we assume further that 
\begin{enumerate}[font=\bfseries,leftmargin=2cm]
\item[H1$'$.] $b_i=0,\qquad i=1,2,\cdots,N,  \qquad \text{and} \qquad\pmb{\sigma}=\pmb{I}_N.$
\end{enumerate}
Moreover, we assume that $h^i(\pmb{x}):= h\left({x}^i-\frac{\sum_{j=1}^N x^j}{N}\right),$ such that
\begin{enumerate}[font=\bfseries,leftmargin=2cm]
\item[H2$'$.] $h$ is symmetric, $h(0) \ge 0$, $h^{\prime \prime}$ is non-increasing  on $\mathbb{R}_{+}$ and $k \leq h^{\prime \prime} \leq K$ for some $0 <k <K$.
\end{enumerate}
 These additional conditions are only used to facilitate the construction of  the  NEP, as well as solving the corresponding Skorokhod problem  presented in Section \ref{sc33}. {

One basic example for $h$ under assumption {\bf H2$'$} is a quadratic function  $h(x) = ax^2 + b$ with $a\in [k,K]$ and $b\ge 0$. Our assumption also holds for a more general class of functions. Take $h$ an even function such that $h'' = f$, where $f$ is an even function, non-increasing on $\mathbb{R}_{+}$ and bounded between $k$ and $K$.
There are many such functions $f$.
A particular example is $f = a$ (constant), which will give $h(x) = ax^2 + b$ (quadratic function). Another possible example is $f = b + c\,\exp (-d x^2)$ with $c>0$ and $d>0$. In the original finite fuel problem \cite{BSW}, the authors treated the quadratic cost $h(x) = x^2$. 
Later Karatzas \cite{Karatzas83} noticed that the results can be extended to any cost function which satisfies Assumption {\bf H2$'$}.  }

{\subsection{Quasi-variational Inequalities}
\label{sc31}
We first derive heuristically the associated QVIs   {of game value under the notion of NE (see Definition \ref{def:nash})} for game (\ref{eq:JA_game}). {The key idea is to utilize the conditional optimality condition introduced in Definition \ref{def:nash}. Namely, player $i$ solves a single agent optimal control problem with optimal solution $\xi^{i*}$ when other agents are applying $\pmb{\xi}^{-i*}$.}  {To start, we define the following  partition of $\mathbb{R}^N\times\mathbb{R}^M_{+}$. Denote $\mathcal{A}_i \subseteq \mathbb{R}^{N} \times \mathbb{R}^M_{+}$ as the $i^{th}$ player's action region and $\mathcal{W}_i: = (\mathbb{R}^{N} \times \mathbb{R}^M_{+}) \setminus \mathcal{A}_i$
as her waiting region.  Let $\mathcal{A}^{-i}:=\cup_{j \ne i} \mathcal{A}_j$ and $\mathcal{W}_{-i}:=\cap_{j \ne i} \mathcal{W}_j$.
Then players' actions are as follows: player $i$ controls if and only if the process $(\pmb{X}_t,\pmb{Y}_t)$ enters  $\mathcal{A}_i$. 
This partition is {usually defined through the quasi-variational inequalities and is} also part of the solution to be derived.
Next,} define the intervene operator $\Gamma$ as
\begin{eqnarray}\label{gamma}
\Gamma_j v^i(\pmb{x},\pmb{y})=\sum_{k=1}^M \frac{a_{jk}y^k}{\sum_{s=1}^M a_{js}y^s}v^i_{y^k}(\pmb{x},\pmb{y}),
\end{eqnarray}
for $(\pmb{x},\pmb{y})\in \mathbb{R}^N\times \mathbb{R}^M_{+}$ and $i, j=1,2,\cdots, N$. Here $v^i_{y^k}:=\frac{\partial v^i}{\partial y^k}$ ($i=1,2,\cdots,N$ and $k=1,2,\cdots,M$). 
Suppose player $j$ takes a possibly suboptimal action $\Delta \xi^{j,+}>0$, then by the resource allocation policy \eqref{eq:XM}, for player $i$,
\begin{eqnarray}\label{ine_2}
v^i(\pmb{x},\pmb{y}) \leq v^i\left(\pmb{x}^{-j}, x^j + {\Delta} \xi^{j,+} , \pmb{y}- \left(\frac{a_{j1 }y^1}{\sum_{k=1}^M a_{jk} y^k},\cdots,\frac{a_{jM }y^M}{\sum_{k=1}^M a_{jk} y^k}\right)\Delta \xi^{j,+} \right).
\end{eqnarray}
By letting ${\Delta} \xi^{j,+}\rightarrow 0$, we have 
 \begin{eqnarray}\label{ine_3}
 0 \leq -\Gamma_j v^i (\pmb{x},\pmb{y})+v^i_{x^j}(\pmb{x},\pmb{y}). 
\end{eqnarray}

Next, we provide the heuristics  for deriving the QVIs. 
Let $\Delta \xi^i:=\Delta \xi^i(\pmb{x},\pmb{y})$ be the control of player $i$ with joint state position $(\pmb{x},\pmb{y})$.
When $ (\pmb{x},\pmb{y}) \in \mathcal{W}_{-i}$, we have $\Delta \xi^{j} = 0$ for $j \neq i$.  
Thus the game for player $i$ becomes a classical control problem with three choices: $\Delta \xi^{i}=0$, $\Delta \xi^{i,+}>0$, and $\Delta \xi^{i,-}>0$. The first case $\Delta \xi^{i}=0$ implies, by simple stochastic calculus,   $-\alpha v^i +h^i \left(\pmb{x} \right) + { \mathcal{L}v^i }\ge 0 $.
By a similar argument as in \eqref{ine_3}, the second case $\Delta \xi^{i,+}> 0$  corresponds to $-\Gamma_i v^i+v^i_{x^i}\ge 0$ and 
the third case $\Delta \xi^{i,-}> 0$ corresponds to $-\Gamma_i v^i-v^i_{x^i}\ge 0$. 
Since one of the three choices will be optimal, one of the inequalities will be an equation. That is,  for $ (\pmb{x},\pmb{y}) \in \mathcal{W}_{-i}$, 
\begin{equation}
\label{eq:HJBA}
\min \left\{-\alpha v^i +h^i \left(\pmb{x} \right) + { \mathcal{L}v^i }, -\Gamma_i v^i+v^i_{x^i}, -\Gamma_i v^i-v^i_{x^i}\right\} = 0.
\end{equation}

\quad When $(\pmb{x},\pmb{y})\in \mathcal{A}_j$, 
 player $j$ will control with the amount of control being $(\Delta \xi^{j,+},\Delta \xi^{j,-})\neq 0$. Therefore,
\begin{eqnarray}
&&v^j(\pmb{x},\pmb{y}) \leq v^j\left(\pmb{x}^{-j}, x^j + {\Delta} \xi^{j,+}  , \pmb{y}- \left(\frac{a_{j1 }y^1}{\sum_{k=1}^M a_{jk} y^k},\cdots,\frac{a_{jM }y^M}{\sum_{k=1}^M a_{jk} y^k}\right)\Delta \xi^{j,+}\right),\label{ln1} \\
&&v^j(\pmb{x},\pmb{y}) \leq v^j\left(\pmb{x}^{-j}, x^j  - {\Delta} \xi^{j,-}, \pmb{y}- \left(\frac{a_{j1 }y^1}{\sum_{k=1}^M a_{jk} y^k},\cdots,\frac{a_{jM }y^M}{\sum_{k=1}^M a_{jk} y^k}\right) \Delta \xi^{j,-}\right),\label{ln2}
\end{eqnarray}
and one of the inequalities in \eqref{ln1}-\eqref{ln2} will be an equality.
This leads to the following condition
\begin{equation}
\label{eq:HJB-jk}
\min \left\{ -\Gamma_j v^j+v^j_{x^j}, -\Gamma_j v^j-v^j_{x^j}\right\} = 0.
\end{equation}
For player $i\neq j$, we should have
{$v^i(\pmb{x},\pmb{y}) = v^i\left(\pmb{x}^{-j}, x^j + {\Delta} \xi^{j,+}, \pmb{y}- \left(\frac{a_{j1 }y^1}{\sum_{k=1}^M a_{jk} y^k},\cdots,\frac{a_{jM }y^M}{\sum_{k=1}^M a_{jk} y^k}\right)\Delta \xi^{j,+} \right)$ when $\Delta \xi^{j,+} >0$ is optimal for player $j$, and \\
$v^i(\pmb{x},\pmb{y}) = v^i\left(\pmb{x}^{-j}, x^j - {\Delta} \xi^{j,-}, \pmb{y}- \left(\frac{a_{j1 }y^1}{\sum_{k=1}^M a_{jk} y^k},\cdots,\frac{a_{jM }y^M}{\sum_{k=1}^M a_{jk} y^k}\right) \Delta \xi^{j,-}\right)$
when $\Delta \xi^{j,-} >0$  is optimal for player $j$. 
This holds  {due to the ``no simultaneous jump'' condition \eqref{eq:admset}}. Intuitively, this implies that player $i$ has no incentive to  {jump} when player $j$ {jumps}. Thus, 
\begin{eqnarray}
\label{eq:HJB-j}
\begin{cases}
 -\Gamma_j v^i + v^i_{x^j}&=0, \,\,\mbox{on} \,\, \{\left.(\pmb{x},\pmb{y})\in \mathbb{R}^N \times  \mathbb{R}^M_{+} \,\,\right\vert\,\,-\Gamma_j v^j  + v^j_{x^j}=0\},\\
  -\Gamma_j v^i - v^i_{x^j}&=0, \,\,\mbox{on} \,\, \{\left.(\pmb{x},\pmb{y})\in \mathbb{R}^N \times  \mathbb{R}^M_{+} \,\,\right\vert\,\,-\Gamma_j v^j  - v^j_{x^j}=0\}.
 \end{cases}
 \end{eqnarray}}
 
 Note that by letting $\Delta \xi^{i,\pm}\rightarrow 0$, equations \eqref{eq:HJBA},\eqref{eq:HJB-jk} and \eqref{eq:HJB-j} describe the behavior in $\overline{\mathcal{W}}_{i}$ and near boundary $\partial{\mathcal{W}}_{i}$. Moreover, we can show that \eqref{eq:HJBA},\eqref{eq:HJB-jk} and \eqref{eq:HJB-j} are consistent with the jump behaviors in $\mathcal{A}_i$.
 To see this, $-\sum_{j=1}^M \frac{a_{ij}y^j}{\sum_{k=1}^M a_{ik}y^k}v^i_{y^j}\pm v^i_{x^i}=0$ has a linear solution $v^i(\pmb{x},\pmb{y})=a\left(\pm x_i + \sum_{j=1}^M a_{ij}y^j\right)+b$ for some $a,b\in \mathbb{R}$. And it is easy to check that if $\sum_{k=1}^{M}a_{ik}y^k \geq \Delta>0$,
\[
\frac{a_{ij}y^j-\frac{a_{ij}y^j}{\sum_{k=1}^{M}a_{ik}y^k}\Delta}{\sum_{k=1}^{M}a_{ik}y^k-\Delta} = \frac{a_{ij}y^j}{\sum_{k=1}^{M}a_{ik}y^k},
\]
which means that the allocation policy  (jump direction) outside the waiting region is linear. Hence the {the non-infinitesimal jump also} satisfies the HJB equation \eqref{eq:HJBA} in $\mathcal{A}_i$. The consistency property also holds for \eqref{eq:HJB-j}.
In summary, we have the following QVIs:
{
\begin{subequations}
\label{QVI}
    \begin{align}
&\min \left\{-\alpha v^i +h^i \left(\pmb{x} \right) + { \mathcal{L}v^i }, -\Gamma_iv^i+v^i_{x^i}, -\Gamma_iv^i-v^i_{x^i}\right\} = 0,\nonumber\\
&\hspace{120pt} \mbox{on} \cap_{j \neq i}\left\{\left\{ -\Gamma_j v^j+v^j_{x^j}>0\right\}\cap \left\{ -\Gamma_j v^j-v^j_{x^j}>0\right\}\right\},\label{QVI-1} \\
& -\Gamma_j v^i + v^i_{x^j}=0, \hspace{33pt}\,\,\mbox{on} \,\, \{-\Gamma_j v^j + v^j_{x^j}=0\},\label{QVI-3}\\
& -\Gamma_j v^i - v^i_{x^j}=0, \hspace{33pt}\,\,\mbox{on} \,\, \{-\Gamma_j v^j - v^j_{x^j}=0\}.\label{QVI-2}
\end{align}
\end{subequations}}

The above conditions are consistent with the conditional optimality of NE for each player and describe interactions between the player in control and those who are not; these conditions ensure that all players control optimally and push sequentially the underlying dynamics until reaching the common waiting region.
}


\subsection{Verification Theorem}
\label{sc32}
Next we present a verification theorem which gives sufficient conditions of an NEP.
Given functions $v^i$ (with sufficient regularity),
we define the action and waiting regions ($\mathcal{A}_i$ and $\mathcal{W}_i$) in terms of $v^i$ $(i=1,2,\cdots,N)$ as the following:
\begin{eqnarray}\label{new_region}
\mathcal{A}_i=\mathcal{A}_i^{+}\cup \mathcal{A}_i^{-},
\end{eqnarray}
where
$\mathcal{A}_i^{+}:=\{(\pmb{x},\pmb{y})\in\mathbb{R}^{N} \times \mathbb{R}^M_{+}\,\,\vert\,\, -\Gamma_iv^i-v_{x^i}^i=0\}$ and $\mathcal{A}_i^{-}:=\{(\pmb{x},\pmb{y})\in\mathbb{R}^{N} \times \mathbb{R}^M_{+}\,\,\vert\,\, -\Gamma_iv^i+v_{x^i}^i=0\}$. Moreover, $\mathcal{W}_i=(\mathbb{R}^{N} \times \mathbb{R}^M_{+}) \setminus \mathcal{A}_i$ and $\mathcal{W}_{-i}=\cap_{j \ne i} \mathcal{W}_j$.

\begin{theorem}[Verification theorem]
\label{thm:verification} Assume \textbf{H1}-\textbf{H2} hold 
and further assume $\mathcal{A}_j \cap \mathcal{A}_i = \emptyset$ for all $i \neq j$ where {$\mathcal{A}_i, \mathcal{W}_i$ and  $\mathcal{W}_{-i}$ are defined according to \eqref{new_region}.}
For each $i = 1, \cdots, N$, suppose that the $i^{th}$ player's strategy $\xi^{i*} \in \mathcal{U}^i_N$ satisfies the following conditions
\begin{enumerate}[label=(\roman*), itemsep=3pt]
\item
$\pmb{\xi}^{*} := (\xi^{1*}, \cdots, \xi^{N*}) \in \mathcal{S}_N(\pmb{x},\pmb{y})$.
\item 
{$v^i(\cdot)$ satisfies the QVIs \eqref{QVI}}.
\item For any $\xi^i \in \mathcal{U}^i_N$ such that $(\pmb{\xi}^{-i*}, \xi^i) \in \mathcal{S}_N(\pmb{x},\pmb{y})$, 
$
\mathbb{P}((\pmb{X}^{-i*}_t, X^i_t, \pmb{Y}_t) \in \overline{\mathcal{W}_{-i}}) = 1 \quad  \mbox{for all } t \ge 0,
$
where $(\pmb{X}^{-i*}_t, X^i_t, \pmb{Y}_t)$ is under $(\pmb{\xi}^{-i*}, \xi^i)$.
\item
$v^i(\pmb{x},\pmb{y}) \in \mathcal{C}^2 (\overline{\mathcal{W}_{-i}} )$ and {$v^i$ is convex } for all $(\pmb{x}, \pmb{y} ) \in \overline{\mathcal{W}_{-i}}$,
\item {$\mathbb{E}\left[\int_0^T e^{-2\alpha t}\left(v^i_{x^j}(\pmb{X}^{-i*}_t, X^i_t, \pmb{Y}_t)\right)^2dt\right]<\infty$ for all $T>0$ where $(\pmb{X}^{-i*}_t, X^i_t, \pmb{Y}_t)$ is under $(\pmb{\xi}^{-i*}, \xi^i)  \in \mathcal{S}_N(\pmb{x},\pmb{y})$ such that (iii) holds.
}

\item For any $(\pmb{X}^{-i*}_t, X^i_t, \pmb{Y}_t)$  under $(\pmb{\xi}^{-i*}, \xi^i)  \in \mathcal{S}_N(\pmb{x},\pmb{y})$ such that (iii) holds,
$v^{i}(\pmb{x},\pmb{y})$ satisfies the transversality condition
\begin{equation}
\label{eq:trans}
\underset{T \rightarrow \infty}{\lim \sup} \, e^{-\alpha T} \mathbb{E}\left[v^i\left(\pmb{X}^{-i*}_t, X^i_t, \pmb{Y}_t \right)\right] = 0.
\end{equation}
\item {For $j \neq i$, $t \ge 0$, and $(\pmb{X}^{-i*}_t, X^i_t, \pmb{Y}_t)$ under $(\pmb{\xi}^{-i*}, \xi^i)$,}
\begin{eqnarray}\label{local_time_condition}
\check{\xi}_t^{j*} = \int_{[0,t]} 1_{\{(\pmb{X}^{-i*}_{s-}, X^i_{s-}, \pmb{Y}_{s-}) \in \mathcal{A}_{j}\}}d\check{\xi}_s^{j*},
\end{eqnarray}
  {and in addition, for $(\pmb{X}^{*}_t, \pmb{Y}^{*}_t)$ under $\pmb{\xi}^*$,
  \begin{eqnarray}\label{local_time_condition2}
\check{\xi}_t^{i*} = \int_{[0,t]} 1_{\{(\pmb{X}^{-i*}_{s-}, \pmb{Y}^*_{s-}) \in \mathcal{A}_{i}\}}d\check{\xi}_s^{i*}.
\end{eqnarray}}
\end{enumerate}
{Then $\pmb{\xi}^{*}$ is an NEP with value function $v^i$ a solution to \eqref{QVI}. That is, \begin{equation*}
v^i(\pmb{x},\pmb{y}) \le J^i(\pmb{x},\pmb{y};(\pmb{\xi}^{-i*}, \xi^i)),
\end{equation*}
for all $\xi \in \mathcal{U}_N^i$ such that $(\pmb{\xi}^{-i*},\xi^i)\in \mathcal{S}_N$, and $v^i(\pmb{x},\pmb{y} ; \pmb{\xi}^{*}) = J^i(\pmb{x},\pmb{y};(\pmb{\xi}^{-i*}, \xi^{i*}))$.} 
\end{theorem} 

\begin{proof}
It suffices to prove that for all $(\pmb{\xi}^{-i*}, \xi^i) \in \mathcal{S}_N(\pmb{x},\pmb{y})$, and for each $i = 1, \cdots, N$,
\begin{equation*}
J^i(\pmb{x},\pmb{y}; \pmb{\xi}^{*}) \le J^i(\pmb{x},\pmb{y};(\pmb{\xi}^{-i*}, \xi^i)).
\end{equation*}

\quad Recall \eqref{Eq. diffusion} and \eqref{eq:XM}.
From condition $(iii)$, under control $(\pmb{\xi}^{-i*}, \xi^i) \in \mathcal{S}_N(\pmb{x},\pmb{y})$, $(\pmb{X}^{-i*}_t, X^i_t, \pmb{Y}_t) \in \overline{\mathcal{W}_{-i}}$ a.s..  
Applying It\^{o}-Meyer's formula \cite[Theorem 21]{Meyer76} to $e^{-\alpha t} v^i(\pmb{X}_t^{-i*}, X_t^i, \pmb{Y}_t)$ yields
\begin{align*}
\quad &\mathbb{E}[e^{-\alpha T} v^i(\pmb{X}_T^{-i*}, X_T^i, \pmb{Y}_T)] - v^i(\pmb{x}, \pmb{y}) \\
       = ~& \mathbb{E} \int_0^T e^{-\alpha t} \left({ \mathcal{L} v^i} - \alpha v^i \right) dt + \mathbb{E}\int_0^T e^{-\alpha t} \sum_{j = 1}^N v^i_{x^j} dB_t^j  + \sum_{j=1,j\neq i}^N \mathbb{E} \int_{[0,T)} e^{-\alpha t} (v^i_{x^j} d\xi_t^{j*,+} - v^i_{x^j} d\xi_t^{j*,-})\\
            & - \sum_{j=1,j \neq i}^N \mathbb{E} \int_{[0,T)} e^{-\alpha t} \Gamma_j v^i(\pmb{X}_{t-}^{-i*},X_{t-}^i,\pmb{Y}_{t-})\left(  d\xi_t^{j*,+} +d\xi_t^{j*,-}\right)+\mathbb{E} \int_{[0,T)} e^{-\alpha t} (v^i_{x^i} d\xi_t^{i,+} - v^i_{x^i} d\xi_t^{i,-}) \\
            & - \mathbb{E} \int_{[0,T)} e^{-\alpha t} \Gamma_iv^i(\pmb{X}_{t-}^{-i*},X_{t-}^i,\pmb{Y}_{t-})\left(  d\xi_t^{i,+} + d\xi_t^{i,-}\right) + \mathbb{E}\sum_{0 \le t <T} e^{-\alpha t} \left(\Delta v^i -\sum_{j=1}^N v^i_{x^j} \Delta X^j_t - \sum_{k=1}^M v^i_{y^k} \Delta Y^k_t\right),
\end{align*}
where $\Gamma_i$ and $\Gamma_j$ are defined in \eqref{gamma}. Here $\Delta v^i := v^{i}(\pmb{X}_{t}^{-i*},X_{t}^i,\pmb{Y}_{t})-v^{i}(\pmb{X}_{t-}^{-i*},X_{t-}^i,\pmb{Y}_{t-})$, $v^{i}_{x^j}:=v^{i}_{x^j}(\pmb{X}_{t-}^{-i*},X_{t-}^i,\pmb{Y}_{t-})$,  $v^{i}_{y^k}:=v^{i}_{y^k}(\pmb{X}_{t-}^{-i*},X_{t-}^i,\pmb{Y}_{t-})$, {$\Delta X_t^{j*}:=X^{j*}_t-X^{j*}_{t-}$, $\Delta X_t^{i}:=X^i_t-X^{i}_{t-}$, and $\Delta Y_t^k:=Y^k_t-Y^k_{t-}$} on the RHS of above equation for $1\leq i,j \leq N$ and $1\leq k \leq M$. By \cite[Theorem3.2.1]{arnold1974stochastic}, condition
$(v)$ implies that the it\^o integral $\int_0^T e^{-\alpha t} \sum_{j=1}^N v_{x^j}^i dB_t^j$ is a martingale. Hence $\mathbb{E} \left[\int_0^T e^{-\alpha t} \sum_{j=1}^N v_{x^j}^i dB_t^j\right]=0$.
The convexity condition in $(iv)$ implies $\mathbb{E}\sum_{0 \le t <T} e^{-\alpha t} (\Delta v^i -\sum_{k\neq i}^N v^i_{x^k} \Delta X^{k*}_t- v^i_{x^i} \Delta X^{i}_t - \sum_{j=1}^M v^i_{y^j} \Delta Y^j_t) \geq 0$.
Next we have
\begin{eqnarray*}
&&\mathbb{E} \int_{[0,T)} e^{-\alpha t} (v^i_{x^i} d\xi_t^{i,+} - v^i_{x^i} d\xi_t^{i,-}) - \mathbb{E} \int_{[0,T)} e^{-\alpha t}\Gamma_iv^i(\pmb{X}_{t-}^{-i*},X_{t-}^i,\pmb{Y}_{t-}) \left(  d\xi_t^{i,+} + d\xi_t^{i,-}\right)\\
&=& \mathbb{E} \int_{[0,T)} e^{-\alpha t}\left[v^i_{x^i}(\pmb{X}_{t-}^{-i*},X_{t-}^i,\pmb{Y}_{t-})-\Gamma_i v^i(\pmb{X}_{t-}^{-i*},X_{t-}^i,\pmb{Y}_{t-})\right]d\xi_t^{i,+} \\
&+& \mathbb{E} \int_{[0,T)} e^{-\alpha t}\left[-v^i_{x^i}(\pmb{X}_{t-}^{-i*},X_{t-}^i,\pmb{Y}_{t-})-\Gamma_iv^i(\pmb{X}_{t-}^{-i*},X_{t-}^i,\pmb{Y}_{t-})\right]d\xi_t^{i,-} \geq 0.
\end{eqnarray*}
{The last inequality holds due to conditions $(ii)$ and $(iv)$. More precisely, $v^i(\pmb{x})$ satisfies the HJB equation \eqref{QVI-1} in $\mathcal{W}_{-i}$. Along with $(iv)$, we have the following with probability one,
\begin{eqnarray*}
v^i_{x^i}(\pmb{X}_{t-}^{-i*},X_{t-}^{i},\pmb{Y}_{t-})-\Gamma_iv^i(\pmb{X}_{t-}^{-i*},X_{t-}^{i},\pmb{Y}_{t-})\ge 0, \\
 -v^i_{x^i}(\pmb{X}_{t-}^{-i*},X_{t-}^{i},\pmb{Y}_{t-})-\Gamma_iv^i(\pmb{X}_{t-}^{-i*},X_{t-}^{i},\pmb{Y}_{t-})\ge 0.
 \end{eqnarray*}
}
For each $j\neq i$, almost surely, we have $d\xi^{j*}_t \neq 0$ only when $(\pmb{X}_t,\pmb{Y}_t) \in \partial \mathcal{W}_{-i}\cap\partial \mathcal{A}_j$. Along with the condition $(ii)$ and \eqref{QVI-3}-\eqref{QVI-2},
\begin{eqnarray*}
&&\mathbb{E} \int_{[0,T)} e^{-\alpha t} (v^i_{x^j}(\pmb{X}_{t-}^{-i*},X_{t-}^i,\pmb{Y}_{t-}) d\xi_t^{j*,+} - v^i_{x^j}(\pmb{X}_{t-}^{-i*},X_{t-}^i,\pmb{Y}_t) d\xi_t^{j*,-}) \\
&&\qquad\qquad\qquad- \mathbb{E} \int_{[0,T)} e^{-\alpha t} \Gamma_j v^i(\pmb{X}_{t-}^{-i*},X_{t-}^i,\pmb{Y}_t)\left(  d\xi_t^{j*,+} + d\xi_t^{j*,-}\right)\\
&=&\mathbb{E} \int_{[0,T)} e^{-\alpha t}\left[v_{x^j}^i-\Gamma_jv^i \right](\pmb{X}_{t-}^{-i*},X_{t-}^i,\pmb{Y}_t)d\xi_t^{j*,+}+\left[-v_{x^j}^i-\Gamma_jv^i \right](\pmb{X}_{t-}^{-i*},X_{t-}^i,\pmb{Y}_t)d\xi_t^{j*,-}=0.
\end{eqnarray*}

Condition $(ii)$ also implies $\mathcal{L}v^i-\alpha v^i \geq -h$.
Combining all of the above, 
\begin{eqnarray}
\label{ineq:ver}
e^{-\alpha T }\mathbb{E}v^i(\pmb{X}_T^{-i*}, X_T^i, \pmb{Y}_T) + \mathbb{E} \int_0^T e^{-\alpha t} h \left(\pmb{X}_t^{-i*},X_t^i\right) dt \ge v^i(\pmb{x},\pmb{y}).
\end{eqnarray}
By letting $T \rightarrow \infty$, the inequality \eqref{ineq:ver} and condition $(vi)$ lead to the desirable inequality. 

{The equality in \eqref{ineq:ver} holds for $\xi^{i}=\xi^{i*}$ by {\eqref{local_time_condition2}},
and
$\mathbb{P}\left((\pmb{X}^{*}_t, \pmb{Y}^*_t) \in \overline{\cap_{i=1}^N\mathcal{W}_{i}}\right) = 1 \quad  \mbox{for all } t \ge 0$ and {the ``no simultaneous jump'' condition in the admissible set \eqref{eq:admset}},
where $(\pmb{X}_t^*,\pmb{Y}_t^*)$ is the dynamics under $\pmb{\xi}^*$.}
\end{proof}


\quad 
Suppose the game value $v^i$ ($i=1,2,\cdots,N$) that satisfies the verification theorem (Theorem \ref{thm:verification}) are given, the next step is to construct the corresponding NE strategies. This is by solving a Skorokhod problem, discussed in the next subsection.

\subsection{Skorokhod Problem}
\label{sc33}

Here we present necessary tools to construct the NE strategies under the additional Assumptions \textbf{H1$'$}-\textbf{H2$'$}.
The key to the analysis is the weak construction of a reflected Brownian motion in a general domain, due to Kang and Williams \cite{KW2007}.
To proceed further, we need a few vocabularies.

\quad Let $G = \cap_{i \in \mathcal{I}} G_i$ be a nonempty domain in $\mathbb{R}^{n+m}$, where $\mathcal{I}$ is a nonempty finite index set and for each $i \in \mathcal{I}$, $G_i$ is a nonempty domain in $\mathbb{R}^{n+m}$.  For simplicity, we assume that $\mathcal{I}=\{1,2,\cdots,I\}$, with $|\mathcal{I}|=I$. For each $i \in \mathcal{I}$, let $\pmb{n}_i: \mathbb{R}^{n+m} \rightarrow \mathbb{R}^{n+m}$ be the unit normal vector field on  $\partial G_i$ that points into $G_i$. And denote $\pmb{r}_i(\cdot):\mathbb{R}^{n+m} \rightarrow \mathbb{R}^{n+m}$ as the reflection direction on $\partial G_i$.
Fix $\pmb{b} \in \mathbb{R}^{n}$ and $\pmb{\sigma} \in \mathbb{R}^{n\times n}$ as the {constant} drift and covariance of the diffusion process without reflection. Let $\nu$ denote a probability measure on $(\overline{G},\mathcal{B}(\overline{G}))$, where $\mathcal{B}(\overline{G})$ is the Borel $\sigma$-algebra on $\overline{G}$.

\quad A Skorokhod problem is to find a reflected diffusion process in $\overline{G}$ such that the initial distribution follows $\nu$, the diffusion parameters are $(\pmb{b},\pmb{\sigma})$, and the reflection direction is $\pmb{r}_i$ on face $\partial G_i$. For each reflection direction $\pmb{r}_i$ ($i \in \mathcal{I}$), denote $\pmb{r}_i^{+}:=({r}_{i,1},\cdots, {r}_{i,n})$ as the vector of the first $n$ components of $\pmb{r}_i$ and denote $\pmb{r}_i^{-}:=({r}_{i ,n+1},\cdots, {r}_{i,n+m})$ as the vector of the next $m$ components of $\pmb{r}_i$. Note that $r_{i,k}^{-} = r_{i,k+n}$ by the usual index rule $(k=1,\cdots,m)$.
Specific to the stochastic game, the following definition is a straightforward modification of \cite[Definition 2.1]{KW2007}.

\begin{definition}\label{CSRBM} A constrained semimartingale reflected Brownian motion (SRBM) associated with the data $(G,\pmb{b},\pmb{\sigma}, \{\pmb{r}_i\}_{i=1}^I,\nu)$ is an $\{\mathcal{F}_t\}$-adapted, $n$-dimensional process $\pmb{X}$ defined on some filtered probability space $(\Omega,\mathcal{F},\{\mathcal{F}_t\},\mathbb{P})$ such that:

\begin{enumerate}[font=\bfseries,leftmargin=2cm]
\item[(i)] $\mathbb{P}$-a.s., $\pmb{X}_t = \pmb{W}_t + \sum_{i \in \mathcal{I}} \int_{[0,t)} \pmb{r}_i^{+}(\pmb{X}_s,\pmb{Y}_s) d \eta^i_s$ for all $t \geq 0$,
\item[(ii)] under $\mathbb{P}$, $ \pmb{W}_t$ is an $n$-dimensional $\mathcal{F}_t$-Brownian motion with drift vector $\pmb{b}$, covariance matrix $\pmb{\sigma}$ and initial distribution $\nu$,
\item[(iii)] $d Y_t^j =\sum_{i \in \mathcal{I}}\int_{[0,t)}\pmb{r}_{i,j}^{-}(\pmb{X}_t,\pmb{Y}_t)d\eta^i_t$ and  $Y^j_t \geq 0$ for $j=1,2,\cdots,m$,
\item[(iv)] for each $i\in \mathcal{I}$, $\eta^i$ is a one-dimensional process such that $\mathbb{P}$-a.s.,
\begin{itemize}
\item[(a)] $\eta^i$ is continuous and nondecreasing with $\eta_0^i=0$,
\item[(b)] $\eta^i_t = \int_{(0,t]}1_{\{(\pmb{X}_s, \pmb{Y}_s) \in \partial G_i \cap \partial G\}}d\eta^i_s$ for all $t\geq 0$,
\end{itemize}
\item[(v)] $\mathbb{P}$-a.s., $(\pmb{X}_t,\pmb{Y}_t)$ has continuous paths and $(\pmb{X}_t,\pmb{Y}_t) \in \overline{G}$ for all $t \geq 0$,
\end{enumerate}
\end{definition}
Here $\pmb{X}_t$ is the controlled diffusion process and $\pmb{Y}_t$ is the resource levels. The domain $G$ restricts the dynamics of both $\pmb{X}_t$ and $\pmb{Y}_t$. 

\quad For each $(\pmb{x},\pmb{y})\in \mathbb{R}^{n+m}$, let $\mathcal{I}(\pmb{x},\pmb{y})=\{i \in \mathcal{I} \,:\, (\pmb{x},\pmb{y})\in \partial G_i\}$.
Let $U_{\epsilon}(S)$ denote the closed set $\{(\pmb{x},\pmb{y})\in \mathbb{R}^{n+m}: dist((\pmb{x},\pmb{y}),S) \leq \epsilon\}$ for any $\epsilon >0$ and $S \subset \mathbb{R}^{n+m}$. If $S=\emptyset$, set $U_{\epsilon}(S)=\emptyset$ for any $\epsilon>0$.
We list the following assumptions on domain $G$ and reflection directions $\{\pmb{r}_i, i\in \mathcal{I}\}$:

\begin{enumerate}[font=\bfseries,leftmargin=2cm]
\item [A1.] G is the nonempty domain in $\mathbb{R}^{n+m}$ such that
\begin{eqnarray}
G = \cap_{i \in \mathcal{I}} G_i,
\end{eqnarray}
where for each $i \in \mathcal{I}$, $G_i$ is a nonempty domain in $\mathbb{R}^{n+m}$, $G_i \neq \mathbb{R}^{m+n}$ and the boundary $\partial G_i$ is $\mathcal{C}^1$. 
\item[A2.] For each $\epsilon \in (0,1)$ there exists $R(\epsilon)>0$ such that for each $i \in \mathcal{I}$, $(\pmb{x},\pmb{y}) \in \partial G_i \cap \partial G$ and $(\pmb{x}^{\prime},\pmb{y}^{\prime}) \in \overline{G}$ satisfying $\|(\pmb{x},\pmb{y})-(\pmb{x}^{\prime},\pmb{y}^{\prime})\| <R(\epsilon)$, we have
\begin{eqnarray*}
\left\langle \pmb{n}_i(\pmb{x},\pmb{y}),(\pmb{x}^{\prime},\pmb{y}^{\prime})-(\pmb{x},\pmb{y}) \right\rangle \geq -\epsilon\|(\pmb{x},\pmb{y})-(\pmb{x}^{\prime},\pmb{y}^{\prime})\|.
\end{eqnarray*}
\item[A3.] The function $D: [0,\infty) \rightarrow [0,\infty]$ is  such that $D(0)=0$ and 
\begin{eqnarray*} 
D(\epsilon) = \sup_{\mathcal{I}_0 \in \mathcal{I},\mathcal{I}_0 \neq \emptyset} \sup \left\{\text{dist}\left((\pmb{x},\pmb{y}), \cap_{i \in \mathcal{I}_0} (\partial G_i \cap \partial G)\right): (\pmb{x},\pmb{y})\in \cap_{i \in \mathcal{I}_0} U_{\epsilon} (\partial G_i \cap \partial G)\right\},
\end{eqnarray*}
for $\epsilon >0$ satisfies $D(\epsilon) \rightarrow 0$ as $\epsilon \rightarrow 0$. 
\item[A4.] There is a constant $L>0$ such that for each $i \in \mathcal{I}$, $\pmb{r}_i(\cdot)$ is a uniformly Lipschitz continuous function from $\mathbb{R}^{n+m}$ into $\mathbb{R}^{n+m}$ with Lipschitz constant $L$ and $\|\pmb{r}_i(\pmb{x},\pmb{y})\|=1$ for each $(\pmb{x},\pmb{y}) \in \mathbb{R}^{n+m}$. 
\item[A5.] There is a constant $a \in (0,1)$, and vector valued function $\pmb{c}(\cdot)=(c_1(\cdot),\cdots,c_I(\cdot))$ and $\pmb{d}(\cdot)=(d_1(\cdot),\cdots,d_I(\cdot))$ from $\partial G$ into $\mathbb{R}_{+}^{I}$ such that for each $(\pmb{x},\pmb{y}) \in \partial G$,
\begin{enumerate}[font=\bfseries,leftmargin=2cm]
\item[(i)] $\sum_{i\in \mathcal{I}(\pmb{x},\pmb{y})} c_i(\pmb{x},\pmb{y})=1$,
$
\min_{k \in \mathcal{I}(\pmb{x},\pmb{y})} \left\langle \sum_{i\in \mathcal{I}(\pmb{x},\pmb{y})}c_i(\pmb{x},\pmb{y}) \pmb{n}_i(\pmb{x},\pmb{y}),\pmb{r}_k(\pmb{x},\pmb{y})\right\rangle \geq a,
$
\item[(ii)] $\sum_{i\in \mathcal{I}(\pmb{x},\pmb{y})} d_i(\pmb{x},\pmb{y})=1$,
$
\min_{k \in \mathcal{I}(\pmb{x},\pmb{y})} \left\langle \sum_{i\in \mathcal{I}(\pmb{x},\pmb{y})}d_i(\pmb{x},\pmb{y}) \pmb{r}_i(\pmb{x},\pmb{y}),\pmb{n}_k(\pmb{x},\pmb{y})\right\rangle \geq a.
$
\end{enumerate}
\end{enumerate}

\begin{theorem}\label{SRBM}
Given Assumptions {\bf{A1}}-{\bf{A5}}, 
there exists a constrained SRBM associated with the data $(G,\pmb{b},\pmb{\sigma},\{\pmb{r}_i, i \in \mathcal{I}\},\nu)$.
\end{theorem}

The proof of Theorem \ref{SRBM} is easily adapted from \cite[Theorem 5.1]{KW2007}, where one constructs a sequence of approximation (random walks) to the {constrained} SRBM and use the invariance principle to establish the weak convergence. 



\section{Nash Equilibrium for  Game $\pmb{C_p}$}
\label{section:gamepooling}

\quad This section analyzes  the NE of  game {$\pmb{C_p}$}. Section \ref{section:HJBpooling} derives the solution to the HJB equations. Section \ref{s33} constructs the controlled process from the HJB solution.
Section \ref{s34} derives the NE for the game $\pmb{C_p}$.
Recall that in  game $\pmb{C_p}$,  $\pmb{A}=[1,1,\cdots,1]^T \in \mathbb{R}^{N \times 1}$, and the unique resource
\begin{equation}
\label{eq:YA}
Y_t = y - \sum_{i = 1}^N \check{\xi}^i_t \quad \mbox{and} \quad Y_{0-} = y .
\end{equation}

\subsection{Solving HJB equations}
\label{s31}
\label{section:HJBpooling}

Define
\begin{equation}
\label{eq:tildex}
\widetilde{x}^i := x^i - \frac{\sum_{j \ne i} x^j}{N-1} \quad \mbox{for } 1 \le i \le N,
\end{equation}
to be the relative position from $x^i$ to the center of $(x^j)_{j \ne i}$.
For game $\pmb{C_p}$, if $\mathcal{A}_i \cap \mathcal{A}_j = \emptyset $, the HJB system simplifies to
 \begin{eqnarray*}
    \textit{(HJB-$C_p$)} \left\{
                \begin{array}{ll}
                \displaystyle \min  \left\{-\alpha v^i +h \left(\frac{N-1}{N} \widetilde{x}^i \right) + \frac{1}{2} \sum_{j=1}^N v^i_{x^jx^j}, -v^i_{y}+v^i_{x^i}, -v^i_{y}-v^i_{x^i}\right\} = 0, \\ [3 pt]
          \hspace{261pt}   \mbox{for } (\pmb{x},y) \in \mathcal{W}_{-i},  \\ [3 pt]
 \displaystyle  {-v^i_{y}-v^i_{x^j} = 0, 
           \hspace{190pt}   \mbox{for } (\pmb{x},y) \in \mathcal{A}^+_{j}, j \ne i, }\\  [3 pt]
 \displaystyle  { -v^i_{y}+v^i_{x^j}= 0, 
           \hspace{190pt}   \mbox{for } (\pmb{x},y) \in \mathcal{A}^-_{j}, j \ne i.}\end{array}
              \right.
\end{eqnarray*}
\quad Now we look for a threshold function $f_N: \mathbb{R}_{+} \rightarrow \mathbb{R}$ such that 
\begin{eqnarray}\label{eq:f_condition}
&f_N \in \mathcal{C}^1(\mathbb{R}_{+},\mathbb{R}),  \quad f^{\prime}_N(x)<0 \mbox{ for } x>0,\quad  \lim_{x\downarrow 0}f_N(x)=\infty,\nonumber\\
& \mbox{ and there exists a unique } x_0>0 \mbox{ such that } f_N(x_0) = 0.
\end{eqnarray}
\noindent It is easy to see that for such $f_N(x)$ satisfying condition \eqref{eq:f_condition},  $z-f_N(z) = \tilde{x}^i-y$ has a unique positive root when $\tilde{x}^i \ge f_N^{-1}(y)$, denoted as $x_{+}^i$. {We consider an even extension of $f_N(x)$ to $(-\infty,0)$ by defining $\tilde{f}_N(x) = f_N(-x)$ for $x<0$.
Then by symmetry,
 $z+\tilde{f}_N(z) = \tilde{x}^i+y$  has a unique negative root
 when $\tilde{x}^i \le -f_N^{-1}(y)$, denoted as  $x_{-}^i$.} See Figure \ref{fig:Cp_initial} for an illustration. In particular, we have $f_N({x}^i_{+}) \ge 0$ when $y \ge x_0+\tilde{x}^i$ and $\tilde{x}^i \ge 0$. Similarly $\tilde{f}_N({x}^i_{-}) \ge 0$ holds when $y \ge -x_0-\tilde{x}^i$ and $\tilde{x}^i\leq 0$. Such an $f_N$ is constructed later in \eqref{eq:fNder} and condition \eqref{eq:f_condition} is verified in Lemma \ref{lemma:f}.

 Then the action region $\mathcal{A}_i$ and the waiting region $\mathcal{W}_i$ of the $i^{th}$ player are specified as
\begin{equation}
\label{eq:AW}
\mathcal{A}^+_i := E_i^{+}  \cap Q_i,\quad \mathcal{A}^-_i := E_i^{-} \cap Q_i \quad,\mathcal{A}_i = \mathcal{A}^+_i\cup \mathcal{A}^-_i,\quad \mbox{and} \quad \mathcal{W}_i:=(\mathbb{R}^{N} \times \mathbb{R}_{+}) \setminus \mathcal{A}_i,
\end{equation}
where
\begin{equation}
\label{eq:threshold}
E_i^{+} : = \left\{(\pmb{x}, y) \in \mathbb{R}^{N} \times \mathbb{R}^*_{+}: \widetilde{x}^i \ge f_N^{-1}(y) \right\}, \, E_i^{-} : = \left\{(\pmb{x}, y) \in \mathbb{R}^{N} \times \mathbb{R}^*_{+}: \widetilde{x}^i \le - f_N^{-1}(y) \right\},
\end{equation}
with  {
\begin{eqnarray}
E_{i,1}^{+} &:=& \left\{(\pmb{x}, y) \in E_{i}^{+} : y \ge \tilde{x}^i+x_0 \right\},\quad E_{i,2}^{+} := \left\{(\pmb{x}, y) \in E_{i}^{+} : y <  \tilde{x}^i+x_0 \right\},\label{eq:threshold21}\\
E_{i,1}^{-} &:=& \left\{(\pmb{x}, y) \in E_{i}^{-} : y \ge - \tilde{x}^i-x_0 \right\},\quad E_{i,2}^{-} := \left\{(\pmb{x}, y) \in E_{i}^{+} : y < - \tilde{x}^i-x_0 \right\},\label{eq:threshold22}
\end{eqnarray}
and  $\{Q_i\}_{i=1}^N$  disjoint and convex partitions of $\mathbb{R}^{N}\times \mathbb{R}_+$
{ such that $Q_i \cap Q_j = (E_i^+\cup E_i^-)\cap (E_j^+\cup E_j^-)\cap \partial \mathcal{W}_{NE}$ for $i\neq j$,}
$\cup_{i=1}^N Q_i = \mathbb{R}^{N}\times \mathbb{R}_+$ and $\alpha \pmb{p} + (1-\alpha) \pmb{q} \in Q_j$ for all $\alpha \in [0,1]$ if $ \pmb{p}\in Q_j$ and $\pmb{q}\in Q_j$ for some $j=1,2,\cdots,N$. {Condition $Q_i \cap Q_j = (E_i^+\cup E_i^-)\cap (E_j^+\cup E_j^-)\cap \partial \mathcal{W}_{NE}$ for $i\neq j$ implies that player $i$ and player $j$ can not jump simultaneous but may apply continuous control (on the boundary of the common waiting region) at the same time.} 
We can define the following mapping
\begin{eqnarray}\label{eq:mapping_Pi}
\Pi(\pmb{x},y) = 
\begin{cases}
 \left(\big(\pmb{x}^{-i},x_{+}^i+\frac{\sum_{k \neq i}x^k}{N-1}\big),f_N(x_+^i)\right), \quad  & {\rm if } \quad (\pmb{x},y)\in Q_i \cap E_{i,1}^{+},\\
 \left((\pmb{x}^{-i},x^i-y),0\right),
  & {\rm if } \quad (\pmb{x},y)\in Q_i \cap E_{i,2}^{+},\\
  \left(\big(\pmb{x}^{-i},\frac{\sum_{k \neq i}x^k}{N-1}+x_{-}^i\big),\,\tilde{f}_N(x_-^i)\right), & {\rm if } \quad (\pmb{x},y)\in Q_i \cap E_{i,1}^{-},\\
\left((\pmb{x}^{-i},x^i+y),0\right), & {\rm if } \quad (\pmb{x},y)\in Q_i \cap E_{i,2}^{-}.
\end{cases}
\end{eqnarray}
 Mapping $\Pi(\cdot)$ is well-defined on $\cup_i \mathcal{A}_i$ since $\{Q_i\}_{i=1}^N$ are disjoint. 
 Note that,  $\Pi(\cdot)$  translates $(\pmb{x},y)$ to the boundary of ${E}_{i,1}^{+}$, i.e., $ \partial {E}_{i,1}^{+}:=\{(\pmb{x},y)\in \mathbb{R}^N \times \mathbb{R}_{+}\,\,:\,\,y={f_N}\left(\tilde{x}^i\right), 0< x \le x_0\}$ when $(\pmb{x},y)\in Q_i \cap E_{i,1}^{+}$, and translates $(\pmb{x},y)$  to the ``zero resource'' plane $\{(\pmb{x},y)\in \mathbb{R}^N \times \mathbb{R}_{+}\,\,:\,\,y=0\}$ when $(\pmb{x},y)\in Q_i \cap E_{i,2}^{+}$, both along the direction $(0,0,\cdots,-1,0,\cdots,-1)\in \mathbb{R}^{N+1}$ nonzero $i$-th and $(N+1)$-th components. 
 Let
\begin{align}
\label{eq:WNE2}
\mathcal{W}_{NE} :&= \{(\pmb{x},y) \in \mathbb{R}^{N+1}: |\widetilde{x}^i| < f_N^{-1}(y) \mbox{ with } y >0, \, 1 \le i \le N\} {  \cup \{(\pmb{x},y) \in \mathbb{R}^{N}\times \mathbb{R}_{+}\,:\,y=0\} } \\
& = \cap_{i = 1}^N \left(E_i^{-} \cup E_i^{+}\right)^c,   \nonumber 
                                   \end{align}

be the common non-action region and  assume that partitions $\{Q_i\}_{i=1}^N$ satisfies the following assumption:
\begin{enumerate}[font=\bfseries,leftmargin=2cm]
\item[H3-${\bf C_p}$.] For any $(\pmb{x},y) \in \cup_i \mathcal{A}_i$, \qquad $\Pi (\pmb{x},y)  \in \overline{\mathcal{W}_{NE}}.$
\end{enumerate}
 Condition {\bf H3-${\bf C_p}$}  implies that if $(\pmb{x},y)\in \mathcal{A}_i$, then the dynamics will be in region $\overline{\mathcal{W}_{NE}}$ after player $i$'s control. For the special case of $N=2$,  we can take $Q_1 = \{(x_1,x_2,y)\in \mathbb{R}^2\times \mathbb{R}_{+}| x_1-x_2 \ge 0\}$ and $Q_2 = \{(x_1,x_2,y)\in \mathbb{R}^2\times \mathbb{R}_{+}| x_2-x_1 >0\}$. Thus Assumption {\bf H3-${\bf C_p}$} is easily satisfied. The verification is deferred to Appendix B. 
}



\quad We seek a solution $v^i(\pmb{x},y) \in \mathcal{C}^2(\overline{\mathcal{W}_{-i}})$ such that if $|\widetilde{x}^i | < f_N^{-1}(y)$, it is 
of the form, \begin{equation}
\label{eq:candidateA}
v^i(\pmb{x},y) =  p_N(\widetilde{x}^i) + A_N(y) \cosh\left(\widetilde{x}^i \sqrt{\frac{2(N-1) \alpha}{N}} \right),
\end{equation}
where 
\begin{equation}
\label{eq:pN}
p_N(x): = \mathbb{E}\int_0^{\infty} e^{-\alpha t} h \left(\frac{N-1}{N} x+ \sqrt{\frac{N-1}{N}} B_t \right)dt,
\end{equation}
with $B_t$ being a one-dimensional Brownian motion.
Note that $p_N(\widetilde{x}^i)$ is a solution to $-\alpha v^i +h (\frac{N-1}{N} \widetilde{x}^i) + \frac{1}{2} \sum_{j=1}^N v^i_{x^jx^j} = 0$,
which corresponds to the waiting region, 
and
$\cosh \left(\sqrt{\frac{2(N-1) \alpha}{N} }\widetilde{x}^i\right)$ is a solution to $-\alpha v^i + \frac{1}{2} \sum_{j=1}^N v^i_{x^jx^j} = 0$.
If there is no resource, then $v^i(\pmb{x}, y) = p_N(\widetilde{x}^i)$, so $A_N(0) = 0$. 
{ 
The following lemma summarizes basic properties of $p_N$, which can be verified by straightforward calculations.
The proof is hence omitted.
\begin{lemma}\label{lemma:pN} Under Assumption  \textbf{H1$'$}-\textbf{H2$'$}, $p_N(x)$ defined in \eqref{eq:pN} satisfies:
\begin{eqnarray}\label{eq:p_property}
& p_N^{\prime}(x) \ge 0 \,\,\mbox{ and }\,\, p_N^{\prime\prime\prime}(x) \leq 0  \mbox{ for } x\ge0;\,\,\,\, p_N(x)=p_N(-x) \,\,\mbox{ and }\,\, \frac{k}{\alpha}\leq p_N^{\prime\prime}(x) \leq \frac{K}{\alpha}\,\,  \mbox{ for } \,\,x\in \mathbb{R}.
\end{eqnarray}
\end{lemma}
}

\quad The {\it smooth-fit principle} states that, along the  boundary $y=f_N(\tilde{x}^i)$ between the continuation set $\mathcal{W}_{-i}$ and the action set $\mathcal{A}_i$,  $v^i$ has certain regularity  properties across the hyperplane. Now applying the smooth-fit principle, we get $v^i_{x^i x^i}=v^i_{yy}=-v^i_{x^i y}$ at the boundary $y = f_N(\widetilde{x}^i)$ with $\widetilde{x}^i > 0$. This follows from $v^i_{x^i}+v^i_{y}=0$ and we expect $v^i \in \mathcal{C}^2(\mathcal{W}_{-i})$. 
{ To see this, we differentiate the form \eqref{eq:candidateA} twice, and the conditions $v^i_{x^i}+v^i_{y}=0$ and $v^i_{x^i x^i}+v^i_{x^i y} = 0$ at the boundary $y = f_N(\widetilde{x}^i)$ lead to
}
 \begin{eqnarray}\label{smoothfit} {\small
 \left\{
                \begin{array}{ll}
                \displaystyle A_N^{\prime}(f_N{(x)}) = - p_N^{\prime}{(x)} \cosh  \left(x \sqrt{\frac{2(N-1) \alpha}{N} } \right) + p_N^{\prime\prime} {(x)}\sqrt{\frac{N}{2(N-1) \alpha}} \sinh \left( x \sqrt{\frac{2(N-1) \alpha}{N} } \right)\Bigg|_{x = f_N^{-1}(y)}, \\
                \displaystyle A_N(f_N{(x)}) =  p_N^{\prime}{(x)} \sqrt{\frac{N}{2(N-1) \alpha}} \sinh  \left( x \sqrt{\frac{2(N-1) \alpha}{N} } \right) -  p_N^{\prime\prime} {(x)}\frac{N}{2(N-1) \alpha} \cosh \left( x \sqrt{\frac{2(N-1) \alpha}{N} } \right)\Bigg|_{x = f_N^{-1}(y)}.
                \end{array}
\right.}
\end{eqnarray}
As a consequence,
\begin{equation}
\label{eq:fNder}
f'_N(x) = \frac{p_N^{\prime}{(x)} -  \frac{N}{2(N-1) \alpha} p_N^{\prime \prime \prime}{(x)} }{p_N^{\prime \prime} {(x)}\sqrt{\frac{N}{2(N-1) \alpha}} \tanh \left(x \sqrt{\frac{2(N-1) \alpha}{N} } \right)-p_N^{\prime}{(x)}},
\end{equation}
and
\begin{equation}
\label{eq:AN}
A_N(y) = p_N^{\prime} {(x)}\sqrt{\frac{N}{2(N-1) \alpha}} \sinh  \left( x \sqrt{\frac{2(N-1) \alpha}{N} } \right) -  p_N^{\prime \prime}{(x)} \frac{N}{2(N-1) \alpha} \cosh \left( x \sqrt{\frac{2(N-1) \alpha}{N} } \right) \Bigg|_{x = f_N^{-1}(y)}.
\end{equation}
{
\begin{lemma}\label{lemma:f} Under Assumptions \textbf{H1$'$}-\textbf{H2$'$}, $f_N$ defined in \eqref{eq:fNder} satisfies condition \eqref{eq:f_condition}. Moreover, the curve $y = f_N(x)$ intersects $\{x > 0\}$ at $x_0$ such that $A_N(f_N(x_0)) = 0$ and $x_0$ is the unique positive root of
\begin{equation}
\label{eq:intercept}
 \sqrt{\frac{2(N-1) \alpha}{N}} \tanh\left(z \sqrt{\frac{2(N-1) \alpha}{N}}  \right) = \frac{p_N^{\prime \prime}(z)}{p_N^{\prime}(z)}.
\end{equation}
\end{lemma}

\begin{proof}
First we prove that $f_N$ is decreasing on $\mathbb{R}_{+}$. Recall the expression of $f'_N$ from \eqref{eq:fNder}, and
we claim that $f'_N(z)<0$ when $z \geq 0$ and $\lim_{z \downarrow 0}f_N^{\prime}(z) =-\infty$. 
To see this, 
$p_N^{\prime}(z) -  \frac{N}{2(N-1) \alpha} p_N^{\prime\prime\prime}(z) \geq 0$ for $z\geq 0$ by Lemma \ref{lemma:pN}. 
Denote $q(z) = p_N^{\prime\prime}(z) \sqrt{\frac{N}{2(N-1) \alpha}} \tanh \left(z \sqrt{\frac{2(N-1) \alpha}{N} } \right)-p_N^{\prime}(z)$. It is easy to see that $q(0)=0$. Moreover, $q^{\prime}(z) = p_N^{\prime\prime\prime}(z)\sqrt{\frac{N}{2(N-1) \alpha}} \tanh \left(z \sqrt{\frac{2(N-1) \alpha}{N} } \right)+p_N^{\prime\prime}(z) \frac{1}{\cosh^2 \left(z \sqrt{\frac{2(N-1) \alpha}{N} } \right)}-p_N^{\prime\prime}(z) < 0$ for $z > 0$ and $q^{\prime}(z)=0$ for $z=0$. This is because $p_N^{\prime\prime\prime}(z)\leq 0$ ($z \geq 0$) by Lemma \ref{lemma:pN}, $\cosh(z) \geq 1$ ($z \geq 0$), and $\cosh(z) = 1$ if and only if $z=0$. 
{Let  $s(x) = p'_N(x) - \frac{N}{2(N-1) \alpha} p''_N(x)$. So 
$f'_N(x) = s(x)/q(x)$.
It is clear that for $x > 0$, $f'_N(x) < 0$ (since $s(x) > 0$ and $q(x) < 0$). Now we consider the asymptotics of $s(x)$ and $q(x)$ as $x \to 0^{+}$. By Taylor's expansion,
$$q(x) = p''_N(0) \sqrt{\frac{N}{2(N-1)\alpha}} x \sqrt{\frac{2(N-1)\alpha}{N}} + o(x) - p''_N(0)x + o(x) = o(x).$$
Since $p''_N(x) < 0$ for $x > 0$, we have $s(x) \ge p'_N(x) = p''_N(0) x + o(x)$.
Therefore, $f'_N(x) = s(x)/q(x) \to -\infty$ as $x \to 0^{+}$. This implies that $f_N(x) \to \infty$ as $x \to 0^+$.} Similarly, $z + \tilde{f}_N(z) = \widetilde{x}^i + y$ has a unique negative root since $f_N(-x) = f_N(x)$.

We then prove the unique positive root of \eqref{eq:intercept}. 
Define $r(z) = \frac{p_N^{\prime\prime}(z)}{p_N^{\prime}(z)}$ where $p_N(x)$ is defined in \eqref{eq:pN}.
Note that 
$r(0) = \frac{p_N^{\prime\prime}(0)}{p_N^{\prime}(0)}=\frac{\mathbb{E}\left[\int_0^{\infty}e^{-\alpha t}h^{\prime \prime}\left(\sqrt{\frac{N-1}{N}}B_t\right) dt\right]}{\mathbb{E}\left[\int_0^{\infty}e^{-\alpha t}h^{\prime }\left(\sqrt{\frac{N-1}{N}}B_t \right)dt\right]}.$
By Assumption \textbf{H2${'}$}, $p_N^{\prime}(0)=0$, $\frac{k}{\alpha}<p_N^{\prime\prime}(0)<\frac{K}{\alpha}$, and
$
r^{\prime}(z) = \frac{p_N^{\prime\prime\prime}(z)p_N^{\prime}(z)-(p_N^{\prime\prime}(z))^2}{(p_N^{\prime}(z))^2}.
$
Along with Lemma \ref{lemma:pN}, we have $r(0)=\infty$ and $r^{\prime}(z)\leq 0$.
Furthermore, since $k \leq h^{\prime\prime}\leq K$ and $h^{\prime} \geq kx +c$ for some constant $c$, we have $\lim_{x \rightarrow \infty}r(x)=0$. Moreover, define $ f(x)= \sqrt{\frac{2(N-1) \alpha}{N}} \tanh\left(x \sqrt{\frac{2(N-1) \alpha}{N}}  \right)$, then it is easy to check that $f(0)=0$, $f^{\prime}(x)>0$ for $x \geq 0$, and $\lim_{x \rightarrow\infty}f(x)= \sqrt{\frac{2(N-1) \alpha}{N}} $. Therefore, $f(x)=r(x)$ has a unique positive solution.
\end{proof}
}


\subsection{Controlled dynamics}
\label{s33}
Given the candidate game value to $(\mbox{HJB-$C_p$})$, we derive the corresponding NEP by showing the existence of a weak solution $(\pmb{X}_t,{Y}_t)$ to a Skorokhod problem with an unbounded domain, where the boundary of the domain depends on both the diffusion term $\pmb{X}_t$ and the degenerate term $\pmb{Y}_t$. 

{Recall the region $\mathcal{W}_{NE}$ defined in \eqref{eq:WNE2} and note that $\mathcal{W}_{NE}$ is  unbounded  in $\mathbb{R}^{N+1}$ with $2N$ boundaries. For $i=1,2,\cdots,N$, define the $2N$ faces of $\mathcal{W}_{NE}$ as
\begin{eqnarray*}
F_i = \{(\pmb{x},y) \in \partial \mathcal{W}_{NE}\,\,\,\vert\,\,\,  (\pmb{x},y) \in \partial E_i^{+}\},\quad
F_{i+N} = \{(\pmb{x},y) \in \partial \mathcal{W}_{NE}\,\,\,\vert \,\,\, (\pmb{x},y) \in \partial E_i^{-}\}.
\end{eqnarray*}
Then the normal direction of each face is given by ($i=1,2,\cdots,N$)
\begin{eqnarray*}
\pmb{n}_i &=&c_i \left(-\frac{1}{N-1}, \cdots, -\frac{1}{N-1}, 1, -\frac{1}{N-1}, \cdots, -\frac{1}{N-1}, (f_{N}^{-1})^{\prime}(y) \right),\\
\pmb{n}_{i+N} &=& c_{i+N} \left(\frac{1}{N-1}, \cdots, \frac{1}{N-1}, -1, \frac{1}{N-1}, \cdots, \frac{1}{N-1}, (f_{N}^{-1})^{\prime}(y) \right),
\end{eqnarray*}
with the $i^{th}$ component to be $\pm 1$.  
$c_i$, $c_{N+i}$ are normalizing constants such that $\|\pmb{n}_i\|=\|\pmb{n}_{N+i}\|=1$.}

Denote the reflection direction on each face as
\begin{eqnarray*}
\pmb{r}_i = c_i^{\prime} \left(0,\cdots,-1,\cdots, 0, -1\right),\quad
\pmb{r}_{N+i} = c^{\prime}_{N+i}\left(0,\cdots,1,\cdots, 0,-1\right),
\end{eqnarray*}
with the $i^{th}$ component to be $\pm 1$.
$c_i^{\prime}$, $c^{\prime}_{N+i}$ are normalizing constants such that $\|\pmb{r}_i\|=\|\pmb{r}_{N+i}\|=1$. NE strategy is defined as follows.

{\bf Case 1:} $(\pmb{X}_{0-},Y_{0-}) = (\pmb{x},y) \in  \overline{\mathcal{W}_{NE}}$. One can check that $\mathcal{W}_{NE}$ defined in \eqref{eq:WNE2} and $\{\pmb{r}_i\}_{i=1}^{2N}$  defined above satisfies assumptions \textbf{A1}-\textbf{A5}. (See Appendix A for the satisfiability of \textbf{A1}-\textbf{A5}). According to Theorem \ref{SRBM}, there exists a weak solution to the Skorokhod problem with data $\left(\mathcal{W}_{NE}, \{\pmb{r}_i\}_{i=1}^{2N},\pmb{b},\pmb{\sigma},\pmb{x}\in \overline{\mathcal{W}_{NE}}\right)$.

{\bf Case 2:} $(\pmb{X}_{0-},Y_{0-}) = (\pmb{x},y) \notin  \overline{\mathcal{W}_{NE}}$, that is, there exists $i \in \{1, \cdots, N\}$ such that $(\pmb{X}_{0-}, Y_{0-}) \in \mathcal{A}_i$. (1) { If $(\pmb{x},y) \in \mathcal{A}_i^{+}\cap E_{i,1}^+$, then $\tilde{x}^i \ge f_N^{-1}(y)$ and $y \ge \tilde{x}^i+x_0$. In this case, player $i$ will move immediately from $X_{0-}^i=x^i$ to $X_{0}^i=x^i_{+}+ \frac{\sum_{k \neq i}x^k}{N-1}$ at time $0$, where $x^i_{+}$ is the unique positive root such that $z - f_N(z) = \tilde{x}^i - y$. This will reduce the initial resource from $Y_{0-}=y$ to $Y_{0} = f_N(x^i_{+})\ge 0$.  $ f_N(x^i_{+})\ge 0$ holds since $y \ge x_0+\tilde{x}_i$ when $(\pmb{x},y)\in {E}_{i,1}^+$. Other players' dynamics remain unchanged, i.e., $X_{0-}^j = X_{0}^j=x^j$ for $j\neq i $ and $1 \leq j \leq N$. By Assumption \textbf{H3-$\bf C_{p}$}, we have $(\pmb{X}_0,Y_{0})=\left(\big(\pmb{x}^{-i},\frac{\sum_{k \neq i}x^k}{N-1}+x_{+}^i\big),f_N(x_+^i)\right) = \Pi(\pmb{X}_{0-},Y_{0-}) \in  \overline{\mathcal{W}_{NE}}$.
(2)  If $(\pmb{x},y) \in \mathcal{A}_i^{+}\cap E_{i,2}^+$, then $\tilde{x}^i \ge f_N^{-1}(y)$ and $y < \tilde{x}^i+x_0$. In this case, player $i$ will move immediately from $X_{0-}^i=x^i$ to $X_{0}^i=x^i-y$ and the initial resource $Y_{0-}=y$ is decreased to $Y_{0} = 0$ at time $0$. 
Other players' dynamics remain unchanged, i.e., $X_{0}^j = X_{0-}^j=x^j$ for $j\neq i $ and $1 \leq j \leq N$. By Assumption \textbf{H3-$\bf C_{p}$}, we have $(\pmb{X}_0,Y_{0})=\left((\pmb{x}^{-i},x^i-y),0\right)= \Pi(\pmb{X}_{0-},Y_{0-})\in  \overline{\mathcal{W}_{NE}}$.
(3) Similarly,  if $(\pmb{x},y) \in \mathcal{A}_i^{-}\cap  E_{i,1}^-$, then $\tilde{x}^i \le - f_N^{-1}(y)$ and $y \ge -\tilde{x}^i-x_0$. And player $i$ will move immediately from $X_{0-}^i=x^i$ to $X_{0}^i=x^i_{-}+ \frac{\sum_{k \neq i}x^k}{N-1}$ at time $0$, where $x^i_{-}$ is the unique negative root such that $z+\tilde{f}_N(z) = \tilde{x}^i+y$, and  $Y_{0-}=y$ is now $Y_{0} = \tilde{f}_N(x^i_{-}) \ge 0$. 
Other players' dynamics remain unchanged, i.e., $X_{0}^j = X_{0-}^j=x^j$ for $j \neq i$ and $1 \leq j \leq N$. By Assumption \textbf{H3-$\bf C_{p}$}, we have $(\pmb{X}_0,Y_{0})=\left(\big(\pmb{x}^{-i},\frac{\sum_{k \neq i}x^k}{N-1}+x_{-}^i\big),
\tilde{f}_N(x_-^i)\right)= \Pi(\pmb{X}_{0-},Y_{0-})\in  \overline{\mathcal{W}_{NE}}.$ (4)
If $(\pmb{x},y) \in \mathcal{A}_i^{-}\cap E_{i,2}^-$, then $\tilde{x}^i \le -f_N^{-1}(y)$ and $y <-\tilde{x}^i-x_0$. In this case, player $i$ will move immediately from $X_{0-}^i=x^i$ to $X_{0}^i=x^i+y$  and this will change $Y_{0-}=y$ to $Y_{0} = 0$ at time $0$.  Other players' dynamics remain unchanged, i.e., $X_{0-}^j = X_{0}^j=x^j$ for $j \neq i$ and $1 \leq j \leq N$. By Assumption \textbf{H3-$\bf C_{p}$}, we have $(\pmb{X}_0,Y_{0})=\left(\big(\pmb{x}^{-i},x^i+y\big),0\right)=\Pi(\pmb{X}_{0-},Y_{0-}) \in  \overline{\mathcal{W}_{NE}}.$ 
}
 \begin{figure}
    \centering 
  \includegraphics[width=0.9\linewidth]{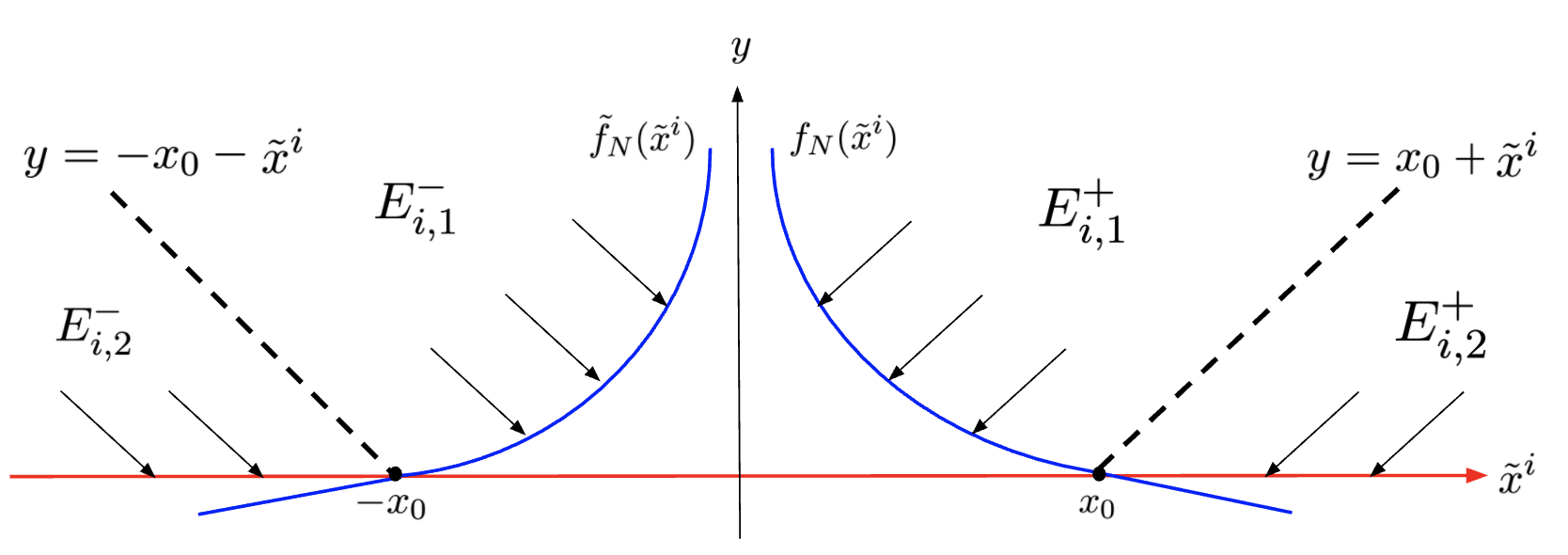}  
  
\caption{\label{fig:Cp_initial}Demonstration of the initial control when  $(\pmb{X}_{0-},Y_{0-}) = (\pmb{x},y) \notin  \overline{\mathcal{W}_{NE}}$.}
\end{figure}
\subsection{NE for the $N$-player game}
\label{s34}
Combining the results in Sections \ref{s31} and \ref{s33}, and based on the verification theorem developed in Section
\ref{section:verification}, we have the following theorem of the NE for the $N$-player game \eqref{eq:JA_game} with constraint \eqref{eq:YA}. 

\begin{theorem}[NE for the $N$-player game {$\pmb{C_p}$}]
\label{thm:NA}
Assume \textbf{H1$'$}-\textbf{H2$'$} {and \textbf{H3-$\bf C_{p}$}}. Define $u^i \in\mathbb{R}^N \times \mathbb{R}_{+} \rightarrow \mathbb{R}$ by {
\begin{equation}
\label{eq:valueN-aux}
u^i(\pmb{x},y) = 
\left\{ \begin{array}{cll}
 p_N(\widetilde{x}^i) + A_N(y) \cosh\left( \widetilde{x}^i \sqrt{\frac{2(N-1) \alpha}{N} }\right) & \mbox{if } |\tilde{x}^i| \le f^{-1}_N(y),  \mbox{ and }y=0,\\  [3 pt]
 u^i\left(\big(\pmb{x}^{-i}, x^i_{+} + \frac{\sum_{k \neq i}x^k}{N-1}\big), f_N(x^i_{+})\right)  & \mbox{if }  (\pmb{x}, y) \in  E_{i,1}^{+}, \\ [3pt]
  u^i\left((\pmb{x}^{-i}, x^i-y), 0\right)  & \mbox{if }  (\pmb{x}, y) \in  E_{i,2}^{+}, \\ [3pt]
 u^i\left(\pmb{x}^{-i}, \frac{\sum_{k \neq i}x^k}{N-1} + x^i_{-}, \tilde{f}_N(x^i_{-})\right)  & \mbox{if } (\pmb{x}, y) \in   E_{i,1}^{-}, \\
   u^i\left((\pmb{x}^{-i}, x^i+y), 0\right)   & \mbox{if } (\pmb{x}, y) \in   E_{i,2}^{-}, 
\end{array}\right.
\end{equation}
and define
$v^i: \mathbb{R}^N \times \mathbb{R}_{+} \rightarrow \mathbb{R}$ as
\begin{equation}
\label{eq:valueN}
v^i(\pmb{x},y) = 
\left\{ \begin{array}{cll}
u^i(\pmb{x},y) & \mbox{if } (\pmb{x}, y) \in \overline{\mathcal{W}_{-i}} , \\ [3 pt]
 v^i\left(\pmb{x}^{-j}, x^j_{+} + \frac{\sum_{k \neq j}x^k}{N-1}, f_N(x^j_{+})\right) & \mbox{if } (\pmb{x}, y) \in \mathcal{A}_j^+ \cap E_{j,1}^{+} \mbox{ for } j \ne i, \\ [3 pt]
 v^i\left(\pmb{x}^{-j}, x^j -y, 0\right) & \mbox{if } (\pmb{x}, y) \in \mathcal{A}_j^+ \cap E_{j,2}^{+} \mbox{ for } j \ne i, \\ [3 pt]
v^i\left(\pmb{x}^{-j}, \frac{\sum_{k \neq j}x^k}{N-1} + x^j_{-}, \tilde{f}_N(x^j_{-})\right) & \mbox{if } (\pmb{x}, y) \in \mathcal{A}_j^- \cap E_{j,1}^{-} \mbox{ for } j \ne i, 
\\ [3 pt]
 v^i\left(\pmb{x}^{-j}, x^j +y, 0\right) & \mbox{if } (\pmb{x}, y) \in \mathcal{A}_j^- \cap E_{j,2}^{-} \mbox{ for } j \ne i, 
\end{array}\right.
\end{equation}}
where 
\begin{itemize}[itemsep = 3 pt]
\item
$\mathcal{A}_i$ and $\mathcal{W}_i$ are given in \eqref{eq:AW}, and $E^{\pm}_{i,1}$ and $E^{\pm}_{i,2}$ are given in \eqref{eq:threshold21}- \eqref{eq:threshold22} with $f_N(\cdot)$ defined by \eqref{eq:fNder}-\eqref{eq:intercept}, {and $\tilde{f}_N (x)= f_N(-x)$ for $x<0$.}
\item
$\widetilde{x}^i$ is defined by \eqref{eq:tildex}, and $A_N(\cdot)$ is defined by \eqref{eq:AN}.
\item
$x^i_{+}$ is the unique positive root of $z - f_N(z) = \widetilde{x}^i - y$ when $\tilde{x}^i \ge f_N^{-1}(y)$, and $x^i_{-}$ is the unique negative root of $z + \tilde{f}_N(z) = \widetilde{x}^i + y$ when $\tilde{x}^i <- f_N^{-1}(y)$.
\end{itemize}
Then $v^i$ is the game value associated with an NEP $\pmb{\xi}^{*} = (\xi^{1*}, \cdots, \xi^{N*})$. That is,
$
v^{i}(\pmb{x}, y) = J_{C_p}^i(\pmb{x},y; \pmb{\xi}^{*}).
$
Moreover, 
the controlled process $(\pmb{X}^{*}, Y^{*})$ under $\pmb{\xi}^{*}$ is given in Section \ref{s33}.
\end{theorem}


\begin{proof} 
First,  $u^i(\pmb{x},y)\in \mathcal{C}^2(\mathbb{R}^N\times \mathbb{R}_{+})$ by construction: { the $\mathcal{C}^2$ regularity near $y = 0$ follows from \eqref{eq:AN}, and 
the facts that $f_N^{-1}(y) \to x_0$ as $y \to 0$ and $A_N(f_N(x_0)) = 0$. 
} { 
To see that $z - f_N(z) = \widetilde{x}^i - y$ has a unique positive root, it suffices to prove that $f_N$ is decreasing on $\mathbb{R}_{+}$.
This fact is shown in Lemma \ref{lemma:f}.}
Now let us check conditions (i)-(vii) in Theorem \ref{thm:verification}.
\begin{itemize}
\item[(i)] Based on the analysis in Section \ref{s33}, when $(\pmb{x},y)\in \overline{\mathcal{W}_{NE}}$, the NE strategy is a solution to the Skorokhod problem specified in Case 2, which is a continuous process. When $(\pmb{x},y)\notin \mathcal{W}_{NE}$, the initial push specified in Case 1 satisfies the ``no simultaneous jump'' condition. {Note when the fuel is used up, the dynamics $\pmb{X}_t$ will become uncontrolled and move freely without control.} 
\item[(ii)] {
Now we check condition (ii) in the verification theorem, i.e.,  $v^i$ defined in \eqref{eq:valueN} satisfying the QVI \eqref{QVI}. It consists of the following three steps.  { The idea is to apply the Implicit Function Theorem and the calculation follows the lemma in \cite[p.58]{BSW}.}

\noindent{\bf Step 1} is to verify that $v^i$ defined in \eqref{eq:valueN} satisfies 
\begin{eqnarray}\label{diffusion_term}
-\alpha v^i +h\left(\frac{N-1}{N}\tilde{x}^i\right) +\frac{1}{2}\sum_{j=1}^N v^i_{x^jx^j} \geq 0
\end{eqnarray}
 for $(\pmb{x},y)\in \overline{\mathcal{W}_{-i}}$ and that the inequality is strict for  $(\pmb{x},y)\in \mathcal{A}_{i}$ and the equality holds in $\overline{\mathcal{W}_{NE}}$.

Since $p_N(\tilde{x}^i)$ is a solution to $-\alpha v^i +h\left(\frac{N-1}{N}\tilde{x}^i\right) +\frac{1}{2}\sum_{j=1}^N v^i_{x^jx^j} = 0$ and $\cosh\left(\sqrt{\frac{2(N-1)\alpha}{N}}\tilde{x}^i\right)$ is a solution to $-\alpha v^i +\frac{1}{2}\sum_{j=1}^N v^i_{x^jx^j} = 0$, $p_N(\widetilde{x}^i) + A_N(y) \cosh\left( \widetilde{x}^i \sqrt{\frac{2(N-1) \alpha}{N} }\right)$ satisfies  $-\alpha v^i +h\left(\frac{N-1}{N}\tilde{x}^i\right) +\frac{1}{2}\sum_{j=1}^N v^i_{x^jx^j} = 0$. Therefore (\ref{diffusion_term}) holds for $(\pmb{x},y)\in\overline{\mathcal{W}_{NE}}$ with equality.

Denote $\pmb{p}=(\pmb{w},z)$ with $\pmb{w}\in \mathbb{R}^N$ and $z \in \mathbb{R}_{+}$. When $\pmb{p}\in \mathcal{A}_i^{+}\cap {E}_{i,1}^{+}$, we have $v^i(\pmb{p})=v^i(\pmb{q})$ where $\pmb{q}:=\left(\pmb{w}^{-i},w_{+}^i+\frac{\sum_{k \neq i}w^k}{N-1},f_N(w_+^i)\right)=\Pi(\pmb{p})$  translates  $\pmb{p}$ to the boundary of ${E}_i^{+}$, i.e., $\partial {E}_i^{+}:=\{(\pmb{x},y)\,\,|\,\,y=f_N^{-1}\left(\tilde{x}^i\right)\}$  along the direction $(0,0,\cdots,-1,0,\cdots,-1)\in \mathbb{R}^{N+1}$ with all components  zero except the $i$-th and $(N+1)$-th components being $-1$. {Note that when $ \pmb{p}=(\pmb{w},z)\in  {E}_{i,1}^{+}$, we have  $z \ge \tilde{w}_i+x_0$ and $f_N(w_+^i) \ge 0$. (See Figure \ref{fig:Cp_initial}). By the Implicit Function Theorem,}  
$
v^i_{x^ix^i}(\pmb{p}) = \frac{v^i_{x^ix^i}(\pmb{q})+f^{\prime}_N(w_{+}^i)v^i_{x^i y}(\pmb{q})}{1-f^{\prime}_N(w_{+}^i)} = v^i_{x^ix^i}(\pmb{q}),
$
the last equality holds since $v^i_{x^i x^i} = -v^i_{x^i y}$ on $y = f_N(\tilde{x}^i)$. {To see this more clearly, Denote 
$\pmb{p} :=(\pmb{w},z)$ with $\pmb{w}\in \mathbb{R}^N$ and $y \in \mathbb{R}_{+}$ such that $\pmb{p}\in \mathcal{A}_i \cap E_{i,1}^{+}$. And also denote $\pmb{q} := (\pmb{w}^{-i},w^i-\theta,z-\theta) $ such that $z-\theta = f_N(\tilde{w}_i-\theta)$. Then we have ${v}^i(\pmb{p}) ={v}^i(\pmb{q})$ by the definition of $v^i$.
Taking the derivative of  $z-\theta = f_N(\tilde{w}_i-\theta)$ with respect to $w_i$ leads to $-\frac{\partial \theta}{\partial w_i} = f^{\prime}_N (\tilde{w}_i-\theta) \left(1-\frac{\partial \theta}{\partial w_i}\right)$, and hence $\frac{\partial \theta}{\partial w_i} = -\frac{f_N^{\prime}(\tilde{w}_i-\theta)}{1-f_N^{\prime}(\tilde{w}_i-\theta)}$. Then 
\begin{eqnarray*}
v^i_{x^i}(\pmb{p}) &=& \frac{\partial v^i}{\partial w_i}(\pmb{w}^{-i},w^i-\theta,z-\theta)\\
&=& \left(1-\frac{\partial \theta}{\partial w_i}\right)v^i_{x^i}(\pmb{w}^{-i},w^i-\theta,z-\theta) - v^i_{y}(\pmb{w}^{-i},w^i-\theta,z-\theta)\frac{\partial \theta}{\partial w_i}\\
&=& \left(1-\frac{\partial \theta}{\partial w_i}\right)v^i_{x^i}(\pmb{w}^{-i},w^i-\theta,z-\theta) + v^i_{x^i}(\pmb{w}^{-i},w^i-\theta,y-\theta)\frac{\partial \theta}{\partial w_i}\\
&=&v^i_{x^i}(\pmb{w}^{-i},w^i-\theta,z-\theta).
\end{eqnarray*}
The second last equation holds since $v_{x^i}^i+v_{y}^i = 0$ on $\overline{\mathcal{W}}_i\cap \overline{\mathcal{A}}_i^{+}$. Similarly, 
\begin{eqnarray}
v^i_{x^ix^i}(\pmb{p}) &=& \frac{\partial v_i^i}{\partial w_i}(\pmb{w}^{-i},w^i-\theta,z-\theta)\nonumber\\
&=& \left(1-\frac{\partial \theta}{\partial x_i}\right)v^i_{x^i x^i}(\pmb{w}^{-i},w^i-\theta,z-\theta) - v^i_{x^i y}(\pmb{w}^{-i},w^i-\theta,z-\theta)\frac{\partial \theta}{\partial w_i}\nonumber\\
&=& \left(1-\frac{\partial \theta}{\partial w_i}\right)v^i_{x^ix^i}(\pmb{w}^{-i},w^i-\theta,z-\theta) + v^i_{x^i x^i}(\pmb{w}^{-i},w^i-\theta,z-\theta)\frac{\partial \theta}{\partial x_i}\nonumber\\
&=&v^i_{x^i x^i}(\pmb{w}^{-i},w^i-\theta,z-\theta).\label{eq:derivative2}
\end{eqnarray}
The second last equation holds since $v_{x^i x^i}^i+v_{x^i y}^i = 0$ on $\overline{\mathcal{W}}_i\cap \overline{\mathcal{A}}_i^{+}$.  }

{Similarly, we have} $v^i_{x^j x^j}(\pmb{p}) =v^i_{x^j x^j}(\pmb{q})$  for $j\neq i$. {To prove this, take the derivative of  $z-\theta = f_N(\tilde{w}_i-\theta)$ with respect to $w_j$ for  $j \neq i$ and $j\le N$, we have $-\frac{\partial \theta}{\partial w_j} = f^{\prime}_N (\tilde{w}_i-\theta) \left(-\frac{1}{N-1}-\frac{\partial \theta}{\partial w_i}\right)$, and hence $\frac{\partial \theta}{\partial w_j} = \frac{1}{N-1}\frac{f_N^{\prime}(\tilde{w}_i-\theta)}{1-f_N^{\prime}(\tilde{w}_j-\theta)}$. 
Therefore,
\begin{eqnarray*}
v^i_{x_j}(\pmb{p}) &=& \frac{\partial v^i}{\partial w_j}(\pmb{w}^{-i},w^i-\theta,z-\theta)\\
&=& -v^i_{x^i}(\pmb{w}^{-i},w^i-\theta,z-\theta) \frac{\partial \theta}{\partial w_j}- v^i_{y}(\pmb{w}^{-i},w^i-\theta,z-\theta)\frac{\partial \theta}{\partial w_j}+ v_{x^j}^i(\pmb{w}^{-i},w^i-\theta,z-\theta) \\
&=&v^i_{x^j}(\pmb{w}^{-i},w^i-\theta,z-\theta).
\end{eqnarray*}
The  last equation holds since $v_{x^i}^i+v_{y}^i = 0$ on $\overline{\mathcal{W}}_i\cap \overline{\mathcal{A}}_i^{+}$. Similarly, we have
\begin{eqnarray*}
v^i_{x^j x^j}(\pmb{p}) &=& \frac{\partial v^i_{x^j}}{\partial w_j}(\pmb{w}^{-i},w^i-\theta,z-\theta)\nonumber\\
&=& -v^i_{x^i x^j}(\pmb{w}^{-i},w^i-\theta,z-\theta)\frac{\partial \theta}{\partial w_j} - v^i_{x^j y}(\pmb{w}^{-i},w^i-\theta,z-\theta)\frac{\partial \theta}{\partial w_j}+ v_{x^j x^j}^i(\pmb{w}^{-i},w^i-\theta,z-\theta) \\
&=&v^i_{x^j x^j}(\pmb{w}^{-i},w^i-\theta,z-\theta)=v^i_{x^j x^j}(\pmb{q}).\label{eq:derivative}
\end{eqnarray*}
The second last equation holds since $v_{x^i x^j}^i+v_{x^j y}^i = 0$ on $\overline{\mathcal{W}}_i\cap \overline{\mathcal{A}}_i^{+}$.  }

{ Therefore when $\pmb{p}=(\pmb{w},z)\in \mathcal{A}_i^{+}\cap E^{+}_{i,1}$,}
\begin{eqnarray*}
&&-\alpha v^i(\pmb{p}) +h\left(\frac{N-1}{N}\tilde{p}^i\right) +\frac{1}{2}\sum_{j=1}^N v^i_{x^j x^j}(\pmb{p})  \\
&=& \big(-\alpha v^i(\pmb{q}) +h\left(\frac{N-1}{N}\tilde{q}^i\right) +\frac{1}{2}\sum_{j=1}^N v^i_{x^j x^j}(\pmb{q})\big) + h\left(\frac{N-1}{N}\tilde{p}^i\right) - h\left(\frac{N-1}{N}\tilde{q}^i\right)\\
&>&-\alpha v^i(\pmb{q}) +h\left(\frac{N-1}{N}\tilde{q}^i\right) +\frac{1}{2}\sum_{j=1}^N v^i_{x^j x^j}(\pmb{q}), 
\end{eqnarray*}
in which $\tilde{q}^i = q^i-\frac{\sum_{j=1,j\neq i}^N q^j}{N-1}$ and $\tilde{p}^i = p^i-\frac{\sum_{j=1,j\neq i}^N p^j}{N-1}=\tilde{w}^i$.
The last inequality holds since $\tilde{p}^i>\tilde{q}^i>0$ and $h$ is convex and symmetric to $0$. Now for $\pmb{q}\in \partial E_i^+$, we have $-\alpha v^i(\pmb{q}) +h\left(\frac{N-1}{N}\tilde{q}^i\right) +\frac{1}{2}\sum_{j=1}^N v^i_{x^j x^j}(\pmb{q})  = 0.$  
Therefore,
$
-\alpha v^i +h\left(\frac{N-1}{N}\tilde{x}^i\right) +\frac{1}{2}\sum_{j=1}^N v^i_{x^jx^j} \geq 0$
for $\pmb{p}:=(\pmb{w},z)\in \mathcal{W}_i \cap {E}_{i,1}^{+}$. {When $\pmb{p}:=(\pmb{w},z)\in \mathcal{A}_i^{+}\cap {E}_{i,2}^{+}$, we have $v^i(\pmb{p})=v^i(\pmb{q})$ where $\pmb{q}:=\left(\pmb{w}^{-i},w^i-z,0\right)=\Pi(\pmb{p})$  translates  $\pmb{p}$ to $\{(\pmb{x},y) \in \mathbb{R}^N \times \mathbb{R}_{+}\,\, \vert \,\,y=0\}$ along the direction $(0,0,\cdots,-1,0,\cdots,-1)\in \mathbb{R}^{N+1}$.  In this case,  $v^i(\pmb{p}) =  p_N(\widetilde{w}^i-z) + A_N(0) \cosh\left( (\widetilde{w}^i -z)\sqrt{\frac{2(N-1) \alpha}{N} }\right) $  by definition. Hence $-\alpha v^i(\pmb{p}) +h\left(\frac{N-1}{N}\tilde{p}^i\right) +\frac{1}{2}\sum_{j=1}^N v^i_{x^j x^j}(\pmb{p}) =0$ holds by straightforward calculation.} Similar analysis holds for $\pmb{p}:=(\pmb{w},z)\in \mathcal{A}_i^{-}$.\\

\noindent {\bf Step 2} is to show 
\begin{eqnarray}\label{gradient_term}
v^i_{x^i} + v^i_{y} \leq 0,\,\,
 \text{and}\,\, -v^i_{x^i} + v^i_{y} \leq 0, \,\,\text{for}\,\, (\pmb{x},y)\in \overline{\mathcal{W}_{-i}}, \ \  \mbox{and} 
\end{eqnarray}
\begin{eqnarray}\label{gradient_term_2}
\begin{cases}
 v_{x^i}^i+v_{y}^i = 0,\,\, \text{for}\,\, (\pmb{x},y) \in \mathcal{A}_i^+ \\
-v_{x^i}^i+v_{y}^i = 0,\,\, \text{for}\,\, (\pmb{x},y) \in\mathcal{A}_i^-.
\end{cases}
\end{eqnarray}
Let us first check \eqref{gradient_term_2}.
When $\pmb{p}:=(\pmb{w},z)\in \mathcal{A}_{i}^{+}\cap E_{i,1}^{+}$, denote $\pmb{q}:=\left(\pmb{w}^i,w_{+}^i+\frac{\sum_{k \neq i}w^k}{N-1},f_N(w_+^i)\right)=\Pi(\pmb{p})$ which translate $\pmb{p}$ to the boundary of ${E}_i^{+}$, i.e., $\partial {E}_i^{+} :=\{(\pmb{x},y)\,\,|\,\,y=f_N\left(\tilde{x}^i\right)\}$  along the direction $(0,0,\cdots,-1,0,\cdots,-1)\in \mathbb{R}^{N+1}$. Then by the definition of \eqref{eq:valueN},  $v^i(\pmb{p}) = v^i(\pmb{q}) = u^i(\pmb{q})$, $
v^i_{x^i}(\pmb{p}) = \frac{1}{1-f_N^{\prime}(w_+^i)}v^i_{x^i}(\pmb{q}) + \frac{f_N^{\prime}(w_+^i)}{1-f_N^{\prime}(w_+^i)} v^i_{y}(\pmb{q}),$
and
$v^i_{y}(\pmb{p}) = -\frac{1}{1-f_N^{\prime}(w_+^i)}v^i_{x^i}(\pmb{q}) - \frac{f_N^{\prime}(w_+^i)}{1-f_N^{\prime}(w_+^i)} v^i_{y}(\pmb{q})$. 
Therefore, $v^i_{x^i}(\pmb{p})+ v^i_{y}(\pmb{p}) =0.$ 
{ When $\pmb{p}:=(\pmb{w},z)\in \mathcal{A}_i^{+}\cap {E}_{i,2}^{+}$, we have $v^i(\pmb{p})=v^i(\pmb{q})$ where $\pmb{q}:=\left(\pmb{w}^{-i},w^i-z,0\right)=\Pi(\pmb{p})$  translates  $\pmb{p}$ to $\{(\pmb{x},y) \in \mathbb{R}^N \times \mathbb{R}_{+}\,\, \vert \,\,y=0\}$ along the direction $(0,0,\cdots,-1,0,\cdots,-1)\in \mathbb{R}^{N+1}$. 
In this case,  $v^i(\pmb{p}) =  p_N(\widetilde{w}^i-z) + A_N(0) \cosh\left( (\widetilde{w}^i -z)\sqrt{\frac{2(N-1) \alpha}{N} }\right) $  by definition. Then $v^i_{x^i}(\pmb{p})+ v^i_{y}(\pmb{p}) =0$ holds by straightforward calculations.  }
Similarly,  $-v_{x^i}^i+v_{y}^i = 0$ for $(\pmb{x},y) \in \mathcal{A}_{i}^-$. As for \eqref{gradient_term}, by symmetry it suffices  to check the first inequality for $0\le \tilde{x}^i\leq f^{-1}_N(y)$. In this case,
 \begin{eqnarray*}
 &&v_{y}^i +v_{x^i}^i =A^{\prime}_N(y)\cosh\left(\tilde{x}^i \sqrt{\frac{2(N-1)\alpha}{N}}\right) +p_N^{\prime} (\tilde{x}^i) +A_N(y)\sinh\left(\tilde{x}^i \sqrt{\frac{2(N-1)\alpha}{N}}\right)\sqrt{\frac{2(N-1)\alpha}{N}}\\
 & = &  p^{\prime}_N(\tilde{x}^i)\left( 1- \cosh\left(\left(f^{-1}_N (y) - \tilde{x}^i \right)\sqrt{\frac{2(N-1)\alpha}{N}}\right)\right) +p_N^{\prime \prime}(f^{-1}_N(y))\sqrt{\frac{N}{2(N-1)\alpha}} \times \\
 &&\left[ \sinh\left(\left(f^{-1}_N (y) - \tilde{x}^i \right)\sqrt{\frac{2(N-1)\alpha}{N}}\right)-\frac{p^{\prime}_N(f^{-1}_N(y))-p^{\prime}_N(\tilde{x}^i)}{p_N^{\prime \prime}(f^{-1}_N(y))\sqrt{\frac{N}{2(N-1)\alpha}}}\cosh\left(\left(f^{-1}_N (y) - \tilde{x}^i \right)\sqrt{\frac{2(N-1)\alpha}{N}}\right)\right]\\
 & \leq &  p^{\prime}_N(\tilde{x}^i)\left( 1- \cosh\left(\left(f^{-1}_N (y) - \tilde{x}^i \right)\sqrt{\frac{2(N-1)\alpha}{N}}\right)\right)\nonumber\\
 && +p_N^{\prime \prime}(f^{-1}_N(y))\sqrt{\frac{N}{2(N-1)\alpha}} \left[\sinh\left(\left(f^{-1}_N (y) - \tilde{x}^i \right)\sqrt{\frac{2(N-1)\alpha}{N}}\right)\right.\nonumber\\
 &&\qquad \left.- \left(\left(f^{-1}_N (y) - \tilde{x}^i \right)\sqrt{\frac{2(N-1)\alpha}{N}}\right)\cosh\left(\left(f^{-1}_N (y) - \tilde{x}^i \right)\sqrt{\frac{2(N-1)\alpha}{N}}\right)\right] \le 0.
 \end{eqnarray*}
{The second to the last inequality holds since $p^{\prime}_N$ is a concave function and $p^{\prime\prime}_N(f^{-1}_N(y))>0$. The last inequality holds since $p^{\prime}_N(\tilde{x}^i) \ge 0$, $|\tilde{x}^i|\leq f^{-1}_N(y)$, and $p^{\prime\prime}_N(f^{-1}_N(y))>0$.}

\noindent {\bf Step 3} is to check  
 \begin{eqnarray}
 \begin{cases}
  & -v^i_{y}-v^i_{x^j} = 0, 
           \hspace{190pt}   \mbox{for } (\pmb{x},y) \in \mathcal{A}^+_{j}, j \ne i, \\ 
& \displaystyle   -v^i_{y}+v^i_{x^j}= 0, 
           \hspace{190pt}   \mbox{for } (\pmb{x},y) \in \mathcal{A}^-_{j}, j \ne i.
 \end{cases}
 \end{eqnarray}
By symmetry it is sufficient to check the first gradient condition.
 When $\pmb{p}:=(\pmb{w},z)\in\mathcal{A}^+_{j}\cap E^+_{j,1}$, denote $\pmb{q}:=\left(\pmb{w}^j,w_{+}^j+\frac{\sum_{k \neq j}w^k}{N-1},f_N(w_+^j)\right)=\Pi(\pmb{p})$ which translates $\pmb{p}$ to the boundary of ${E}_j^{+}$, i.e.,  $\partial {E}_j^{+} :=\{(\pmb{x},y)\,\,|\,\,y=f_N^{-1}\left(\tilde{x}^j\right)\}$ along the direction $(0,0,\cdots,-1,0,\cdots,-1)\in \mathbb{R}^{N+1}$ with all components  zero except the $j$-th and $(N+1)$-th components being $-1$. Then by the definition of \eqref{eq:valueN}, we have $v^i(\pmb{p}) = v^i(\pmb{q})$, 
$v^i_{x^j}(\pmb{p}) = \frac{1}{1-f_N^{\prime}(\tilde{q}^j)}v^i_{x^j}(\pmb{q}) + \frac{f_N^{\prime}(\tilde{q}^j)}{1-f_N^{\prime}(\tilde{q}^j)} v^i_{y}(\pmb{q}),$
and $v^i_{y}(\pmb{p}) = -\frac{1}{1-f_N^{\prime}(\tilde{q}^j)}v^i_{x^j}(\pmb{q}) - \frac{f_N^{\prime}(\tilde{q}^j)}{1-f_N^{\prime}(\tilde{q}^j)} v^i_{y}(\pmb{q}) $ where  $\tilde{q}^i = q^i-\frac{\sum_{j=1,j\neq i}^N q^j}{N-1}$.
Therefore, $v^i_{x^j}(\pmb{p})+ v^i_{y}(\pmb{p}) =0.$ { When $\pmb{p}:=(\pmb{w},z)\in \mathcal{A}_j^{+}\cap {E}_{j,2}^{+}$, we have $v^i(\pmb{p})=v^i(\pmb{q})$ where $\pmb{q}:=\left(\pmb{w}^{-j},w^j-z,0\right)=\Pi(\pmb{p})$  translates  $\pmb{p}$ to $\{(\pmb{x},y) \in \mathbb{R}^N \times \mathbb{R}_{+}\,\, \vert \,\,y=0\}$  along the direction $(0,0,\cdots,-1,0,\cdots,-1)\in \mathbb{R}^{N+1}$.  In this case,  $v^i(\pmb{p}) =  p_N(\widetilde{w}^j-z) + A_N(0) \cosh\left( (\widetilde{w}^j -z)\sqrt{\frac{2(N-1) \alpha}{N} }\right) $ holds by definition, and $v^i_{x^j}(\pmb{p})+ v^i_{y}(\pmb{p}) =0$ by straightforward calculations.  }
}
\item[(iii)] By the construction of Case 1 and Case 2, when $(\pmb{x},y)\notin \overline{\mathcal{W}_{-i}}$, there is a push at time $0$ to move the joint position to some point $(\hat{\pmb{x}},\hat{y})\in \partial \overline{\mathcal{W}_{-i}}$ such that $\Delta Y_0 \leq y$. when $(\pmb{x},y)\in \overline{\mathcal{W}_{-i}}$, $(\pmb{\xi}^{-i*},\xi^i)$ forms a solution to the Skorokhod problem in $\cap_{j \neq i} (E_j^{-}\cup E_j^{+})^c$. It is easy to verify that $\cap_{j \neq i} (E_j^{-}\cup E_j^{+})^c \subset \mathcal{W}_{-i}$ and the Skorokhod problem with $\cap_{j \neq i} (E_j^{-}\cup E_j^{+})^c$ has a weak solution. {When the fuel is used up, the dynamics $\pmb{X}_t$ will become uncontrolled and move freely without control.} Therefore condition $(iii)$ is satisfied.

\item[(iv)] Solution \eqref{eq:valueN} satisfies the smooth-fit principle in Section \ref{s31}, therefore, $v^i \in \mathcal{C}^2(\overline{\mathcal{W}}_{-i})$. 
{ Let us define a two-dimensional auxiliary function $$\widetilde{v}(x,y) = p_N(x)+A_N(y)\cosh\left(x\sqrt{\frac{2(N-1)\alpha}{N}}\right).$$
We first show that $\widetilde{v}(x,y)$ is convex when $|x| \leq f^{-1}_N(y)$ and  then  show  that $v^i(\pmb{x},y) $ defined in \eqref{eq:valueN} is convex in $\overline{\mathcal{W}}_{-i}$. 

\noindent{\bf Step 1} is to show that $\widetilde{v}(x,y)$ is convex when $|x| \leq f^{-1}_N(y)$. By straightforward calculation,
$
\widetilde{v}_{xx}(x,y) = p^{\prime \prime}_N(x)+\frac{2(N-1)\alpha}{N}\,A_N(y)\cosh\left(x\sqrt{\frac{2(N-1)\alpha}{N}}\right)$, $
\widetilde{v}_{xy}(x,y) = \sqrt{\frac{2(N-1)\alpha}{N}}\,A^{\prime}_N(y)\sinh\left(x\sqrt{\frac{2(N-1)\alpha}{N}}\right)$, and $\widetilde{v}_{yy}(x,y) = A^{\prime\prime}_N(y)\cosh\left(x\sqrt{\frac{2(N-1)\alpha}{N}}\right).$
When $0\leq x<f_N^{-1}(y)$, plugging \eqref{smoothfit} into the formula for $\widetilde{v}_{xx}(x,y)$ we have, 
\begin{eqnarray*}
\widetilde{v}_{xx}(x,y)  &=&   p^{\prime \prime}_N(x)+\,p_N^{\prime}  (f_N^{-1}(y))\sqrt{\frac{2(N-1)\alpha}{N}} \sinh  \left( f_N^{-1}(y)\sqrt{\frac{2(N-1) \alpha}{N} } \right) \, \cosh\left(x\sqrt{\frac{2(N-1)\alpha}{N}}\right) \\
&& -  p_N^{\prime\prime} (f_N^{-1}(y)) \cosh \left( f_N^{-1}(y) \sqrt{\frac{2(N-1) \alpha}{N} } \right) \, \cosh\left(x\sqrt{\frac{2(N-1)\alpha}{N}}\right).
\end{eqnarray*}
{ Given Lemma \ref{lemma:pN}, $p_N^{\prime}(x)$ is concave when $x>0$. Therefore for $y \ge 0$, }
$$p_N^{\prime}  (f_N^{-1}(y)) \ge p_N^{\prime}  (0) +p_N^{\prime\prime}  (f_N^{-1}(y))  (f_N^{-1}(y)-0) = p_N^{\prime\prime} ( f_N^{-1}(y))  f_N^{-1}(y).$$
The last equality holds since $h^{\prime}(0)=0$ from Assumption {\bf H2'}.
Combining the fact that $\sinh(z) \ge 0$ and $\cosh(z) \ge 0$ when $z \ge 0$, we have

\begin{eqnarray}
\widetilde{v}_{xx}(x,y)  &\ge&   p_N^{\prime\prime}  (f_N^{-1}(y))  f_N^{-1}(y)\sqrt{\frac{2(N-1)\alpha}{N}} \sinh  \left( f_N^{-1}(y)\sqrt{\frac{2(N-1) \alpha}{N} } \right) \, \cosh\left(x\sqrt{\frac{2(N-1)\alpha}{N}}\right) \nonumber\\
&& + p^{\prime \prime}_N(x)-  p_N^{\prime\prime} (f_N^{-1}(y)) \cosh \left( f_N^{-1}(y) \sqrt{\frac{2(N-1) \alpha}{N} } \right) \, \cosh\left(x\sqrt{\frac{2(N-1)\alpha}{N}}\right) \nonumber\\
&\ge& p^{\prime \prime}_N(x)+ p_N^{\prime\prime}  (x) \cosh\left(x\sqrt{\frac{2(N-1)\alpha}{N}}\right) \times \nonumber \\
&&\quad \left( x\sqrt{\frac{2(N-1)\alpha}{N}}\sinh  \left( x\sqrt{\frac{2(N-1) \alpha}{N} } \right) -  \cosh \left(x\sqrt{\frac{2(N-1) \alpha}{N} } \right) \right)\label{eq:bound1}\\
&=&p^{\prime \prime}_N(x) \left[  1+ z\sinh(z)\cosh(z)-\cosh^2(z)\right]\Bigg|_{z = x\sqrt{\frac{2(N-1)\alpha}{N}}} \ge 0\label{eq:bound2}
\end{eqnarray}
\eqref{eq:bound1} holds since $p_N^{\prime\prime}$ is non-increasing {(Lemma \ref{lemma:pN})} and $g_1(z):= z\sinh(z)-\cosh(z)$ is non-decreasing when $z \ge 0$. \eqref{eq:bound2} holds since $g_2(z) :=1+ z\sinh(z)\cosh(z)-\cosh^2(z)$ is non-negative when $z \ge 0$. To see this, $g_2(0)=0$ and 
\begin{eqnarray*}
g_2^{\prime} (z) = \cosh(z)[z\cosh(z)-\sinh(z)] + z \sinh^2(z) \ge 0, \mbox{ when } z\ge 0.
\end{eqnarray*}
On the other hand, denote 
$ g_3(z) := - p_N^{\prime}(z) \cosh  \left(z \sqrt{\frac{2(N-1) \alpha}{N} } \right) + p_N^{\prime\prime} (z)\sqrt{\frac{N}{2(N-1) \alpha}} \sinh \left( z \sqrt{\frac{2(N-1) \alpha}{N} } \right)$, then
$g_3^{\prime}(z) = -  \sqrt{\frac{2(N-1) \alpha}{N} }p_N^{\prime}(z) \sinh  \left(z \sqrt{\frac{2(N-1) \alpha}{N} } \right) +  p_N^{\prime\prime\prime} (z)\sqrt{\frac{N}{2(N-1) \alpha}} \sinh \left( z \sqrt{\frac{2(N-1) \alpha}{N} } \right).$
{
From Lemma \ref{lemma:pN}, we have $p_N^{\prime}(z) \ge 0$ and $p_N^{\prime\prime\prime}(z) \leq 0$ when $z \ge 0$, and hence $g_3^{\prime}(z)\leq 0$ when $z \ge 0$. Along with the fact that $f^{'}_N(z)<0$ when $z>0$ from Lemma \ref{lemma:f},} we have 
$A_N^{\prime\prime}(y) = g_3^{\prime}(f_N^{-1}(y))  \frac{1}{f_N^{\prime}(f_N^{-1}(y))} \ge 0.$
Therefore $\widetilde{v}_{yy}(x,y)  \ge 0$. Finally we show that $\widetilde{v}_{xx} \, \widetilde{v}_{yy}  - (\widetilde{v}_{xy} )^2 \ge 0$ when $0 \leq x \leq f^{-1}_N(y)$. To see this, denote $z=f^{-1}_N(y)$,
\begin{eqnarray*}
&&\widetilde{v}_{xx} \, \widetilde{v}_{yy}  - (\widetilde{v}_{xy} )^2= \Big (p^{\prime \prime}_N(x)+\frac{2(N-1)\alpha}{N}\,A_N(y)\cosh\left(x\sqrt{\frac{2(N-1)\alpha}{N}}\right)\Big)
\Big(A^{\prime\prime}_N(y)\cosh\left(x\sqrt{\frac{2(N-1)\alpha}{N}}\right)\Big)\\
 &&- \Big(\sqrt{\frac{2(N-1)\alpha}{N}}\,A^{\prime}_N(y)\sinh\left(x\sqrt{\frac{2(N-1)\alpha}{N}}\right) \Big)^2\\
&=&\frac{2(N-1)\alpha}{N}\left(- p_N^{\prime} \cosh  \left(z \sqrt{\frac{2(N-1) \alpha}{N} } \right) + p_N^{\prime\prime} \sqrt{\frac{N}{2(N-1) \alpha}} \sinh \left( z \sqrt{\frac{2(N-1) \alpha}{N} } \right)\right)\times\\
&&\left(p_N^{\prime} \cosh  \left(x \sqrt{\frac{2(N-1) \alpha}{N} } \right)- p_N^{\prime} \cosh  \left(z \sqrt{\frac{2(N-1) \alpha}{N} } \right)  \right) \ge 0.
\end{eqnarray*}
Similar result holds when $-f_N^{-1}(y) \leq x <0$ by symmetry.

\noindent{\bf Step 2} is to  show  that $v^i(\pmb{x},y) $ defined in \eqref{eq:valueN} is convex in $\overline{\mathcal{W}_{-i}}$. We take player one as an example to show $v^1(\pmb{x},y)=\tilde{v}(\tilde{x}_1,y)$ is convex in $\overline{\mathcal{W}_{-1}}$ where $\tilde{x}_1 =x_1 -\frac{\sum_{k=2}^N x_k}{N-1}$. The convexity of other players' value functions can be verified similarly. When $(\pmb{x},y) \in \overline{\mathcal{W}_{-1}}$, we have $|\tilde{x}_1| \leq y$ hence $\widetilde{v}(\tilde{x}_1,y)$ is non-negative definite. By  chain rule,  for $2 \leq k \neq j \leq N$,
\begin{eqnarray*}
&v^1_{x_1 x_1} (\pmb{x},y) = \widetilde{v}_{x x} (\tilde{x}_1,y), \quad v_{x_1 x_k} (\pmb{x},y) = - \frac{1}{N-1}  \widetilde{v}_{x x} (\tilde{x}_1,y), \quad v_{x_1 y} (\pmb{x},y) =   \widetilde{v}_{x y} (\tilde{x}_1,y),\quad v_{y y} (\pmb{x},y) =   \widetilde{v}_{y y} (\tilde{x}_1,y),\\
&v^1_{x_k x_j} (\pmb{x},y) = \frac{1}{(N-1)^2} \widetilde{v}_{x x} (\tilde{x}_1,y), \quad v_{x_1 x_k} (\pmb{x},y) = - \frac{1}{N-1}  \widetilde{v}_{x x},\quad v_{x_k y} (\pmb{x},y) = -\frac{1}{N-1}  \widetilde{v}_{x y} (\tilde{x}_1,y).\\
\end{eqnarray*}

Denote $H(\pmb{x},y) := \nabla^2 v^1(\pmb{x},y) \in \mathbb{R}^{(N+1) \times (N+1)}$ as the Hessian matrix of $v^1$ at some point $(\pmb{x},y) \in \overline{\mathcal{W}_{-1}}$. Then for any $\pmb{d} = (b_1,\cdots,b_{N},c)\in  \mathbb{R}^{N+1}$, 
\begin{eqnarray*}
\pmb{d}^TH(\pmb{x},y) \pmb{d}
&=& \left(b_1-\frac{1}{N-1}\sum_{k=2}^Nb_k\right)^2 \widetilde{v}_{x x} +2 \left(b_1-\frac{1}{N-1}\sum_{k=2}^Nb_k\right)c \widetilde{v}_{x y} + c^2 \widetilde{v}_{y y} = \pmb{e}^T \tilde{H}(\tilde{x}_1,y)\pmb{e} \ge 0,
\end{eqnarray*}
where $\pmb{e} = \left(b_1-\frac{1}{N-1}\sum_{k=2}^Nb_k,c \right)$ and $ \tilde{H}(\tilde{x}_1,y) = \nabla^2 \tilde{v}(\tilde{x}_1,y)$. The last inequality follows from the convexity of $ \tilde{v}(\tilde{x}_1,y)$  when $|\tilde{x}_1| \leq y$. Therefore $v^1$ is convex in $\overline{\mathcal{W}_{-1}}$. 
}
\item[(v)] 
{
Denote $\mathcal{W}_{-i}(y) = \{(\pmb{x},z):(\pmb{x},z)\in \mathcal{W}_{-i}\,\, {\rm and } \,\,z \le y\}$.  $(\pmb{X}^{-i*}_t,X_t^i,Y_t)\in \overline{\mathcal{W}_{-i}(y)}$ holds a.s. when $(\xi_t^{-i*},\xi_t^i)\in \mathcal{S}_N(\pmb{x},y)$. This is because $0 \leq Y_t \leq y$ a.s.  $\forall t \ge 0$ under $(\xi_t^{-i*},\xi_t^i)\in \mathcal{S}_N(\pmb{x},y)$.
First, we show that $v_{x_j}^i$ is bounded for $(\pmb{x},z) \in E_{i,1}^{+}\cap \overline{\mathcal{W}_{-i}(y)} $, $(\pmb{x},z) \in E_{i,1}^{-}\cap \overline{\mathcal{W}_{-i}(y)} $ and  $(\pmb{x},z)\in B(y):=\overline{\mathcal{W}_{-i}(y)} \cap \{(\pmb{x},z): |\tilde{x}^i| \le f_N^{-1}(z)\}$. For $(\pmb{x},z)\in B(y)$, $|\tilde{x}^i| \leq f^{-1}_N(z) \leq  f^{-1}_N(y)<\infty$ since $f^{-1}_N$ is non-increasing. This implies that $\tilde{x}^i$ is bounded in $B(y)$. By the definition of  $A_N(z)$ in \eqref{eq:candidateA}, $A_N(z)$ is bounded in $B(y)$. Hence  $v_{x_k}^i$ is bounded on $B(y)$  $(k=1,2,\cdots,N)$. Following Step 2 in (ii), there exists $\pmb{q}\in \partial B(y)$ such that $v_{x^k}(\pmb{q}) = v_{x^k}(\pmb{x},z)$ ($k=1,2,\cdots,N$)  for $(\pmb{x},z) \in E_{i,1}^{+}\cap \overline{\mathcal{W}_{-i}(y)}$. Similar result holds for $(\pmb{x},z) \in E_{i,1}^{-}\cap \overline{\mathcal{W}_{-i}(y)}$.
Hence $v_{x_k}^i$ is bounded on $(\pmb{x},z) \in E_{i,1}^{+}\cap \overline{\mathcal{W}_{-i}(y)}$ and $(\pmb{x},z) \in E_{i,1}^{-}\cap \overline{\mathcal{W}_{-i}(y)}$.
Second,   $v^i(\pmb{x},0) = p_N(\tilde{x}^i)$ holds since $A_N(0)=0$ (Lemma \ref{lemma:f}).  By the  definition of $v^i$ and following  Step 2 in (ii), we have $v^i_{x^k}(\pmb{x},z)=v^i_{x^k}((\pmb{x}^{-i},x^i-z),0)$ ($k=1,2,\cdots,N$) and $0<\tilde{x}^i-z <\tilde{x}^i$ for $(\pmb{x},z)\in E_{i,2}^{+}\cap \overline{\mathcal{W}_{-i}(y)}$. From Lemma \ref{lemma:pN}, $0\leq p_N^{\prime}(\tilde{x}^i-z) \leq p_N^{\prime}(\tilde{x}^i)$. Hence $|v^i_{x^k}(\pmb{x},z)|\leq |p^{\prime}_N(\tilde{x}^i)|$ for $(\pmb{x},z)\in E_{i,1}^{+}\cap \overline{\mathcal{W}_{-i}(y)}$ and the same result holds for  $(\pmb{x},z)\in E_{i,2}^{-}\cap \overline{\mathcal{W}_{-i}(y)}$.
 Combine above analysis with Lemma \ref{lemma:pN}, there exists a constant $C(y)>0$ such that
 $|v_{x^j}^i(\pmb{x},z)| \leq C(y) + |p^{\prime}_N(\tilde{x}^i)|\leq C(y) +\frac{K}{\alpha} |\tilde{x}^i|$ for $(\pmb{x},z) \in \overline{\mathcal{W}_{-i}(y)}$. Hence  by Tonelli's Theorem, $\mathbb{E}\left[\int_0^Te^{-2\alpha t}(v_{x^j}^i\left(\pmb{X}_t^{-i*},X_t^i,Y_t)\right)^2dt\right]\leq C_0 \left(C^2(y) + (x^i-\frac{\sum_{j \neq i}x_j}{N-1})^2 + y^2 + T\right)<\infty$ for some $C_0>0$ and (v) is satisfied.
}
\item[(vi)] { Recall the definition of  ${\mathcal{W}}_{-i}(y)$ in (v) and  the fact that $(\pmb{X}^{-i*}_t,X_t^i,Y_t)\in \overline{\mathcal{W}_{-i}(y)}$ when $(\xi_t^{-i*},\xi_t^i)\in \mathcal{S}_N(\pmb{x},y)$.
Following the same argument as in (v), there exists $\widetilde{C}(y)>0$ such that $|v^i(\pmb{x},z)|\leq \widetilde{C}(y)$ for $(\pmb{x},z) \in E_{i,1}^{+}\cap \overline{\mathcal{W}_{-i}(y)} $, $(\pmb{x},z) \in E_{i,1}^{-}\cap \overline{\mathcal{W}_{-i}(y)} $ and  $(\pmb{x},z)\in B(y):=\overline{\mathcal{W}_{-i}(y)} \cap \{(\pmb{x},z): |\tilde{x}^i| \le f_N^{-1}(z)\}$.
In addition,  $v^i(\pmb{x},0) = p_N(\tilde{x}^i)$ holds since $A_N(0)=0$ (Lemma \ref{lemma:f}).  By the definition of $v^i$,  $v^i(\pmb{x},z)=v^i((\pmb{x}^{-i},x^i-z),0)$ and $0<\tilde{x}^i-z <\tilde{x}^i$ for $(\pmb{x},z)\in E_{i,2}^{+}\cap\overline{\mathcal{W}_{-i}(y)}$. From Lemma \ref{lemma:pN}, $0\leq p_N(\tilde{x}^i-z) \leq p_N(\tilde{x}^i)$. Hence $v^i(\pmb{x},z)\leq p_N(\tilde{x}^i)$ for $(\pmb{x},z)\in E_{i,2}^{+}\cap \overline{\mathcal{W}_{-i}(y)}$ and the same result holds for $(\pmb{x},z)\in E_{i,2}^{-}\cap \overline{\mathcal{W}_{-i}(y)}$.
Combine above analysis with Lemma \ref{lemma:pN},
$|v(\pmb{x},{y})|\leq p_N(\tilde{x}^i) + \widetilde{C}(y) \leq p_N(0) + \frac{K}{\alpha}(\tilde{x}^i)^2+ \widetilde{C}(y)$.
Given $(\pmb{\xi}^{-i*}, \xi^i) \in \mathcal{S}_N(\pmb{x},\pmb{y})$, $\sum_{j\neq i} \check{\xi}_T^{j*} +\check{\xi}_T^{i} \leq y$ holds a.s..
 Therefore
$ \mathbb{E}\left[\left(X_T^i-\frac{\sum_{j \neq i}X_T^{j*}}{N-1}\right)^2\right] \leq  \tilde{C}_0\left( \left(x_0^i-\frac{\sum_{j \neq i}x_j}{N-1}\right)^2 + y^2 +  T\right)$ for some $ \tilde{C}_0>0$. Hence $\underset{T \rightarrow \infty}{\lim \sup} \, e^{-\alpha T} \mathbb{E}\left[p_N\left(X_T^i-\frac{\sum_{j \neq i}X_T^{j*}}{N-1}\right)\right] = 0$ and the transversality condition (vi) holds.}
\item[(vii)] {This condition is satisfied by the property of the Skorokhod problem and  the initial jump described in Section \ref{s33}.}
\end{itemize}
\end{proof}


\section{Nash Equilibrium For  Game $\pmb{C_d}$}
\label{s4}
\label{section:gamedivide}
\quad In this section, we study the NEP of the $N$-player game $\pmb{C_d}$. That is $A=\pmb{I_N}\in \mathbb{R}^{N\times N}$, and  
\begin{equation}
\label{eq:YB}
Y^i_t = y^i -  \check{\xi}^i_t \quad \mbox{with} \quad Y^i_{0-} = y^i.
\end{equation}

Recall that the major difference between  game $\pmb{C_p}$ and  game $\pmb{C_d}$ is that, in the former all $N$ players share a fixed amount of the same resource, while in the latter each player has her own individual fixed resource constraint.
This difference is reflected in $(HJB-C_p)$ and $(HJB-C_d)$ in terms of  their dimensionality, and in each player's control based on the remaining resources. In particular, $(HJB-C_p)$ and the state space $(\pmb{x}, y)$ of $\pmb{C_p}$ are of dimension $N+1$, whereas $(HJB-C_d)$ and the state space $(\pmb{x}, \pmb{y})$ of $\pmb{C_d}$ are of dimension $2N$. Moreover, in game $\pmb{C_p}$, 
 the gradient constraint is $-v^i_{y}\pm v^i_{x^i}$ for player $i$. In contrast, in  game $\pmb{C_d}$, each player controls her own resource level, the gradient constraint becomes $-v^i_{y^i}\pm v^i_{x^i}$ for player $i$. 
So if $\mathcal{A}_i \cap \mathcal{A}_j = \emptyset $, the HJB equation for $v^i(\pmb{x}, \pmb{y})$ in game $\pmb{C}_d$ is as follows.
 \begin{eqnarray*}
    \textit{(HJB-$C_d$)} \left\{
                \begin{array}{ll}
                \displaystyle \min \left\{-\alpha v^i +h \left(\frac{N-1}{N} \widetilde{x}^i \right) + \frac{1}{2} \sum_{j=1}^N v^i_{x^jx^j}, -v^i_{y^i}+v^i_{x^i}, -v^i_{y^i}-v^i_{x^i}\right\} = 0, \\ [3 pt]
          \hspace{261pt}   \mbox{for } (\pmb{x},\pmb{y}) \in \mathcal{W}_{-i},  \\ [3 pt]
 \displaystyle  {-v^i_{y^j}-v^i_{x^j} = 0, 
           \hspace{190pt}   \mbox{for } (\pmb{x},y) \in \mathcal{A}^+_{j}, j \ne i, }\\  [3 pt]
 \displaystyle  { -v^i_{y^j}+v^i_{x^j}= 0, 
           \hspace{190pt}   \mbox{for } (\pmb{x},y) \in \mathcal{A}^-_{j}, j \ne i.}                \end{array}
              \right.
\end{eqnarray*}
Note that the control policy of the $i^{th}$ player only depends on $(\pmb{x},y^i)$ in $\mathcal{W}_{-i}$. 
As seen in Section \ref{section:gamepooling}, for the controlled process of type {\bf $\pmb{C_p}$},
upon hitting the boundary of the polyhedron, the polyhedron will expand in all directions.
While for the controlled process of type {\bf $\pmb{C_d}$}, only one direction of the the polyhedron will move once hit.

\quad To proceed, similar to Section \ref{section:gamepooling}, define the action region $\mathcal{A}_i {\in \mathbb{R}^N\times\mathbb{R}^N_{+} }$ and the waiting region $\mathcal{W}_i$ of the $i^{th}$ player by
\begin{equation}
\label{eq:AWB}
\mathcal{A}^+_i := E_i^{+}  \cap Q_i,\quad \mathcal{A}^-_i := E_i^{-} \cap Q_i \quad,\mathcal{A}_i = \mathcal{A}^+_i\cup \mathcal{A}^-_i, \mbox{and} \quad  \mathcal{W}_i:= {\mathbb{R}^{N} \times \mathbb{R}^{N}_+}\setminus \mathcal{A}_i,
\end{equation}
where
\begin{equation}
\label{eq:thresholdB}
E_i^{+} : = \left\{(\pmb{x}, \pmb{y}) \in \mathbb{R}^{N} \times (\mathbb{R}_{+}^*)^N: \widetilde{x}^i \ge f_N^{-1}(y^i) \right\}, \, E_i^{-} : = \left\{(\pmb{x}, \pmb{y}) \in \mathbb{R}^{N} \times (\mathbb{R}_{+}^*)^N: \widetilde{x}^i \le - f_N^{-1}(y^i) \right\},
\end{equation}
with {
\begin{eqnarray}
E_{i,1}^{+} &:=& \left\{(\pmb{x}, \pmb{y}) \in E_{i}^{+} : y^i \ge \tilde{x}^i+x_0 \right\},\quad E_{i,2}^{+} := \left\{(\pmb{x}, \pmb{y}) \in E_{i}^{+} : y^i <  \tilde{x}^i+x_0 \right\},\label{eq:thresholdB21}\\
E_{i,1}^{-} &:=& \left\{(\pmb{x}, \pmb{y}) \in E_{i}^{-} : y^i \ge - \tilde{x}^i-x_0 \right\},\quad E_{i,2}^{-} := \left\{(\pmb{x}, \pmb{y}) \in E_{i}^{+} : y^i < - \tilde{x}^i-x_0 \right\},\label{eq:thresholdB22}
\end{eqnarray}
and  $\{Q_i\}_{i=1}^N$  convex partitions of $\mathbb{R}^{N}\times \mathbb{R}_+$ 
 { such that $Q_i \cap Q_j = (E_i^+\cup E_i^-)\cap (E_j^+\cup E_j^-)\cap \partial \mathcal{W}_{NE}$ for $i\neq j$,}
$\cup_{i=1}^N Q_i = \mathbb{R}^{N}\times \mathbb{R}_+$, and $\alpha \pmb{p} + (1-\alpha) \pmb{q} \in Q_j$ for all $\alpha \in [0,1]$ if $ \pmb{p}\in Q_j$ and $\pmb{q}\in Q_j$ for some $j=1,2,\cdots,N$.  We can define the following mapping
\begin{eqnarray}\label{eq:mapping_Pi2}
\Pi(\pmb{x},\pmb{y}) = 
\begin{cases}
\left((\pmb{x}^{-i}, x^i_{+} + \frac{\sum_{k \neq i}x^k}{N-1}),\big (\pmb{y}^{-i},f_N(x_{+}^i)\big)\right), \quad  & {\rm if } \quad (\pmb{x},\pmb{y})\in Q_i \cap E_{i,1}^{+},\\
\left((\pmb{x}^{-i},x^i-y^i),(\pmb{y}^{-i},0)\right),
  & {\rm if } \quad (\pmb{x},\pmb{y})\in Q_i \cap E_{i,2}^{+},\\
\left((\pmb{x}^{-i},\frac{\sum_{k \neq i}x^k}{N-1}+x_{-}^i),(\pmb{y}^{-i},\tilde{f}_N(x_-^i))\right), & {\rm if } \quad (\pmb{x},\pmb{y})\in Q_i \cap E_{i,1}^{-},\\
\left((\pmb{x}^{-i},x^i+y^i),(\pmb{y}^{-i},0)\right), &  {\rm if } \quad (\pmb{x},\pmb{y})\in Q_i \cap E_{i,2}^{-},
\end{cases}
\end{eqnarray}
in which the threshold function $f_N(\cdot)$ is defined in \eqref{eq:fNder}-\eqref{eq:intercept}, $x_{+}^i$ is the unique positive root such that $z-f_N(z) = \tilde{x}^i-y^i$ and $x_{-}^i$ is the unique negative root such that $z+\tilde{f}_N(z) = \tilde{x}^i+y^i$.
Note that, $\Pi(\cdot)$ translates $(\pmb{x},\pmb{y})$ to the boundary of ${E}_{i,1}^{+}$, i.e., $\partial{E}_{i,1}^{+}:=\{(\pmb{x},\pmb{y})\in\mathbb{R}^N \times\mathbb{R}^N_{+} \,\,:\,\,y^i=f_N^{-1}\left(\tilde{x}^i\right), 0<\tilde{x}^i \le x_0\}$ when $(\pmb{x},\pmb{y})\in Q_i \cap E_{i,1}^{+}$, and translates $(\pmb{x},\pmb{y})$ to the ``zero-resource'' plane $\{(\pmb{x},\pmb{y})\in\mathbb{R}^N \times\mathbb{R}^N_{+} \,\,:\,\,y^i=0\}$ when $(\pmb{x},\pmb{y})\in Q_i \cap E_{i,2}^{+}$, both along the direction $(0,\cdots,-1,0,\cdots,-1,\cdots,0)\in \mathbb{R}^{2N}$ with 
nonzero $i$-th and $(N+i)$-th components. 
} Let
\begin{align}
\label{eq:WNE}
\mathcal{W}_{NE} :&= \{(\pmb{x},\pmb{y}) \in \mathbb{R}^{N}\times \mathbb{R}_{+}^{N}: |\widetilde{x}^i| < f_N^{-1}(y^i) \mbox{ for } 1 \le i \le N\}  {\, \cup\, \{(\pmb{x},\pmb{y}) \in \mathbb{R}^{N}\times \mathbb{R}_{+}^{N}: \pmb{y}=0\}, }
\end{align}
 and assume {$\{Q_i\}_{i=1}^N$ satisfies the following assumption:
\begin{enumerate}[font=\bfseries,leftmargin=2cm]
\item[H3-${\bf C_d}$.] For any $(\pmb{x},\pmb{y}) \in \cup_i \mathcal{A}_i$, \quad $\Pi (\pmb{x},\pmb{y}) \in \overline{\mathcal{W}_{NE}}.$
\end{enumerate}
 Condition {\bf H3-${\bf C_d}$} implies that if $(\pmb{x},\pmb{y})\in \mathcal{A}_i$, then the dynamics will be in region $\overline{\mathcal{W}_{NE}}$ after player $i$'s control. 
}

We now investigate   control  of player $i$ which only depends on $(\pmb{x},y^i)$ in $\mathcal{W}_{-i}$. That is, for $|\widetilde{x}^i| < f_N^{-1}(y^i)$,
\begin{equation}
\label{eq:candidateB}
v^i(\pmb{x},\pmb{y}) =  p_N(\widetilde{x}^i) + A_N(y^i) \cosh\left(\widetilde{x}^i \sqrt{\frac{2(N-1) \alpha}{N}} \right),
\end{equation}
is a solution to $(\mbox{HJB-$C_d$})$, where $p_N(\cdot)$ is defined by \eqref{eq:pN}, and $A_N(\cdot)$ defined by \eqref{eq:AN}.

\quad The next step is to construct the controlled process $(\pmb{X}, \pmb{Y})$ corresponding to the HJB solution \eqref{eq:candidateB}.

\quad  Note that $\mathcal{W}_{NE}$ is an unbounded domain in $\mathbb{R}^{2N}$ with $2N$ boundaries. For $i=1,2,\cdots,N$, define the $2N$ faces of $\mathcal{W}_{NE}$
\begin{eqnarray*}
F_i = \{(\pmb{x},\pmb{y}) \in \partial \mathcal{W}_{NE}\,\,\,\vert\,\,\,  (\pmb{x},\pmb{y}) \in \partial E_i^{+}\},\quad
F_{i+N} = \{(\pmb{x},\pmb{y}) \in \partial \mathcal{W}_{NE}\,\,\,\vert \,\,\, (\pmb{x},\pmb{y}) \in \partial E_i^{-}\}.
\end{eqnarray*}
The normal direction on each face is given by
\begin{eqnarray*}
\pmb{n}_i &=& c_i \left(\frac{1}{N-1},\cdots,\frac{1}{N-1}-1,\frac{1}{N-1}\cdots,\frac{1}{N-1};0,\cdots,0,(f_N^{-1})^{\prime}\left(y^i\right),0,\cdots,0\right),\\
\pmb{n}_{N+i} &=& c_{N+i}\left(-\frac{1}{N-1},\cdots,-\frac{1}{N-1},1,-\frac{1}{N-1},\cdots,-\frac{1}{N-1};0,\cdots,0,(f_N^{-1})^{\prime}\left(y^i\right),0,\cdots,0\right).\\
\end{eqnarray*}
with the $i^{th}$ component to be $\pm 1$ and the $(N+i)^{th}$ component to be $(f_N^{-1})^{\prime}(y^i)$. $c_i$ and $c_{N+i}$ are normalizing constants such that $\|\pmb{n}_i\|=\|\pmb{n}_{N+i}\|=1$.

Denote the reflection direction on each face as
\begin{eqnarray*}
\pmb{r}_i = c_i^{\prime} \left(0,\cdots,0,-1,0,\cdots 0; 0,\cdots,0,-1,0,\cdots,0\right),
\pmb{r}_{N+i} =c^{\prime}_{N+i}\left(0,\cdots,0,1,0,\cdots 0; 0,\cdots,0,-1,0,\cdots,0\right),
\end{eqnarray*}
with the $i^{th}$ component to be $\pm 1$ and the $(N+i)^{th}$ component to be $\pm 1$. $c_i^{\prime}$ and $c^{\prime}_{N+i}$ are normalizing constants such that $\|\pmb{r}_i\|=\|\pmb{r}_{N+i}\|=1$.
The NE strategy is defined as follows.

{\bf Case 1:} $(\pmb{X}_{0-},\pmb{Y}_{0-}) = (\pmb{x},\pmb{y}) \in  \overline{\mathcal{W}_{NE}}$. One can check that $\mathcal{W}_{NE}$ defined in \eqref{eq:WNE} and $\{\pmb{r}_i\}_{i=1}^{2N}$ defined above satisfies assumptions \textbf{A1}-\textbf{A5}. Therefore, there exists a weak solution to the Skorokhod problem with data $\left(\mathcal{W}_{NE}, \{\pmb{r}_i\}_{i=1}^{2N},\pmb{b},\pmb{\sigma},\pmb{x}\in \overline{\mathcal{W}_{NE}}\right)$. See Appendix A for the satisfiability of \textbf{A1}-\textbf{A5}.

{\bf Case 2:} $(\pmb{X}_{0-},\pmb{Y}_{0-}) = (\pmb{x},\pmb{y}) \notin  \overline{\mathcal{W}_{NE}}$. 
There exists $i \in \{1, \cdots, N\}$ such that $(\pmb{X}_{0-}, \pmb{Y}_{0-}) \in \mathcal{A}_i$. {(1) If $(\pmb{x},\pmb{y}) \in \mathcal{A}_i^{+}\cap E_{i,1}^+$, then player $i$ will move immediately from $X_{0-}^i=x^i$ to $X_{0}^i=x^i_{+}+ \frac{\sum_{k \neq i}x^k}{N-1}$ at time $0$, where $x^i_{+}$ is the unique positive root such that $z - f_N(z) = \tilde{x}^i - y^i$. This will reduce of player $i$'s resource from $Y^i_{0-}=y^i$ to $Y^i_{0} = f_N(x^i_{+})\ge 0$. Other players' dynamics and resources remain unchanged, i.e., $X_{0}^j = X_{0-}^j=x^j$ and $Y_{0}^j = Y_{0-}^j=y^j$ for $j\neq i $ and $1 \leq j \leq N$. By Assumption \textbf{H3-$\bf C_{d}$}, we have $(\pmb{X}_0,\pmb{Y}_{0})=\left((\pmb{x}^{-i}, x^i_{+} + \frac{\sum_{k \neq i}x^k}{N-1}),\big (\pmb{y}^{-i},f_N(x_{+}^i)\big)\right) =\Pi((\pmb{X}_{0-},\pmb{Y}_{0-}))\in  \overline{\mathcal{W}_{NE}}$. (2) If $(\pmb{x},\pmb{y}) \in \mathcal{A}_i^{+}\cap E_{i,2}^+$, then player $i$ will move immediately from $X_{0-}^i=x^i$ to $X_{0}^i=x^i-y^i$ and her resource changes from $Y^i_{0-}=y^i$ to $Y^i_{0} = 0$ at time $0$. Other players' positions and resources remain unchanged, i.e., $X_{0}^j = X_{0-}^j=x^j$ and $Y_{0}^j = Y_{0-}^j=y^j$ for $j\neq i $ and $1 \leq j \leq N$. By Assumption \textbf{H3-$\bf C_{d}$}, we have $(\pmb{X}_0,\pmb{Y}_{0})=\left((\pmb{x}^{-i}, x^i-y^i,\big (\pmb{y}^{-i},0\big)\right)=\Pi((\pmb{X}_{0-},\pmb{Y}_{0-})) \in  \overline{\mathcal{W}_{NE}}$.
(3) Similarly,  if $(\pmb{x},\pmb{y}) \in \mathcal{A}_i^{-}\cap E_{i,1}^-$, then player $i$ will move immediately from $X_{0-}^i=x^i$ to $X_{0}^i=x^i_{-}+ \frac{\sum_{k \neq i}x^k}{N-1}$ at time $0$, where $x^i_{-}$ is the unique negative root such that $z+\tilde{f}_N(z) = \tilde{x}^i+y^i$. This will reduce her resource from $Y^i_{0-}=y$ to $Y^i_{0} = f_N(x^i_{-}) \ge 0$. Other players' dynamics and resources remain unchanged, i.e., $X_{0}^j = X_{0-}^j=x^j$ and $Y_{0}^j = Y_{0-}^j=y^j$  for $j \neq i$ and $1 \leq j \leq N$. By Assumption \textbf{H3-$\bf C_{d}$}, we have $(\pmb{X}_0,\pmb{Y}_{0})=\left((\pmb{x}^{-i},\frac{\sum_{k \neq i}x^k}{N-1}+x_{-}^i),(\pmb{y}^{-i},\tilde{f}_N(x_-^i))\right) =\Pi((\pmb{X}_{0-},\pmb{Y}_{0-}))\in  \overline{\mathcal{W}_{NE}}.$ (4) If $(\pmb{x},\pmb{y}) \in \mathcal{A}_i^{-}\cap E_{i,2}^-$, then player $i$ will move immediately from $X_{0-}^i=x^i$ to $X_{0}^i=x^i+y^i$ and her resource reduces from $Y^i_{0-}=y^i$ to $Y^i_{0} =0$ at time $0$. Other players' dynamics and resources remain unchanged, i.e., $X_{0}^j = X_{0-}^j=x^j$ and $Y_{0}^j = Y_{0-}^j=y^j$  for $j \neq i$ and $1 \leq j \leq N$. By Assumption \textbf{H3-$\bf C_{d}$}, we have $(\pmb{X}_0,\pmb{Y}_{0})=\left((\pmb{x}^{-i},x^i+y^i),(\pmb{y}^{-i},0)\right) =\Pi((\pmb{X}_{0-},\pmb{Y}_{0-}))\in  \overline{\mathcal{W}_{NE}}.$ 
}

\quad In summary,  the NE for the $N$-player game \eqref{eq:JA_game} with constraint    $\pmb{C_d}$ is stated as follows.
\begin{theorem}[NE for the $N$-player game  $\pmb{C_d}$]
\label{thm:NB}
Assume \textbf{H1$'$}-\textbf{H2$'$} {and \textbf{H3-$\bf C_{d}$}}. Define $u^i \in\mathbb{R}^N \times \mathbb{R}_{+} \rightarrow \mathbb{R}$ as {
\begin{equation}
\label{eq:valueN2-aux}
u^i(\pmb{x},{y}) = 
\left\{ \begin{array}{cll}
 p_N(\widetilde{x}^i) + A_N(y) \cosh\left( \widetilde{x}^i \sqrt{\frac{2(N-1) \alpha}{N} }\right) & \mbox{if } |\tilde{x}^i|\le f^{-1}_N(y), \mbox{ and } y=0, \\  [3 pt]
 u^i\left(\pmb{x}^{-i}, x^i_{+} + \frac{\sum_{k \neq i}x^k}{N-1}, f_N(x^i_{+})\right)  & \mbox{if }  \tilde{x}^i>f^{-1}_N(y) \mbox{ and } y \ge \tilde{x}^i+x_0, \\ [3pt]
  u^i\left(\pmb{x}^{-i}, x^i - y, 0\right)  & \mbox{if }  \tilde{x}^i>f^{-1}_N(y) \mbox{ and } y < \tilde{x}^i+x_0, \\ [3pt]
 u^i\left(\pmb{x}^{-i}, \frac{\sum_{k \neq i}x^k}{N-1} + x^i_{-}, \tilde{f}_N(x^i_{-})\right)  & \mbox{if }  \tilde{x}^i<-f^{-1}_N(y) \mbox{ and } y \ge- \tilde{x}^i+x_0,
 \\ [3pt]
 u^i\left(\pmb{x}^{-i}, x^i + y, 0\right)  & \mbox{if }  \tilde{x}^i<-f^{-1}_N(y) \mbox{ and } y < -\tilde{x}^i+x_0,
\end{array}\right.
\end{equation}
and define $v^i: \mathbb{R}^N \times \mathbb{R}^N_{+} \rightarrow \mathbb{R}$ as
\begin{equation}
\label{eq:valueN2}
v^i(\pmb{x},\pmb{y}) = 
\left\{ \begin{array}{cll}
u^i(\pmb{x},{y}^i) & \mbox{if } (\pmb{x}, \pmb{y}) \in \overline{\mathcal{W}_{-i}} , \\ 
 v^i\left(\pmb{x}^{-j}, x^j_{+} + \frac{\sum_{k \neq j}x^k}{N-1},\big (\pmb{y}^{-j},f_N(x_{+}^j)\big)\right) & \mbox{if } (\pmb{x}, \pmb{y}) \in \mathcal{A}_j^+ \cap {E}_{j,1}^+  \mbox{ for } j \ne i, \\ [3 pt]
 v^i\left(\pmb{x}^{-j}, x^j -y^j,\big(\pmb{y}^{-j},0\big)\right) & \mbox{if } (\pmb{x}, \pmb{y}) \in \mathcal{A}_j^+ \cap {E}_{j,2}^+  \mbox{ for } j \ne i, \\ [3 pt]
v^i\left(\pmb{x}^{-j}, \frac{\sum_{k \neq j}x^k}{N-1} + x^j_{-},\big(\pmb{y}^{-j}, \tilde{f}_N(x_{-}^j)\big) \right) & \mbox{if } (\pmb{x}, \pmb{y}) \in \mathcal{A}_j^- \cap {E}_{j,1}^- \mbox{ for } j \ne i, \\ [3 pt]
v^i\left(\pmb{x}^{-j}, x^j+y^j,\big(\pmb{y}^{-j}, 0\big) \right) & \mbox{if } (\pmb{x}, \pmb{y}) \in \mathcal{A}_j^- \cap {E}_{j,2}^- \mbox{ for } j \ne i,
\end{array}\right.
\end{equation}}
where 
\begin{itemize}[itemsep = 3 pt]
\item
$\mathcal{A}_i$ and $\mathcal{W}_i$ are given in \eqref{eq:AWB}, $E^{\pm}_{i,1}$ and $E^{\pm}_{i,2}$  are given in \eqref{eq:thresholdB21}-\eqref{eq:thresholdB22} with $f_N(\cdot)$ defined by \eqref{eq:fNder}-\eqref{eq:intercept},and $\tilde{f}_N (x)= f_N(-x)$ for $x<0$.
\item
$\widetilde{x}^i$ is defined by \eqref{eq:tildex}, and $A_N(\cdot)$ is defined by \eqref{eq:AN}.
\item
$x^i_{+}$ in \eqref{eq:valueN2-aux} is the unique positive root of $z - f_N(z) = \widetilde{x}^i - y$ when $\tilde{x}^i \ge f_N^{-1}(y)$, and $x^i_{-}$ is the unique negative root of $z + \tilde{f}_N(z) = \widetilde{x}^i + y$ when  $\tilde{x}^i <  -f_N^{-1}(y)$.
\item
$x^j_{+}$ in \eqref{eq:valueN2} is the unique positive root of $z - f_N(z) = \widetilde{x}^j - y^j$ if $\tilde{x}^j \ge  f_N^{-1}(y^j)$, and $x^j_{-}$ is the unique negative root of $z + \tilde{f}_N(z) = \widetilde{x}^i + y^j$  if $\tilde{x}^j <- f_N^{-1}(y^j)$.
\end{itemize}
Then $v^i$ is the game value associated with an NEP $\pmb{\xi}^{*} = (\xi^{1*}, \cdots, \xi^{N*})$. That is,
$
v^{i}(\pmb{x}, \pmb{y}) = J_{C_d}^i(\pmb{x},\pmb{y}; \pmb{\xi}^{*}).
$
Moreover, the controlled process $(\pmb{X}^{*}, \pmb{Y}^{*})$ under $\pmb{\xi}^{*}$ is given in this section:
{\bf Case 1} if $(\pmb{x},\pmb{y}) \in \overline{\mathcal{W}_{NE}}$, and {\bf Case 2} if $(\pmb{x},\pmb{y})\notin \overline{\mathcal{W}_{NE}}$.
\end{theorem}
The proof of Theorem \ref{thm:NB} is similar to that of Theorem \ref{thm:NA} {and hence omitted}.

\section{Nash Equilibrium for game $\pmb{C}$} \label{s6}
\label{section:general}
\quad In the previous two sections, we have dealt with two special games $\pmb{C_p}$ and $\pmb{C_d}$. Analysis of these two games provides important insight into the solution structure of the general game $\pmb{C}$. Namely, the NE strategy depends on the positions of players and their remaining resource levels.
With these two special cases in mind, now recall that in game $\pmb{C}$, 
\begin{equation}
\label{eq:YB}
d Y^j_t = -  \sum_{i=1}^N \frac{a_{ij}Y^j_{t-}}{\sum_{k=1}^M a_{ik} Y^k_{t-}} d\check{\xi}^i_t \quad \mbox{and} \quad Y^j_{0-} = y^j{\ge0}.
\end{equation}

 \quad For the HJB equation $(HJB-C)$,  the gradient constraint is more complicated than the two special cases $\pmb{C_p}$ and $\pmb{C_d}$. When $\mathcal{A}_i \cap \mathcal{A}_j = \emptyset$,
\small{
 \begin{eqnarray*}
    \textit{(HJB-$C$)} \left\{
                \begin{array}{ll}
                \displaystyle \min \Bigg\{-\alpha v^i +h + \frac{1}{2} \sum_{j=1}^N v^i_{x^jx^j}, \Gamma_iv^i +v^i_{x^i}, -\Gamma_iv^i-v^i_{x^i}\Bigg\} = 0, \\ [3 pt]
          \hspace{265pt}   \mbox{for } (\pmb{x},\pmb{y}) \in \mathcal{W}_{-i},  \\  [3 pt]
 \displaystyle  {-\Gamma_jv^i-v^i_{x^j} = 0, 
           \hspace{190pt}   \mbox{for } (\pmb{x},y) \in \mathcal{A}^+_{j}, j \ne i, }\\  [3 pt]
 \displaystyle  { -\Gamma_jv^i+v^i_{x^j}= 0, 
           \hspace{190pt}   \mbox{for } (\pmb{x},y) \in \mathcal{A}^-_{j}, j \ne i.}
                \end{array}
              \right.
\end{eqnarray*}
}
In particular, if $\pmb{A}=[1,1,\cdots,1]^T \in \mathbb{R}^{N \times 1}$, then $(HJB-C)$ becomes $(HJB-C_p)$; and if $\pmb{A}=\pmb{I_N}$, then it is $(HJB-C_d)$.

\quad Similar to Section \ref{section:gamepooling}, define the action region $\mathcal{A}_i {\in \mathbb{R}^N\times\mathbb{R}^M_{+} }$ and the waiting region $\mathcal{W}_i$ of the $i^{th}$ player by
\begin{equation}
\label{eq:AWB-general}
\mathcal{A}^+_i := E_i^{+}  \cap Q_i,\quad \mathcal{A}^-_i := E_i^{-} \cap Q_i \quad,\mathcal{A}_i = \mathcal{A}^+_i\cup \mathcal{A}^-_i,\quad \mbox{and} \quad  \mathcal{W}_i:= {\mathbb{R}^{N} \times \mathbb{R}^{M}_+}\setminus \mathcal{A}_i,
\end{equation}
where
\begin{equation}
\label{eq:thresholdB-general}
E_i^{+} : = \left\{(\pmb{x}, \pmb{y}) \in \mathbb{R}^{N} \times (\mathbb{R}_{+}^*)^M: \widetilde{x}^i \ge f_N^{-1}\left(\sum_{j=1}^M a_{ij}y^j \right) \right\}, \, E_i^{-} : = \left\{(\pmb{x}, \pmb{y}) \in \mathbb{R}^{N} \times (\mathbb{R}_{+}^*)^M: \widetilde{x}^i \le - f_N^{-1} \left(\sum_{j=1}^M a_{ij}y^j \right) \right\},
\end{equation}
{ with
\begin{eqnarray}
E_{i,1}^{+} &:=& \left\{(\pmb{x}, y) \in E_{i}^{+} : \sum_{j=1}^M a_{ij}y^j \ge \tilde{x}^i+x_0 \right\},\quad E_{i,2}^{+} := \left\{(\pmb{x}, y) \in E_{i}^{+} :\sum_{j=1}^M a_{ij}y^j <  \tilde{x}^i+x_0 \right\},\label{eq:thresholdB21-general}\\
E_{i,1}^{-} &:=& \left\{(\pmb{x}, y) \in E_{i}^{-} : \sum_{j=1}^M a_{ij}y^j  \ge - \tilde{x}^i-x_0 \right\},\quad E_{i,2}^{-} := \left\{(\pmb{x}, y) \in E_{i}^{-} : \sum_{j=1}^M a_{ij}y^j  < - \tilde{x}^i-x_0 \right\},\label{eq:thresholdB22-general}
\end{eqnarray}
and  $\{Q_i\}_{i=1}^N$ are   convex partitions  { such that $Q_i \cap Q_j = (E_i^+\cup E_i^-)\cap (E_j^+\cup E_j^-)\cap \partial \mathcal{W}_{NE}$ for $i\neq j$.}
We then define 
\begin{eqnarray}\label{eq:mapping_Pi3}
\Pi(\pmb{x},\pmb{y}) = 
\begin{cases}
\left((\pmb{x}^{-i}, x^i_{+} + \frac{\sum_{k \neq i}x^k}{N-1}), \pmb{y}^1_{+}\right) , \quad  & {\rm if } \quad (\pmb{x},\pmb{y})\in Q_i \cap E_{i,1}^{+},\\
\left((\pmb{x}^{-i}, x^i -\sum_{q=1}^M a_{iq}y^q , \pmb{y}^2_{+}\right),
  & {\rm if } \quad (\pmb{x},\pmb{y})\in Q_i \cap E_{i,2}^{+},\\
\left((\pmb{x}^{-i},\frac{\sum_{k \neq i}x^k}{N-1}+x_{-}^i),\pmb{y}^1_{-}\right), & {\rm if } \quad (\pmb{x},\pmb{y})\in Q_i \cap E_{i,1}^{-},\\
\left((\pmb{x}^{-i},x^i+\sum_{q=1}^M a_{iq}y^q),\pmb{y}^2_{-}\right), &{\rm if } \quad (\pmb{x},\pmb{y}) \in Q_i \cap E_{i,2}^{-},
\end{cases}
\end{eqnarray}
in which the threshold function $f_N(\cdot)$ is defined in \eqref{eq:fNder}-\eqref{eq:intercept}, $x_{+}^i$ is the unique positive root such that $z-f_N(z) = \tilde{x}^i-y^i$ when $\tilde{x}^i \ge f_N^{-1}(y^i)$, $x_{-}^i$ is the unique negative root such that $z+\tilde{f}_N(z) = \tilde{x}^i+y^i$ when $\tilde{x}^i \le -f_N^{-1}(y^i)$. Here $ \pmb{y}^1_{+}\in \mathbb{R}^M_{+}$ with  the $j$-th component being $(\pmb{y}^1_{+})_j = y^j - \frac{a_{ij}y^j}{\sum_{q=1}^M a_{iq}y^q} \left(\sum_{q=1}^M a_{iq}y^q-f_N(x_{+}^i) \right)$, $ \pmb{y}^2_{+}\in \mathbb{R}^M_{+}$ with  the $j$-th component being $(\pmb{y}^2_{+})_j = y^j - a_{ij}y^j$,   $ \pmb{y}^1_{-}\in \mathbb{R}^M_{+}$ with the $k$-th component being $(\pmb{y}^1_{-})_j = y^j - \frac{a_{ij}y^j}{\sum_{q=1}^M a_{iq}y^q} \left(\sum_{q=1}^M a_{iq}y^q-\tilde{f}_N(x_{-}^i) \right)$, $ \pmb{y}^2_{-}\in \mathbb{R}^M_{+}$ with the $j$-th component being $(\pmb{y}^2_{-})_j = y^j - {a_{ij}y^j}$. 
Note that, $\Pi (\cdot)$ translates  $(\pmb{x},\pmb{y})$ to the boundary of ${E}_{i,1}^{+}$, i.e., $\partial {E}_{i,1}^{+} :=\{(\pmb{x},\pmb{y})\,\,:\,\,\sum_{j=1}^M a_{ij}y^j=f_N^{-1}\left(\tilde{x}^i\right), 0 < \tilde{x}_i \leq x_0 \}$  when $(\pmb{x},\pmb{y})\in Q_i \cap E_{i,1}^{+}$, and to $\{ (\pmb{x},\pmb{y})\,\,:\,\,a_{ij}y^j=0, \forall j=1,2,\cdots,M\}$ when $(\pmb{x},\pmb{y})\in Q_i \cap E_{i,2}^{+}$, both along the direction $$\left(0,\cdots,-1,\cdots 0;-\frac{a_{i1}y^1}{\sum_{j=1}^M a_{ij}y^j},\cdots,-\frac{a_{iM}y^M}{\sum_{j=1}^M a_{ij}y^j}\right)\in \mathbb{R}^{N}\times \mathbb{R}^{M}_{+}$$ with the  $i$-th component being $-1$. Denote
\begin{align}
\label{eq:WNE3}
\mathcal{W}_{NE} := \left\{(\pmb{x},\pmb{y}) \in \mathbb{R}^{N}\times \mathbb{R}_{+}^{M}: |\widetilde{x}^i| < f_N^{-1}\left(\sum_{j=1}^M a_{ij}y^j\right), \, 1 \le i \le N\right\}  { \cup \{(\pmb{x},\pmb{y}) \in \mathbb{R}^{N}\times \mathbb{R}_{+}^{M}: \pmb{y}=0\}, }
\end{align}
and assume the partition $\{Q_i\}_{i=1}^N$ satisfies following assumption:
\begin{enumerate}[font=\bfseries,leftmargin=2cm]
\item[H3-${\bf C}$.] For any $(\pmb{x},\pmb{y}) \in \cup_i \mathcal{A}_i$, \quad $\Pi(\pmb{x},\pmb{y}) \in \overline{\mathcal{W}_{NE}}.$
\end{enumerate}
 Condition {\bf H3-${\bf C}$} implies that if $(\pmb{x},\pmb{y})\in \mathcal{A}_i$, then the dynamics will be in region $\overline{\mathcal{W}_{NE}}$ after player $i$'s control. 
}

\quad From the analysis in Sections \ref{section:gamepooling} and \ref{section:gamedivide}, and the ``guess'' that the control policy of player $i$ only depends on $(\pmb{x},\sum_{j=1}^M a_{ij}y^j)$ when in $\mathcal{W}_{-i}$, we get for $|\widetilde{x}^i| < f_N^{-1}(\sum_{j=1}^M a_{ij}y^j)$,
\begin{equation}
\label{eq:candidateB-general}
v^i(\pmb{x},\pmb{y}) =  p_N(\widetilde{x}^i) + A_N \left(\sum_{j=1}^M a_{ij}y^j\right) \cosh\left(\widetilde{x}^i \sqrt{\frac{2(N-1) \alpha}{N}} \right),
\end{equation}
is a solution to $(\mbox{HJB-$C$})$, where $p_N(\cdot)$ is defined by \eqref{eq:pN}, and $A_N(\cdot)$ defined by \eqref{eq:AN}.

\quad The next step is to construct the controlled process $(\pmb{X}, \pmb{Y})$ corresponding to the HJB solution \eqref{eq:candidateB-general}.


\quad Note that $\mathcal{W}_{NE}$ is an unbounded domain in $\mathbb{R}^{2N}$ with $2N$ boundaries. For $i=1,2,\cdots,N$, define the $2N$ faces of $\mathcal{W}_{NE}$
\begin{eqnarray*}
F_i = \{(\pmb{x},\pmb{y}) \in \partial \mathcal{W}_{NE}\,\,\,\vert\,\,\,  (\pmb{x},\pmb{y}) \in \partial E_i^{+}\},\quad
F_{i+N} = \{(\pmb{x},\pmb{y}) \in \partial \mathcal{W}_{NE}\,\,\,\vert \,\,\, (\pmb{x},\pmb{y}) \in \partial E_i^{-}\}.
\end{eqnarray*}
The normal direction on each face is given by
\begin{eqnarray*}
\pmb{n}_i &=& c_i \left(\frac{1}{N-1},\cdots,-1,\cdots,\frac{1}{N-1};(f_N^{-1})^{\prime}\left(\sum_{j=1}^M a_{ij}y^j\right)a_{i1},\cdots,(f_N^{-1})^{\prime}\left(\sum_{j=1}^M a_{ij}y^j\right)a_{iM}\right),\\
\pmb{n}_{N+i} &=& c_{N+i}\left(-\frac{1}{N-1},\cdots,1,\cdots,-\frac{1}{N-1};(f_N^{-1})^{\prime}\left(\sum_{j=1}^M a_{ij}y^j\right)a_{i1},\cdots,(f_N^{-1})^{\prime}\left(\sum_{j=1}^M a_{ij}y^j\right)a_{iM}\right),\\
\end{eqnarray*}
with the $i^{th}$ component being $\pm 1$, and $c_i$ and $c_{N+i}$ the normalizing constants such that $\|\pmb{n}_i\|=\|\pmb{n}_{N+i}\|=1$.

Denote the reflection direction on each face as
\begin{eqnarray*}
\pmb{r}_i &=& c_i^{\prime} \left(0,\cdots,-1,\cdots 0;-\frac{a_{i1}y^1}{\sum_{j=1}^M a_{ij}y^j},\cdots,-\frac{a_{iM}y^M}{\sum_{j=1}^M a_{ij}y^j}\right),\\
\pmb{r}_{N+i} &=& c^{\prime}_{N+i}\left(0,\cdots,1,\cdots 0;-\frac{a_{i1}y^1}{\sum_{j=1}^M a_{ij}y^j},\cdots,-\frac{a_{iM}y^M}{\sum_{j=1}^M a_{ij}y^j}\right),\\
\end{eqnarray*}
with the $i^{th}$ component to be $\pm 1$. $c_i^{\prime}$ and $c^{\prime}_{N+i}$ are normalizing constants such that $\|\pmb{r}_i\|=\|\pmb{r}_{N+i}\|=1$.
The NE strategy is defined as follows.

{\bf Case 1:} $(\pmb{X}_{0-},\pmb{Y}_{0-}) = (\pmb{x},\pmb{y}) \in  \overline{\mathcal{W}_{NE}}$. One can check that $\mathcal{W}_{NE}$ defined in \eqref{eq:WNE3} and $\{\pmb{r}_i\}_{i=1}^{2N}$ defined above satisfies assumptions \textbf{A1}-\textbf{A5}. Therefore, there exists a weak solution to the Skorokhod problem with data $\left(\mathcal{W}_{NE}, \{\pmb{r}_i\}_{i=1}^{2N},\pmb{b},\pmb{\sigma},\pmb{x}\in \overline{\mathcal{W}_{NE}}\right)$. See Appendix A for the satisfiability of \textbf{A1}-\textbf{A5}.

{\bf Case 2:} $(\pmb{X}_{0-},\pmb{Y}_{0-}) = (\pmb{x},\pmb{y}) \notin  \overline{\mathcal{W}_{NE}}$. There exists $i \in \{1, \cdots, N\}$ such that $(\pmb{X}_{0-}, \pmb{Y}_{0-}) \in \mathcal{A}_i$. 
 {(1) If $(\pmb{x},\pmb{y}) \in \mathcal{A}^i_{+}\cap E_{i,1}^+$, then player $i$ will move immediately from $X_{0-}^i=x^i$ to $X_{0}^i=x^i_{+}+ \frac{\sum_{k \neq i}x^k}{N-1}$ at time $0$, where $x^i_{+}$ is the unique positive root such that $z - f_N(z) = \tilde{x}^i - (\sum_{q=1}^M a_{iq}y^q)$. This will reduce the resources from $\pmb{Y}_{0-}=\pmb{y}$ to $\pmb{Y}_{0}=\pmb{y}_{+}$ with  the $j$-th component of $\pmb{y}_{+}$ is $(\pmb{y}_{+})_j = y^j - \frac{a_{ij}y^j}{\sum_{q=1}^M a_{iq}y^q} \left(\sum_{q=1}^M a_{iq}y^q-f_N(x_{+}^i) \right) \ge 0$. Other players' dynamics remain unchanged, i.e., $X_{0}^k = X_{0-}^k=x^k$  for $k\neq i $ and $1 \leq k \leq N$. By Assumption \textbf{H3-$\bf C$}, we have $(\pmb{X}_0,\pmb{Y}_{0})=\left((\pmb{x}^{-i}, x^i_{+} + \frac{\sum_{k \neq i}x^k}{N-1}),\pmb{y}_{+}\right)=\Pi(\pmb{X}_{0-},\pmb{Y}_{0-})  \in  \overline{\mathcal{W}_{NE}}$. (2) If $(\pmb{x},\pmb{y}) \in \mathcal{A}_i^{+}\cap E_{i,2}^+$, then player $i$ will move immediately from $X_{0-}^i=x^i$ to $X_{0}^i=x^i-\sum_{q=1}^M a_{iq}y^q$  and resource $j$ is changed from $Y^j_{0-}=y^j$ to $Y^j_{0} = y^j - a_{ij}y^j$ at time $0$. Other players' dynamics  remain unchanged, i.e., $X_{0}^k = X_{0-}^k=x^k$ for $k\neq i $ and $1 \leq k \leq N$. Under Assumption \textbf{H3-$\bf C$}, we have $(\pmb{X}_0,\pmb{Y}_{0})=\Pi(\pmb{X}_{0-},\pmb{Y}_{0-})  \in  \overline{\mathcal{W}_{NE}}$. (3) Similarly,  if $(\pmb{x},\pmb{y}) \in \mathcal{A}_i^{-}\cap E_{i,1}^-$, then player $i$ will move immediately from $X_{0-}^i=x^i$ to $X_{0}^i=x^i_{-}+ \frac{\sum_{k \neq i}x^k}{N-1}$ at time $0$, where $x^i_{-}$ is the unique negative root such that $z+\tilde{f}_N(z) = \tilde{x}^i+ (\sum_{q=1}^M a_{iq}y^q)$. This changes the resources from $\pmb{Y}_{0-}=\pmb{y}$ to $\pmb{Y}_{0} =\pmb{y}_{-}$ where $j$-th component of $\pmb{y}_{-}$ is $(\pmb{y}_{-})_j = y^j - \frac{a_{ij}y^j}{\sum_{q=1}^M a_{iq}y^q} \left(\sum_{q=1}^M a_{iq}y^q-\tilde{f}_N(x_{-}^i) \right)\ge 0$. Other players' dynamics remain unchanged at time $0$, i.e., $X_{0}^k = X_{0-}^k=x^k$   for $k \neq i$ and $1 \leq k \leq N$. By Assumption \textbf{H3-$\bf C$}, we have $(\pmb{X}_0,\pmb{Y}_{0})=\left((\pmb{x}^{-i},\frac{\sum_{k \neq i}x^k}{N-1}+x_{-}^i),\pmb{y}_{-}\right) =\Pi(\pmb{X}_{0-},\pmb{Y}_{0-}) \in  \overline{\mathcal{W}_{NE}}.$
(4) If $(\pmb{x},\pmb{y}) \in \mathcal{A}_i^{+}\cap E_{i,2}^+$, then player $i$ will move immediately from $X_{0-}^i=x^i$ to $X_{0}^i=x^i+\sum_{q=1}^M a_{iq}y^q$ and resource $j$ is reduced from $Y^j_{0-}=y^j$ to $Y^j_{0} = y^j - a_{ij}y^j$ at time $0$. Other players' dynamics  remain unchanged at time $0$, i.e., $X_{0}^k = X_{0-}^k=x^k$ for $k\neq i $ and $1 \leq k \leq N$. By Assumption \textbf{H3-$\bf C$}, we have $(\pmb{X}_0,\pmb{Y}_{0}) =\Pi(\pmb{X}_{0-},\pmb{Y}_{0-}) \in  \overline{\mathcal{W}_{NE}}$. 
}

\quad The NE for the $N$-player game \eqref{eq:JA_game} with constraint $\pmb{C}$ is stated as follows. 
\begin{theorem}[NE for the $N$-player game  $\pmb{C}$]
\label{thm:ND} 
Assume \textbf{H1$'$}-\textbf{H2$'$} and {\textbf{H3-${\bf C}$}}. Define  $u^i \in\mathbb{R}^N \times \mathbb{R}_{+} \rightarrow \mathbb{R}$ as {
\begin{equation}
\label{eq:valueN3-aux}
u^i(\pmb{x},{y}) = 
\left\{ \begin{array}{cll}
 p_N(\widetilde{x}^i) + A_N(y) \cosh\left( \widetilde{x}^i \sqrt{\frac{2(N-1) \alpha}{N} }\right) & \mbox{if } |\tilde{x}^i|\le f^{-1}_N(y), \mbox{ and } y=0, \\  [3 pt]
 u^i\left(\pmb{x}^{-i}, x^i_{+} + \frac{\sum_{k \neq i}x^k}{N-1}, f_N(x^i_{+})\right)  & \mbox{if }  \tilde{x}^i>f^{-1}_N(y) \mbox{ and } y \ge \tilde{x}^i+x_0, \\ [3pt]
  u^i\left(\pmb{x}^{-i}, x^i - y, 0\right)  & \mbox{if }  \tilde{x}^i>f^{-1}_N(y) \mbox{ and } y < \tilde{x}^i+x_0, \\ [3pt]
 u^i\left(\pmb{x}^{-i}, \frac{\sum_{k \neq i}x^k}{N-1} + x^i_{-}, \tilde{f}_N(x^i_{-})\right)  & \mbox{if }  \tilde{x}^i<-f^{-1}_N(y) \mbox{ and } y \ge- \tilde{x}^i-x_0,
 \\ [3pt]
 u^i\left(\pmb{x}^{-i}, x^i + y, 0\right)  & \mbox{if }  \tilde{x}^i<-f^{-1}_N(y) \mbox{ and } y < -\tilde{x}^i-x_0,
\end{array}\right.
\end{equation} 
and define $v^i: \mathbb{R}^N \times \mathbb{R}^M_{+} \rightarrow \mathbb{R}$ as
\begin{equation}
\label{eq:valueND}
v^i(\pmb{x},\pmb{y}) = 
\left\{ \begin{array}{cll}
u^i\left(\pmb{x},\sum_{j=1}^M a_{ij}y^j\right) & \mbox{if } (\pmb{x}, \pmb{y}) \in \overline{\mathcal{W}_{-i}} , \\ 
 v^i\left(\pmb{x}^{-j}, x^j_{+} + \frac{\sum_{k \neq j}x^k}{N-1},\pmb{y}^1_{+} \right) & \mbox{if } (\pmb{x}, \pmb{y}) \in \mathcal{A}_j^+\cap E_{j,1}^+\mbox{ for } j \ne i, \\ [3 pt]
  v^i\left(\pmb{x}^{-j}, x^j-(\sum_{q=1}^M a_{jq}y^q),\pmb{y}^2_{+}\right) & \mbox{if } (\pmb{x}, \pmb{y}) \in \mathcal{A}_j^+\cap E_{j,2}^+\mbox{ for } j \ne i, \\ [3 pt]
v^i\left(\pmb{x}^{-j}, \frac{\sum_{k \neq j}x^k}{N-1} + x^j_{-}, \pmb{y}^1_{-} \right) & \mbox{if } (\pmb{x}, \pmb{y}) \in \mathcal{A}_j^- \cap E_{j,1}^-\mbox{ for } j \ne i,\\ [3 pt]
v^i\left((\pmb{x}^{-j}, x^j+(\sum_{q=1}^M a_{jq}y^q), \pmb{y}^2_{-} \right) & \mbox{if } (\pmb{x}, \pmb{y}) \in \mathcal{A}_j^- \cap E_{j,2}^-\mbox{ for } j \ne i,
\end{array}\right.
\end{equation}}
where 
\begin{itemize}[itemsep = 3 pt]
\item
$\mathcal{A}_i$ and $\mathcal{W}_i$ are given in \eqref{eq:AWB-general},  $E^{\pm}_{i,1}$ and  $E^{\pm}_{i,2}$ are given in \eqref{eq:thresholdB21-general}-\eqref{eq:thresholdB22-general} with $f_N(\cdot)$ defined by \eqref{eq:fNder}-\eqref{eq:intercept},   {and $\tilde{f}_N (x)= f_N(-x)$ for $x<0$.}
\item
$\widetilde{x}^i$ is defined by \eqref{eq:tildex}, and $A_N(\cdot)$ defined by \eqref{eq:AN}.
\item
$x^i_{+}$ in \eqref{eq:valueN3-aux} is the unique positive root of $z - f_N(z) = \widetilde{x}^i - y$ if $ \tilde{x}^i \ge f^{-1}_N(y) $, and $x^i_{-}$ is the unique negative root of $z + \tilde{f}_N(z) = \widetilde{x}^i + y$ if $ \tilde{x}^i<-f^{-1}_N(y) $.
\item 
$x^j_{+}$ in \eqref{eq:valueND} is the unique positive root of $z - f_N(z) = \widetilde{x}^j - \sum_{k=1}^M a_{jk}y^k$ if $ \tilde{x}^j \ge f^{-1}_N(\sum_{q=1}^M a_{jq}y^q) $, and $x^j_{-}$ is the unique negative root of $z + \tilde{f}_N(z) = \widetilde{x}^j + \sum_{k=1}^M a_{jk}y^k$ if $ \tilde{x}^j<-f^{-1}_N(\sum_{q=1}^M a_{jq}y^q) $.
\item { The $k$-th component of $\pmb{y}^1_{+}$ in \eqref{eq:valueND} is
$(\pmb{y}_{+}^1)_k = y^k - \frac{a_{jk}y^k}{\sum_{q=1}^M a_{jq}y^q} \left(\sum_{q=1}^M a_{jq}y^q-f_N(x_{+}^j) \right),$
and  the $k$-th component of $\pmb{y}^1_{-}$ is
$(\pmb{y}^1_{-})_k = y^k - \frac{a_{jk}y^k}{\sum_{q=1}^M a_{jq}y^q} \left(\sum_{q=1}^M a_{jq}y^q-\tilde{f}_N(x_{-}^j) \right).$
\item The $k$-th component of $\pmb{y}^2_{+}$ in \eqref{eq:valueND} is
$(\pmb{y}_{+}^2)_k = y^k - a_{jk}y^k,$
and  the $k$-th component of $\pmb{y}^2_{-}$ is $(\pmb{y}^2_{-})_k = y^k - a_{jk}y^k.$}
\end{itemize}
Then $v^i$ is the value associated with a NEP $\pmb{\xi}^{*} = (\xi^{1*}, \cdots, \xi^{N*})$. That is,
$
v^i(\pmb{x}, \pmb{y}) = J_{C}^i(\pmb{x},\pmb{y}; \pmb{\xi}^{*}).
$
Moreover, the controlled process $(\pmb{X}^{*},\pmb{Y}^{*})$ under $\pmb{\xi}^{*}$ is a solution to a Skorokhod problem as described in {\bf Case 1} if $(\pmb{x},\pmb{y}) \in \overline{\mathcal{W}_{NE}}$, and described as {\bf Case 2} if $(\pmb{x},\pmb{y}) \notin \overline{\mathcal{W}_{NE}}$.
\end{theorem}
 The proof of Theorem \ref{thm:ND} is similar to that of Theorem \ref{thm:NA} and hence omitted. {To demonstrate the similarity, we provide the proof for the convexity of the value function $v^i$ in $\overline{\mathcal{W}_{-i}}$ here.
 \begin{proof}
 We take player one as an example to show $v^1(\pmb{x},\pmb{y})=\tilde{v}(\tilde{x}_1,\sum_{j=1}^M  a_{1j}y^j)$ is convex in $\overline{\mathcal{W}_{-1}}$. Other players' value functions follow similarly.
Recall $\tilde{x}_i =x_i -\frac{\sum_{k\neq i}^N x_k}{N-1}$. Similarly we define $\tilde{y}_i = \sum_{k=1}^M a_{ij}y^j$. When $(\pmb{x},\pmb{y}) \in \overline{\mathcal{W}}_{-1}$, we have $|\tilde{x}_1| \leq \tilde{y}_1$ or $\tilde{y}_1=0$. Hence $\widetilde{v}\left(\tilde{x}_1,\tilde{y}_1\right)$ is positive semi-definite. By  the chain rule,  for $2 \leq k \neq j \leq N$ and $1 \leq q \neq p \leq M$,
\begin{eqnarray*}
&v^1_{x_1 x_1} (\pmb{x},\pmb{y}) = \widetilde{v}_{x x} (\tilde{x}_1,\tilde{y}_1), \quad v_{x_1 x_k}  (\pmb{x},\pmb{y}) = - \frac{1}{N-1}  \widetilde{v}_{x x} (\tilde{x}_1,\tilde{y}_1), \quad v_{x_1 y_q}  (\pmb{x},\pmb{y})  =  a_{1q}\, \widetilde{v}_{x y} (\tilde{x}_1,\tilde{y}_1),\\
&v_{y_p y_q} (\pmb{x},\pmb{y}) =  a_{1p}q_{1q} \widetilde{v}_{y y} (\tilde{x}_1,\tilde{y}_1), v^1_{x_k x_j}  (\pmb{x},\pmb{y}) = \frac{1}{(N-1)^2} \widetilde{v}_{x x} (\tilde{x}_1,\tilde{y}_1), \quad v_{x_k y_q} (\pmb{x},\pmb{y})  = -\frac{1}{N-1}a_{1q}  \widetilde{v}_{x y} (\tilde{x}_1,\tilde{y}_1)\\
\end{eqnarray*}

Denote $H(\pmb{x},\pmb{y}) := \nabla^2 v^1(\pmb{x},\pmb{y}) \in \mathbb{R}^{(N+M) \times (N+M)}$ as the Hessian matrix of $v^1$ at some point $(\pmb{x},\pmb{y}) \in \overline{\mathcal{W}_{-1}}$. Then for any $\pmb{d} = (b_1,\cdots,b_{N},c_1,\cdots,c_M)\in  \mathbb{R}^{N+M}$, 
\begin{eqnarray*}
\pmb{d}^TH(\pmb{x},y) \pmb{d} &=&b_1^2  \widetilde{v}_{xx} -\frac{2}{N-1} \sum_{k=2}^N b_1 b_k  \widetilde{v}_{xx} +\frac{1}{(N-1)^2}\sum_{k=2}^N b_k^2 +\frac{2}{(N-1)^2} \sum_{2 \leq j \neq k \leq N} b_jb_k  \widetilde{v}_{xx} \\
&& \,\, + 2\sum_{q=1}^M b_1\, c_q a_{1q} \widetilde{v}_{x y} - \frac{2}{N-1}\sum_{k=2}^N\sum_{q=1}^M b_k c_q a_{1q}\, \widetilde{v}_{x y} + (\sum_{q=1}^M a_{1q}c_q)^2\widetilde{v}_{y y} \\
&=& \left(b_1-\frac{1}{N-1}\sum_{k=2}^Nb_k\right)^2 \widetilde{v}_{x x} +2 \left(b_1-\frac{1}{N-1}\sum_{k=2}^Nb_k\right)\left( \sum_{q=1}^M a_{1q}c_q\right) \widetilde{v}_{x y} + \left( \sum_{q=1}^M a_{1q}c_q\right)^2 \widetilde{v}_{y y} \\
&=& \pmb{e}^T \tilde{H}(\tilde{x}_1,y)\pmb{e} \ge 0,
\end{eqnarray*}
in which $\pmb{e} = \left(b_1-\frac{1}{N-1}\sum_{k=2}^Nb_k,\sum_{q=1}^M a_{1q}c_q\right)$ and $ \tilde{H}(\tilde{x}_1,y) = \nabla^2 \tilde{v}(\tilde{x}_1,y)$. The last inequality holds since $ \tilde{v}(\tilde{x}_1,y)$ is convex when $|\tilde{x}_1| \leq y$ which is a result in the proof of Theorem 4.3 (Step 1 of (iv)). Therefore $v^1$ is convex in $\overline{\mathcal{W}_{-1}}$. 
 \end{proof}

 }

\section{Comparing Games  {\bf $C_p$},  {\bf $C_d$} and {\bf $C$}}
\label{s7}
\label{section:relation}
\quad  In this section, we compare the games $\pmb{C_p}$, $\pmb{C_d}$ and  $\pmb{C}$. 
We will first compare their game values and discuss their economic implications.  
We will then discuss their difference in terms of the NEP.
Finally, we discuss their perspective NEs in the framework of controlled rank-dependent SDEs. 

\quad To make the games comparable, let us assume 
$y=\sum_{j=1}^N y^j$. Let us also consider a special sharing game $\pmb{C_s}$ which  can be connected with both $\pmb{C_d}$ and $\pmb{C_p}$:
\begin{itemize}[font=\bfseries]
\item[$\pmb{C_s}$:] $M=N$ and $a_{ii}=1$ for $i=1,2,\cdots,N$.
\end{itemize}

\subsection{Pooling, Dividing, and Sharing}
 Denote the game value and waiting region for each player $i$ as $v^i_{C_p}$ and  $\mathcal{W}^{C_p}_{i}$ respectively for game {$\pmb{C_p}$}. Similar notations are defined for $\pmb{C_d}$ and $\pmb{C_s}$.
 \begin{proposition}[Game values comparison] 
Assume \textbf{H1$'$}-\textbf{H2$'$}. 
 For each $ (\pmb{x},\pmb{y})\in \mathbb{R}^N \times \mathbb{R}^N_{+}$, denote $y=\sum_{i=1}^N y^i$. If $(\pmb{x},y)\in \mathcal{W}^{C_p}_i$, and $(\pmb{x},\pmb{y})\in \mathcal{W}^{C_d}_i \cap \mathcal{W}^{C_s}_i$, then,
 \[
 v^i_{C_p}(\pmb{x},y)\leq v^i_{C_s}(\pmb{x},y) \leq v^i_{C_d}(\pmb{x},\pmb{y}), \,\,i=1,2,\cdots,N.
 \]

 \end{proposition}
 
 \begin{proof}

 The comparison is by direct computation. Indeed,
 recall that in case $\pmb{C_p}$, when $(\pmb{x},y)\in \mathcal{W}^{C_p}_i$,
$
 v_{{C_p}}^i(\pmb{x},y)=p_N(\widetilde{x}^i) + A_N(y) \cosh\left( \widetilde{x}^i \sqrt{\frac{2(N-1) \alpha}{N} }\right) ,
$
 for $i=1,2,\cdots,N$, where $\widetilde{x}^i$ is defined in \eqref{eq:tildex} and $A_N$ is defined in \eqref{eq:AN}.
 Similarly, in case {$\pmb{C_d}$}, when $(\pmb{x},\pmb{y})\in \mathcal{W}^{C_d}_i$,
$
 v_{C_d}^i(\pmb{x},\pmb{y})= p_N(\widetilde{x}^i) + A_N(y^i) \cosh\left( \widetilde{x}^i \sqrt{\frac{2(N-1) \alpha}{N} }\right),
$
  for each $i=1,2,\cdots,N$.
  And in case {$\pmb{C_s}$}, when $(\pmb{x},\pmb{y})\in \mathcal{W}^{C_s}_i$,
$
 v_{C_s}^i(\pmb{x},\pmb{y})= p_N(\widetilde{x}^i) + A_N \left(\sum_{j=1}^N a_{ij}y^j \right) \cosh\left( \widetilde{x}^i \sqrt{\frac{2(N-1) \alpha}{N} }\right),
$
  for each $i=1,2,\cdots,N$.
   By elementary calculations, 
$
 A_N^{\prime}(y)<0.
$
Therefore, when $y=\sum_{j=1}^N y^j$, $(\pmb{x},y)\in \mathcal{W}^{C_p}_i$, and $(\pmb{x},\pmb{y})\in \mathcal{W}^{C_d}_i \cap \mathcal{W}^{C_s}_i$,
 \[
 v^i_{C_p}(\pmb{x},y)\leq v^i_{C_s}(\pmb{x},y) \leq v^i_{C_d}(\pmb{x},\pmb{y}).
 \]
 
The first inequality holds because $y =\sum_{i=1}^N y^i\geq \sum_{i=1}^N a_{ij}y^j$ and the equality holds if and only if $a_{ij}=1$ for each $j=1,2,\cdots,N$. The second inequality holds because $a_{ii}=1$ and the equality holds if and only if $a_{ij}=0$ for each $j \neq i$.
 \end{proof}
\quad This result has a clear economic interpretation. 
In a stochastic game where players have the options to share resources, versus the possibility to divide  resources in advance,    sharing  will have lower cost than dividing. Pooling yields the lowest cost for each player. 
\vskip 3 pt


 \begin{figure}
    \centering 
\begin{subfigure}{0.25\textwidth}
\centering
 \includegraphics[width=0.7\linewidth]{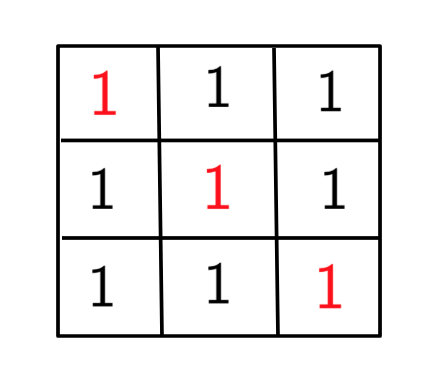}
\end{subfigure}\hfil 
\begin{subfigure}{0.25\textwidth}
\centering
  \includegraphics[width=0.7\linewidth]{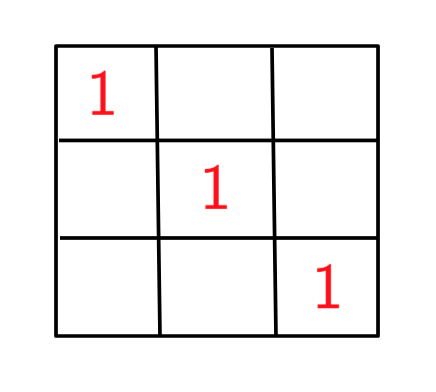}
\end{subfigure}\hfil 
\begin{subfigure}{0.25\textwidth}
\centering
  \includegraphics[width=0.7\linewidth]{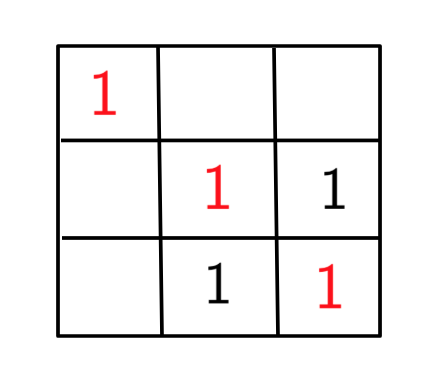}
\end{subfigure}
    \centering 
\begin{subfigure}{0.25\textwidth}
 \includegraphics[width=\linewidth]{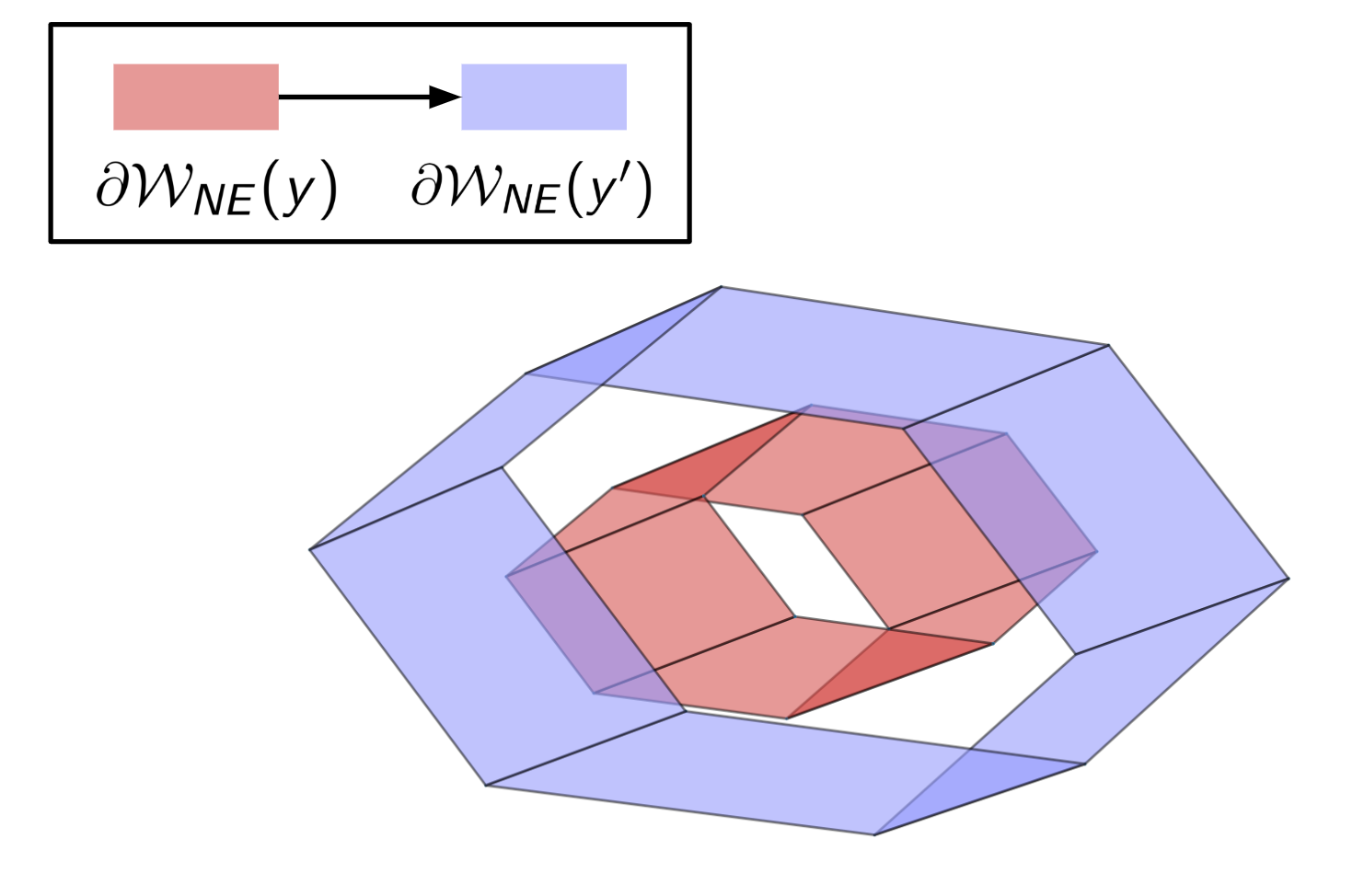}
  \caption{$\pmb{C_p}$}\label{fig:pooling}
\end{subfigure}\hfil 
\begin{subfigure}{0.25\textwidth}
  \includegraphics[width=\linewidth]{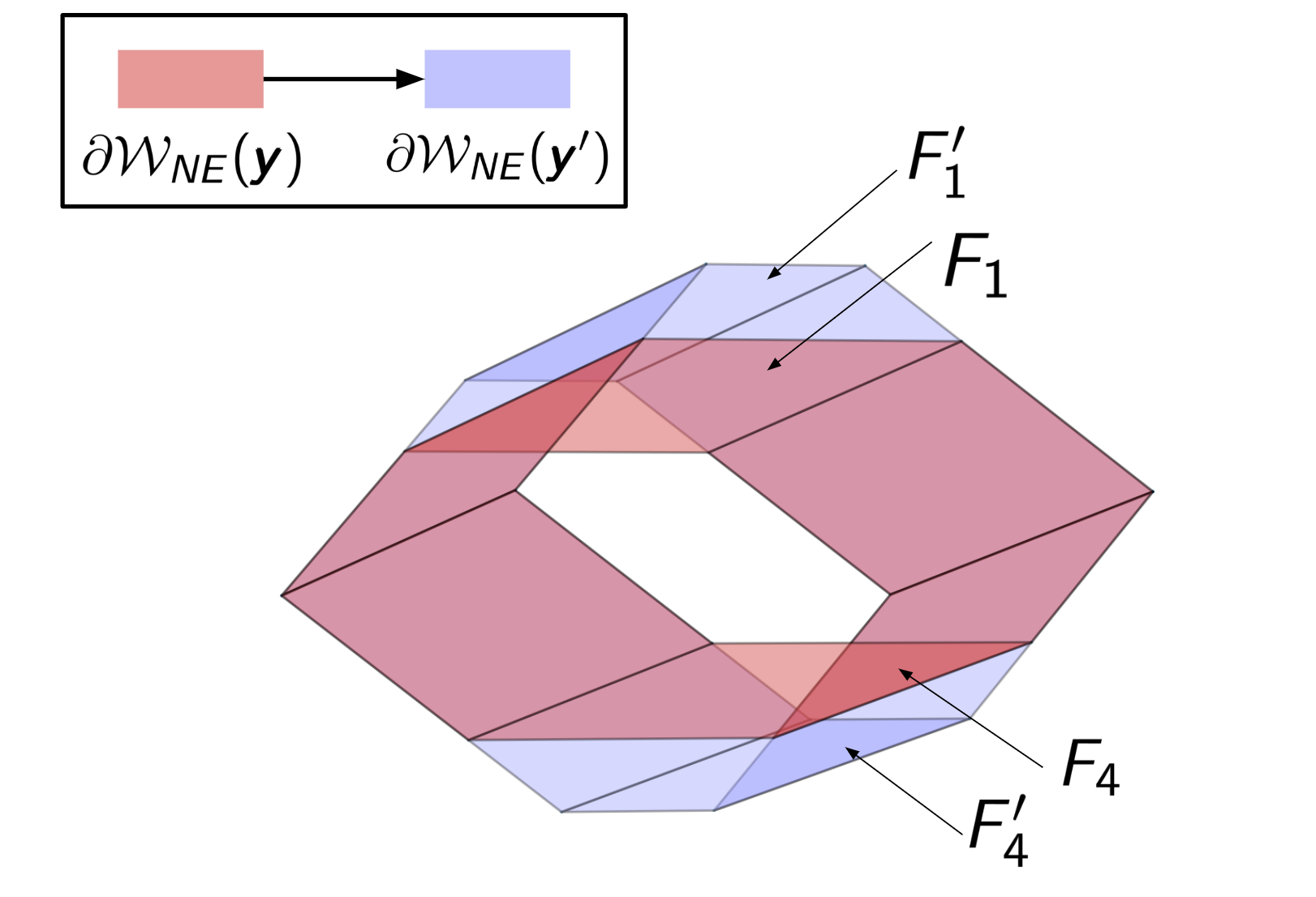}
    \caption{$\pmb{C_d}$}  \label{fig:dividing}
\end{subfigure}\hfil 
\begin{subfigure}{0.25\textwidth}
  \includegraphics[width=\linewidth]{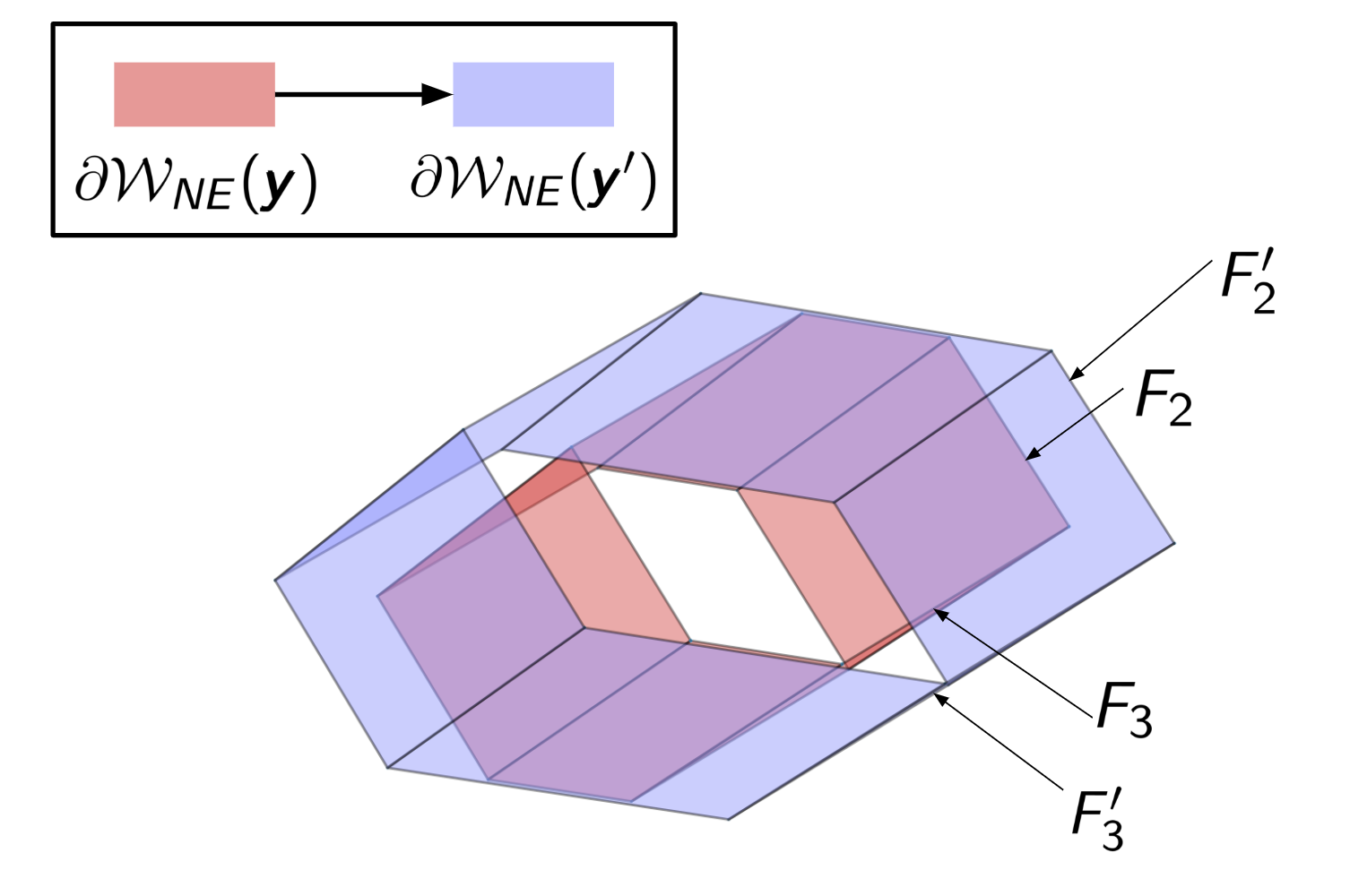}  
    \caption{$\pmb{C}$}  \label{fig:sharing}
\end{subfigure}\hfil 
\caption{Comparison of projected evolving boundaries for $\pmb{C_p}$,
$\pmb{C_d}$, $\pmb{C}$ when $N=3$.}\label{strategy_comparison}
\end{figure}

\quad 
Define the projected common waiting region 
\begin{equation*}
\mathcal{W}_{NE}(\pmb{y}) := \left\{(\pmb{x},\pmb{y}) \in \mathbb{R}^N\times (\mathbb{R}_{+}^*)^M: |\widetilde{x}^i| < f_N^{-1}\left(\sum_{j=1}^M a_{ij}y^j \right) \mbox{ for } 1 \le i \le N \right\} {\cup \{\pmb{y} = 0\}}.
\end{equation*}
for any fixed resource level $\pmb{y}$. Then $\mathcal{W}_{NE}(\pmb{y}) $ is a polyhedron with $2N$ boundary faces. 
Figure \ref{fig:pooling} shows a pooling game $\pmb{C_p}$.  After one player exercises controls, all the faces of the boundary move.
Figure~\ref{fig:dividing} corresponds to a dividing game $\pmb{C_d}$. After player $i$ exercises controls, her faces of $F_i$ and $F_{i+N}$ move. Here $i=1, N=3$.
For a sharing game $\pmb{C}$,  shown in Figure~\ref{fig:sharing}, after one player exercises her controls, the faces of the players who are connected with her will move, while the faces for other players remain unchanged. Here $i=2$ and player $2$ and $3$ are connected.
%
%

%
%
%
%
%
%
%
%
%
%
%
%

\subsection{NEs for the games and controlled rank-dependent SDEs}
\label{section:discussion}

In the previous sections, the controlled dynamics is constructed directly via the reflected Brownian motion.
This class of SDEs can also be cast in the framework of rank-dependent SDEs. 
 Indeed, the controlled dynamics of NE in the action regions of the $N$-player can be written as a {\em controlled rank-dependent SDEs}:
\begin{eqnarray*}
    dX_t^i &=& \sum_{j = 1}^N 1_{F^i(\pmb{X}_t, \pmb{Y}_t)= F^{(j)} (\pmb{X}_t, \pmb{Y}_t)}\left(\delta_{j} dt + \sigma_j dB^j_t + d\xi^{j, +}_t - d\xi^{j, -}_t\right),\quad
d Y_t^j =- \sum_{i = 1}^N  \frac{a_{ij}Y_{s-}^j}{\sum_{k=1}^M a_{ik}Y_{s-}^k}d \check{\xi}^i_s,
\end{eqnarray*}
with $(\xi^{i,+},\xi^{i,-}) $ the controls, 
 $F^i :  \mathbb{R}^{N} \times \mathbb{R}^M_{+} \rightarrow \mathbb{R}$ a rank function depending on both  $\pmb{X}$ and $\pmb{Y}$,  $F^{(1)} \le \cdots \le F^{(N)}$  the order statistics of $(F^i)_{1 \le i \le N}$, 
and $\delta_i \in \mathbb{R}$, $\sigma_i \ge 0$.
\quad  In game  {$\pmb{C_p}$}, the controlled dynamics in the action regions satisfies the SDEs with 
$F^i_{C_p}(\pmb{x}, \pmb{y}) = |x^i - \frac{\sum_{j\ne i}x^j}{N-1} |$, $\delta_i = 0$ and  $\sigma_i = 0$ for each $i = 1, \cdots N$, and
\begin{equation*}
   \xi^{i,\pm} = 0 \quad \mbox{for each } i = 1, \cdots, N-1 \quad \mbox{and} \quad  \xi^{N, \pm} \ne 0.
\end{equation*}
In game $\pmb{C_d}$,
$F^i_{C_d}(\pmb{x}, \pmb{y}) = \Bigg|x^i - \frac{\sum_{j\ne i}x^j}{N-1} - f^{-1}_{N}(y^i)\Bigg|.$
For the general game $\pmb{C}$, the controlled process in the action regions is governed by the rank-dependent dynamics with $F^i_{C}(\pmb{x}, \pmb{y}) = |x^i - \frac{\sum_{j\ne i}x^j}{N-1} - f^{-1}_{N}(\sum_{j=1}^M a_{ij}y^j)|$ with $f_N$ a threshold function defined in \eqref{eq:fNder}-\eqref{eq:intercept} 
and $\delta_i$, $\sigma_i$ and $\xi^{i, \pm}$ satisfying the same condition as before.

\quad Note that the special case without controls, i.e., $F^i(\pmb{x}, \pmb{y}) = x^i$ and $\xi^{i, \pm} = 0$, corresponds to the {\em rank-dependent SDEs}.
In particular, the rank-dependent SDEs with  $\delta_1 = 1$, $\delta_2 = \cdots \delta_N = 0$ is known as the {\em Atlas model}. 
To the best of our knowledge, rank-dependent SDEs with additional controls or a general rank function $F^i$ has not been 
studied before. There are various aspects including uniqueness and sample path properties that await  further investigation and we leave them to interested readers.

\bibliographystyle{abbrv}
\bibliography{FF}

\section*{Appendix A}
Take $n=N$, $m=M$ and $I=2N$, and $\mathcal{I}=\{1,2,\cdots,2N\}$ in Definition \ref{CSRBM}. We then check the satisfiability for Assumptions \textbf{A1}-\textbf{A5} for game $\pmb{C}$. $\pmb{C_p}$ and $\pmb{C_d}$ are two special cases. 

{\bf A1 } Assumption \textbf{A1} is trivially satisfied by definition. We write 
$$
G = \cap_{j=1}^{2N} G_j,
$$
where 
$G_i =\left\{(\pmb{x},\pmb{y})\in \mathbb{R}^{N+M}\right\vert \tilde{x}^i \leq f^{-1}_N \left(\sum_{j=1}^M a_{ij}y^j )\right\}$ and  
$G_{N+i} =\left\{\left.(\pmb{x},\pmb{y})\in \mathbb{R}^{N+M} \right\vert \tilde{x}^i \geq -f^{-1}_N \left(\sum_{j=1}^M a_{ij}y^j \right)\right\}$
for $i=1,2,\cdots,N$. The boundary of $G_i$ is smooth since $f^{-1}_N$ is smooth.

{\bf A2 }
Assumption \textbf{A2} is satisfied since $f^{-1}_N$ is smooth and decreasing. It satisfies the uniform exterior cone condition. At any boundary point $(\pmb{x}_0,\pmb{y}_0) \in \partial G_j$, we can put a truncated closed right circular cone $V_{(\pmb{x}_0,\pmb{y}_0)}$ satisfying $V_{(\pmb{x}_0,\pmb{y}_0)} \cap \bar{G}=\{(\pmb{x}_0,\pmb{y}_0)\}$.

{\bf A3 }
Assumption \textbf{A3} can be shown by contradiction. The proof is inspired from that of \cite[Lemma ({A.2})]{KW2007} which is for bounded region with tightness argument. We modify the proof via a shifting argument.

Suppose that Assumption \textbf{A3} does not hold. {Since there are only finite many subsets $\mathcal{I}_0 \subseteq \mathcal{I}=\{1,2,\cdots,2N\}$ such that $\mathcal{I}_0 \neq \emptyset$,} there is an $\epsilon>0$, a nonempty set $\mathcal{I}_0 \subseteq \mathcal{I}$, a sequence $\{\epsilon_n\} \subset (0,\infty)$ with $\epsilon_n \rightarrow 0$ as $n \rightarrow \infty$, a sequence $\{(\pmb{x}_n,\pmb{y}_n)\} \subset \mathbb{R}^{N+M}$ such that for each $n$, $(\pmb{x}_n,\pmb{y}_n) \in \cap_{j \in \mathcal{I}_0} U_{\epsilon_n}(\partial G_j \cap \partial G)$ and $\text{dist}((\pmb{x}_n,\pmb{y}_n),\cap_{j \in \mathcal{I}_0}( \partial G_j \cap \partial G)) \geq \epsilon$. Note that $\text{dist}((\pmb{x},\pmb{y}),\cap_{j \in \mathcal{I}_0}( \partial G_j \cap \partial G)) =\text{dist}((\pmb{x}-a\textbf{1},\pmb{y}),\cap_{j \in \mathcal{I}_0}( \partial G_j \cap \partial G)) $ for any $a \in \mathbb{R}$ and $(\pmb{x},\pmb{y}) \in \mathbb{R}^{N+M}$. Here $\textbf{1} \in \mathbb{R}^{N}$ is a vector with all ones. Intuitively, this is because for any fixed $\pmb{y}$, the projection of $G$ onto $\pmb{x}$-space is a polyhedron unbounded along the directions of $\pm\bf{1}\in \mathbb{R}^N$.
This is consistent with the model where we only look at the relative distance between positions. {Mathematically speaking,  recall that 
\begin{eqnarray*}
\partial G_i &=&\left\{(\pmb{x},\pmb{y})\in \mathbb{R}^{N+M}\Bigg\vert \tilde{x}^i = f^{-1}_N \left(\sum_{j=1}^M a_{ij}y^j \right)\right\}, \\
\partial G_{N+i} &=&\left\{(\pmb{x},\pmb{y})\in \mathbb{R}^{N+M} \Bigg\vert \tilde{x}^i = -f^{-1}_N \left(\sum_{j=1}^M a_{ij}y^j \right)\right\}.
\end{eqnarray*}
For a given point $\pmb{p} = (\pmb{x},\pmb{y})$, denote $d_k := {\rm dist} ( (\pmb{x},y),\partial G_k) $ for $k=1,2,\cdots,N$. Then there exists a point $\pmb{q}= (\pmb{w},\pmb{z})$  such that
\begin{eqnarray*}
\tilde{w}^i &=& f^{-1}_N \left(\sum_{j=1}^M a_{ij}z^j\right), \quad i.e., \pmb{q} \in \partial G_k\\
\pmb{q} - \pmb{p} &=& d_k \pmb{n}_k (\pmb{q}), \mbox{ or }\quad \pmb{q} - \pmb{p} = -d_k \pmb{n}_k (\pmb{q}).
\end{eqnarray*}
where $\pmb{n}_k (\pmb{q})$ is the normal direction of surface $\partial G_k$ at point $\pmb{q}$:
\begin{eqnarray*}
\pmb{n}_k(\pmb{q}) &=& c_k \left(\frac{1}{N-1},\cdots,-1,\cdots,\frac{1}{N-1};(f_N^{-1})^{\prime}\left(\sum_{j=1}^M a_{ij}z^j\right)a_{i1},\cdots,(f_N^{-1})^{\prime}\left(\sum_{j=1}^M a_{ij}z^j\right)a_{iM}\right).
\end{eqnarray*}

Denote $\tilde{\pmb{p}} =  (\pmb{x}-a\pmb{1},\pmb{y})$ and  $\tilde{\pmb{q}} =  (\pmb{w}-a\pmb{1},\pmb{z})$. Then it is easy to check that 
\begin{eqnarray}
&&\tilde{\pmb{q}} \in \partial G_k, \label{condition1}\\
&&\pmb{n}_k(\pmb{q}) = \pmb{n}_k(\tilde{\pmb{q}}),\label{condition2}\\ 
&&\tilde{\pmb{q}}-\tilde{ \pmb{p}} = \pmb{q} - \pmb{p} = d_k \pmb{n}_k (\pmb{q}) = d_k  \pmb{n}_k(\tilde{\pmb{q}}).\label{condition3}
\end{eqnarray}
\eqref{condition1} holds since $(w^i-a)-\frac{\sum_{j=1,j \neq i}^N (w^j-a)}{N-1} = x^i-\frac{\sum_{j=1,j \neq i}^N w^j}{N-1}$,  
\eqref{condition2} holds since the last M elements, representing the resource levels, are the same for $\pmb{q}$ and $\tilde{\pmb{q}}$, and  \eqref{condition3} holds by definition and \eqref{condition2}.

By \eqref{condition3} we conclude that $ {\rm dist} ( (\pmb{x}-a\pmb{1},y),\partial G_k)=d_k$. Similar results hold for $k=N+1,\cdots,2N$. Therefore we have 
\begin{eqnarray*}
{\rm dist} ((\pmb{x},\pmb{y}),\cap_{j\in \mathcal{I}_0} \partial G_j \cap \partial G) = {\rm dist} ((\pmb{x}-a\pmb{1},\pmb{y}),\cap_{j\in \mathcal{I}_0} \partial G_j \cap \partial G).
\end{eqnarray*}
}

Therefore, for each $(\pmb{x}_n,\pmb{y}_n)$, there exists $a_n \in \mathbb{R}$ such that $\|\pmb{x}_n-a_n\textbf{1} \| \leq 1$. Denote $\tilde{\pmb{x}}_n = \pmb{x}_n-a_n\textbf{1}$. Hence $(\tilde{\pmb{x}}_n,\pmb{y}_n)$ is a bounded sequence in $\mathbb{R}^{N+M}$ and $\text{dist}((\tilde{\pmb{x}}_n,\pmb{y}_n),\cap_{j \in \mathcal{I}_0}( \partial G_j \cap \partial G)) \geq \epsilon$. WLOG, we may assume that $(\tilde{\pmb{x}}_n,\pmb{y}_n) \rightarrow (\pmb{x},\pmb{y})$
as $n \rightarrow \infty$ for some $(\pmb{x},\pmb{y})\in \mathbb{R}^{N+M}$. It follows that $ (\pmb{x},\pmb{y}) \in \cap_{j \in \mathcal{I}_0}( \partial G_j \cap \partial G)$, since for each $j \in \mathcal{I}_0$,
\[
\text{dist}( (\pmb{x},\pmb{y}),\partial G_j \cap \partial G) \leq \|(\tilde{\pmb{x}}_n,\pmb{y}_n)-(\pmb{x},\pmb{y})\| +\text{dist}((\tilde{\pmb{x}}_n,\pmb{y}_n),\partial G_j \cap \partial G) \leq  \|(\tilde{\pmb{x}}_n,\pmb{y}_n)-(\pmb{x},\pmb{y})\| + \epsilon_n \rightarrow 0,
\]
as $n \rightarrow \infty$. This contradicts with the fact that $(\tilde{\pmb{x}}_n,\pmb{y}_n) \rightarrow (\pmb{x},\pmb{y})$ and $\text{dist}((\tilde{\pmb{x}}_n,\pmb{y}_n),\cap_{j \in \mathcal{I}_0}( \partial G_j \cap \partial G)) \geq \epsilon$.

{\bf A4 }
On each face $j=1,2,\cdots,2N$, $\pmb{r}_j$ is a function of $\pmb{y}$, which is bounded. Moreover, $\pmb{r}_j$ is smooth and $D_{\pmb{y}}\pmb{r}_j$ is bounded. Therefore, $\pmb{r}_j(\cdot)$ is uniformly Lipschitz continuous function. Note that when the adjacent matrix $A=\{a_{kj}\}_{1\leq k,j \leq N}$ is an identity matrix or matrix with all ones, $\pmb{r}_i$ is constant on $\partial G_i$ for all $i\in l$.

{\bf A5 }
Denote $g := f^{-1}_N$. First we show that $g$ is a non-negative decreasing function on $[0,y_{\tiny \mbox{total}}]$ where  $y_{\tiny \mbox{total}}:=\sum_{j=1}^M y^j$ is the total resource. 
{ We have proved in Lemma \ref{lemma:f} that $f^{\prime}_N(z) < 0$ for $z \geq 0$.} 
So there exists $0<\tilde{k}(y_{\tiny \mbox{total}})< \tilde{K}(y_{\tiny \mbox{total}})<\infty$ such that $-\infty<-\tilde{K}(y_{\tiny \mbox{total}})<f^{\prime}_N(z)<-\tilde{k}(y_{\tiny \mbox{total}})<0$ when $z \in [\underline{x},\overline{x}]$. Here $\underline{x}=g(y_{\tiny \mbox{total}})>0$ and $\overline{x}=g(0)$. Note that $g^{\prime}(\cdot) = \frac{1}{f^{\prime}(f^{-1}(\cdot))}$, therefore $-\frac{1}{\tilde{k}(y_{\tiny \mbox{total}})} \leq g^{\prime}(w) \leq -\frac{1}{\tilde{K}(y_{\tiny \mbox{total}})}$ when $w\in[0,y_{\tiny \mbox{total}}]$. 
Now let  $k(y_{\tiny \mbox{total}}) := \frac{1}{\tilde{K}(y_{\tiny \mbox{total}})}$ and $K(y_{\tiny \mbox{total}}) := \frac{1}{\tilde{k}(y_{\tiny \mbox{total}})}$.
%
%
%
%
%
%



\quad It is straightforward that all the latter $M$ components in $\pmb{n}_{j}$ and $\pmb{r}_j$ are non-positive ($1\leq j \leq 2N$).
By simple calculation, we have $\frac{1}{\sqrt{\frac{N}{N-1}+K^2(y_{\rm total})N}} \leq c_j \leq \frac{1}{\sqrt{\frac{N}{N-1}+k^2(y_{\rm total})}}$ and $ \sqrt{\frac{N}{N+1}}\leq c_{j}^{\prime} \leq \frac{1}{\sqrt{2}}$for all $1 \leq j \leq N$.
Similar to the definition of $\pmb{r}_{j}^{+}$ and $\pmb{r}_{j}^{-}$, denote $\pmb{n}_j^{+}$ as the first $N$ components in $\pmb{n}^j$ and $\pmb{n}_j^{-}$ as the latter $M$ components in $\pmb{n}^j$.
Since face $i$ and $N+i$ are parallel to each other ($i=1,2,\cdots,N$), there are at most $N$ faces intersecting with each other. It suffices to consider $(\pmb{x},\pmb{y})$ such that $|\mathcal{I}((\pmb{x},\pmb{y}))|=N$. For these points, consider $c_i=\frac{1}{N}$ and $d_i=\frac{1}{N}$ $(i=1,2,\cdots,N)$. Therefore, for $i^* \in \{i,N+i\}$ with $i=1,2,\cdots,N$,
\begin{eqnarray*}
\left\langle \frac{\sum_{i=1}^N {\pmb{n}}_{i*}}{N},{\pmb{r}_{i*}} \right\rangle 
\geq \frac{1}{N}\langle \pmb{n}_{i*}^{-} ,{\pmb{r}_{i*}^{-}}\rangle= \frac{1}{N}c_{i^*}^{\prime}c_{i^*}\langle \pmb{n}_{i*}^{-} ,{\pmb{r}_{i*}^{-}}\rangle = -c_{i^*}^{\prime}c_{i^*}g^{\prime}\left(\sum_{j=1}^M a_{ij}y^j \right) \geq \frac{1}{\sqrt{\frac{N+1}{N-1}+(N+1)K(y_{\tiny \mbox{total}})}}k(y_{\tiny \mbox{total}}).
\end{eqnarray*}
\quad Similarly,  for $i^* \in \{i,N+i\}$ with $i=1,2,\cdots,N$,
\begin{eqnarray*}
\left\langle \frac{\sum_{i=1}^N {\pmb{r}}_{i*}}{N},\pmb{n}_{i*} \right\rangle 
\geq \frac{1}{N}\langle \pmb{n}_{i*}^{-} ,{\pmb{r}_{i*}^{-}}\rangle= \frac{1}{N}\langle \pmb{n}_{i*}^{-} ,{\pmb{r}_{i*}^{-}}\rangle = -c_{i^*}^{\prime}c_{i^*}g^{\prime}\left(\sum_{j=1}^M a_{ij}y^j \right) \geq \frac{1}{\sqrt{\frac{N+1}{N-1}+(N+1)K(y_{\tiny \mbox{total}})}}k(y_{\tiny \mbox{total}}).
\end{eqnarray*}


\section*{Appendix B}
\subsection*{Verification of H3-${\bf C_p}$ when $N=2$.}
When $N=2$, we have $E_{1}^{+} = E_{2}^{-}$, $E_{2}^{+} = E_{1}^{-}$ and $\mathcal{W}_{NE} = \{(x^1,x^2,y)\,|\, |x^1-x^2|\leq f^{-1}_N(y)\}\cup \{y=0\}$. We set $Q_1 = \{(x_1,x_2,y)\in \mathbb{R}^2\times \mathbb{R}_{+}\,|\, x_1-x_2 \ge 0\}$ and $Q_2 = \{(x_1,x_2,y)\in \mathbb{R}^2\times \mathbb{R}_{+}\,|\, x_2-x_1 >0\}$. In this case, $\mathcal{A}_1 = E_{1}^{+}$ and $\mathcal{A}_2 =E_{2}^{+}$. When $(\pmb{x},y)\in \mathcal{A}_1$, there are two possibilities: either $(\pmb{x},y)\in \mathcal{A}_1\cap E_{1,1}^{+}$ or $(\pmb{x},y)\in \mathcal{A}_1\cap E_{1,2}^{+}$. If $(\pmb{x},y)\in \mathcal{A}_1\cap E_{1,1}^{+}$, then $\pmb{q} = (x^2 + x^1_{+},x^2,f(x^1_{+}))$ with $x_{+}^1$ the unique positive root such that $z-f_N(z) = x^1-x^2 -y$. Then it is easy to check that $\pmb{q}\in \partial \mathcal{W}_{NE}$. To see this, $(x^2 + x^1_{+} )- x^2 = x^1_{+}  = f_N^{-1}(f_N(x_{+}^1)) $. If  $(\pmb{x},y)\in \mathcal{A}_1\cap E_{1,2}^{+}$, then $\pmb{q} = (x^1 -y,x^2,0)$. Then $\pmb{q}\in \partial \mathcal{W}_{NE}$ since $y=0$. Similar analysis holds for $(\pmb{x},y)\in \mathcal{A}_2$ by symmetry.

\end{document}